\PassOptionsToPackage{unicode}{hyperref}
\PassOptionsToPackage{hyphens}{url}
\PassOptionsToPackage{dvipsnames,svgnames,x11names}{xcolor}
\documentclass[
  12pt]{article}

\usepackage{amsmath,amssymb}
\usepackage{iftex}
\ifPDFTeX
  \usepackage[T1]{fontenc}
  \usepackage[utf8]{inputenc}
  \usepackage{textcomp} 
\else 
  \usepackage{unicode-math}
  \defaultfontfeatures{Scale=MatchLowercase}
  \defaultfontfeatures[\rmfamily]{Ligatures=TeX,Scale=1}
\fi
\usepackage{lmodern}
\ifPDFTeX
  \usepackage{bm}
\else
  \newcommand{\bm}[1]{\symbfit{#1}}
\fi
\ifPDFTeX\else  
\fi
\IfFileExists{upquote.sty}{\usepackage{upquote}}{}
\IfFileExists{microtype.sty}{
  \usepackage[]{microtype}
  \UseMicrotypeSet[protrusion]{basicmath} 
}{}
\makeatletter
\@ifundefined{KOMAClassName}{
  \IfFileExists{parskip.sty}{%
    \usepackage{parskip}
  }{
    \setlength{\parindent}{0pt}
    \setlength{\parskip}{6pt plus 2pt minus 1pt}}
}{
  \KOMAoptions{parskip=half}}
\makeatother
\usepackage{xcolor}
\makeatletter
\ifx\paragraph\undefined\else
  \let\oldparagraph\paragraph
  \renewcommand{\paragraph}{
    \@ifstar
      \xxxParagraphStar
      \xxxParagraphNoStar
  }
  \newcommand{\xxxParagraphStar}[1]{\oldparagraph*{#1}\mbox{}}
  \newcommand{\xxxParagraphNoStar}[1]{\oldparagraph{#1}\mbox{}}
\fi
\ifx\subparagraph\undefined\else
  \let\oldsubparagraph\subparagraph
  \renewcommand{\subparagraph}{
    \@ifstar
      \xxxSubParagraphStar
      \xxxSubParagraphNoStar
  }
  \newcommand{\xxxSubParagraphStar}[1]{\oldsubparagraph*{#1}\mbox{}}
  \newcommand{\xxxSubParagraphNoStar}[1]{\oldsubparagraph{#1}\mbox{}}
\fi
\makeatother

\usepackage{longtable,booktabs,array}
\usepackage{calc} 
\usepackage{etoolbox}
\makeatletter
\patchcmd\longtable{\par}{\if@noskipsec\mbox{}\fi\par}{}{}
\makeatother
\IfFileExists{footnotehyper.sty}{\usepackage{footnotehyper}}{\usepackage{footnote}}
\makesavenoteenv{longtable}
\usepackage{graphicx}
\makeatletter
\def\maxwidth{\ifdim\Gin@nat@width>\linewidth\linewidth\else\Gin@nat@width\fi}
\def\maxheight{\ifdim\Gin@nat@height>\textheight\textheight\else\Gin@nat@height\fi}
\makeatother
\setkeys{Gin}{width=\maxwidth,height=\maxheight,keepaspectratio}
\makeatletter
\def\fps@figure{htbp}
\makeatother

\makeatletter
\@ifpackageloaded{caption}{}{\usepackage{caption}}
\AtBeginDocument{%
\ifdefined\contentsname
  \renewcommand*\contentsname{Table of contents}
\else
  \newcommand\contentsname{Table of contents}
\fi
\ifdefined\listfigurename
  \renewcommand*\listfigurename{List of Figures}
\else
  \newcommand\listfigurename{List of Figures}
\fi
\ifdefined\listtablename
  \renewcommand*\listtablename{List of Tables}
\else
  \newcommand\listtablename{List of Tables}
\fi
\ifdefined\figurename
  \renewcommand*\figurename{Figure}
\else
  \newcommand\figurename{Figure}
\fi
\ifdefined\tablename
  \renewcommand*\tablename{Table}
\else
  \newcommand\tablename{Table}
\fi
}
\@ifpackageloaded{float}{}{\usepackage{float}}
\floatstyle{ruled}
\@ifundefined{c@chapter}{\newfloat{codelisting}{h}{lop}}{\newfloat{codelisting}{h}{lop}[chapter]}
\floatname{codelisting}{Listing}

\makeatother
\makeatletter
\@ifpackageloaded{caption}{}{\usepackage{caption}}
\@ifpackageloaded{subcaption}{}{\usepackage{subcaption}}
\makeatother

\ifLuaTeX
  \usepackage{selnolig}  
\fi
\usepackage[]{natbib}
\usepackage{bookmark}

\IfFileExists{xurl.sty}{\usepackage{xurl}}{} 
\hypersetup{
  pdftitle={AgentCue: Statistical Summaries for Guiding Decentralized Multi-Agent Self-Organization},
  pdfauthor={},
  pdfkeywords={},
  colorlinks=true,
  linkcolor={blue},
  filecolor={Maroon},
  citecolor={Blue},
  urlcolor={Blue},
  pdfcreator={LaTeX via pandoc}}

\usepackage{tabularx,placeins}
\makeatletter
\renewcommand\section{\@startsection{section}{1}{\z@}%
  {-2ex \@plus -.5ex \@minus -.2ex}{.8ex \@plus .2ex}%
  {\normalfont\Large\bfseries}}
\renewcommand\subsection{\@startsection{subsection}{2}{\z@}%
  {-1.5ex \@plus -.4ex \@minus -.2ex}{.6ex \@plus .2ex}%
  {\normalfont\large\bfseries}}
\makeatother
\usepackage{tcolorbox}
\tcbuselibrary{breakable,skins}
\AtBeginEnvironment{figure}{\spacingset{1}\setkeys{Gin}{height=.58\textheight,keepaspectratio}}

\date{}
\providecommand{\harvardurl}[1]{}
\renewcommand{\harvardurl}[1]{\url{#1}}
\newcommand{\manuscripttitle}{AgentCue: Statistical Summaries for Guiding Decentralized Multi-Agent Self-Organization}
\newcommand{\keywords}{}

\newcommand{\anon}{1}

\usepackage{amsthm}
\newtheorem{theorem}{Theorem}
\newtheorem{remark}{Remark}
\newtheorem{proposition}{Proposition}
\newtheorem{corollary}{Corollary}
\newtheorem{lemma}{Lemma}

\theoremstyle{definition}
\newtheorem{assumption}{Assumption}
\newcommand{\E}{\mathbb E}
\newcommand{\PP}{\mathbb P}
\newcommand{\RR}{\mathbb R}
\newcommand{\FF}{\mathcal F}
\newcommand{\one}{\bm{1}}
\newcommand{\rank}{\operatorname{rank}}
\newcommand{\op}{\mathrm{op}}
\newcommand{\fro}{\mathrm F}
\newcommand{\norm}[1]{\lVert #1\rVert}
\newcommand{\dist}{\operatorname{dist}_{n,m}}

\begin{document}
\raggedbottom

\def\spacingset#1{\renewcommand{\baselinestretch}%
{#1}\small\normalsize} \spacingset{1}


\if1\anon
{
  \hypersetup{pdfauthor={Guojun Zhu; Xinyou Chen; Kejian Zhang; Sanguo Zhang; Zhi Zhang; Doudou Zhou}}
  \title{\bf \manuscripttitle}
\author{Guojun Zhu\textsuperscript{1,2}, Xinyou Chen\textsuperscript{3}, Kejian Zhang\textsuperscript{2},\\
Sanguo Zhang\textsuperscript{1}, Zhi Zhang\textsuperscript{4}, Doudou Zhou\textsuperscript{2,*}\\
\small \textsuperscript{1}School of Mathematical Sciences, University of Chinese Academy of Sciences\\
\small \textsuperscript{2}Department of Statistics and Data Science, National University of Singapore\\
\small \textsuperscript{3}Department of Statistics and Data Science, Tsinghua University\\
\small \textsuperscript{4}Department of Statistics and Data Science, University of California, Los Angeles\\[0.3em]
\small \textsuperscript{*}Corresponding author: ddzhou@nus.edu.sg}
  \maketitle
} \fi

\if0\anon
{
  \hypersetup{pdfauthor={}}
  \bigskip
  \bigskip
  \bigskip
  \begin{center}
    {\LARGE\bf \manuscripttitle}
\end{center}
  \medskip
} \fi

\bigskip
\begin{abstract}
Multi-agent systems (MAS) combine multiple agents to solve complex tasks, and recent MAS scale to more than a thousand agents. At this scale, decentralized MAS remove the central orchestrator that assigns sub-tasks, and the agents, which we call workers, self-organize through their own choices of sub-tasks and collaborators on a shared board. As workers keep adding participation and collaboration records, however, the information explodes: workers spend more and more effort on reading the records of other workers, and self-organization loses efficiency. In the statistical view, we represent the execution history of an MAS by two continually updated networks, a collaboration network among workers and a participation network between workers and task types. We then propose AgentCue, a method built on a new latent space model that places workers and task types in a shared latent space. At each checkpoint, that is, after a fixed number of sessions of the workers, the estimated parameters of the latent space model form a compact and interpretable statistical summary, and workers use the statistical summary to choose sub-tasks and collaborators without repeated searches of the shared board. We also provide theoretical results on model identifiability and estimation accuracy for continually updated networks. 
Experiments on three benchmarks for cross-document reasoning, mathematical theorem proving, and software patch generation show that AgentCue improves the efficiency, accuracy, and interpretability of self-organization.
\end{abstract}

\noindent%
\vfill

\newpage
\spacingset{1.8} 

\section{Introduction}
\label{sec:introduction}

Over the past few years, foundation models, that is, large pretrained models such as large language models, have advanced rapidly, and so has the software around these models, which agent developers call a harness. A harness runs a foundation model in a loop: the harness maintains the context of the foundation model, executes the tools that the foundation model calls, and feeds the results back to the foundation model. A foundation model together with a harness forms an agent, and agents are now in wide use \citep{yao2023react}. 
On complex tasks, however, a single agent can only explore step by step: the agent must first decompose the task and then spend a long time on sub-tasks such as coding, inference, and review, one after another \citep{zhan2026agensh}. Moreover, on a complex task, a single agent may even lose focus on local details and become highly prone to errors, and the decomposition of the task may itself be flawed \citep{yang2026programbench,huang2026deepswe}. 
Therefore, agent developers equip multiple agents with an auxiliary system for collaboration, and the result is a multi-agent system (MAS): several agents work on different sub-tasks at the same time, and one or more agents then integrate their results. Such collaboration greatly improves performance and saves time \citep{zhan2026agensh}, and MAS have become a mainstream research paradigm for complex tasks \citep{wu2024autogen,li2026beyondmoreagents}.

To meet increasingly complex real-world tasks, MAS continue to grow in scale. Well-known MAS such as CAMEL, ChatDev, and LAMBDA use only two to seven agents for role-playing task solving, software development, and data analysis, respectively \citep{li2023camel,qian2024chatdev,sun2026lambda}, whereas the reconstruction of a complex codebase within a limited time can now call for as many as 1,024 agents \citep{zhan2026agensh}. 
However, most current MAS rely on a central orchestrator, that is, a main agent that decomposes the task and assigns sub-tasks to the other agents, which we call workers. The capacity of the central orchestrator to allocate sub-tasks and coordinate workers often limits the scalability of the whole MAS \citep{zhan2026agensh}. 
Therefore, decentralized MAS have emerged as a core research paradigm for large-scale MAS \citep{yang2025agentnet,li2026beyondmoreagents}. Decentralized MAS replace the central orchestrator with environment components such as a shared workspace for sub-tasks in progress and finished sub-tasks, shared context for reusable findings, a message interface for communication, and a task queue for open sub-tasks \citep{mao2026delm,zhan2026agensh}, and we refer to these components collectively as the shared board. 
Decentralized MAS raise many research questions, and the most important question asks how a decentralized MAS achieves self-organization through the shared board. Here, self-organization means that the division of labor and collaboration arise from the local choices of workers instead of central assignment, where a local choice is a choice of a sub-task or a collaborator that a worker makes on its own. Studies such as DeLM \citep{mao2026delm}, Proxifield \citep{tambwekar2026proxifield}, and recent analyses of agent systems \citep{kim2025scaling,li2026beyondmoreagents} agree that self-organization without any external design suffers from low efficiency, redundant work, and error propagation, where an external design refers to a mechanism that agent developers add to shape the interactions among workers.
Figure~\ref{fig:motivation} illustrates how the difficulties of self-organization grow with the number of workers. There is therefore an urgent need for a suitable external design that helps workers choose their next sub-tasks and collaborators efficiently.

\begin{figure}[!ht]
    \centering
    \includegraphics[width=\textwidth]{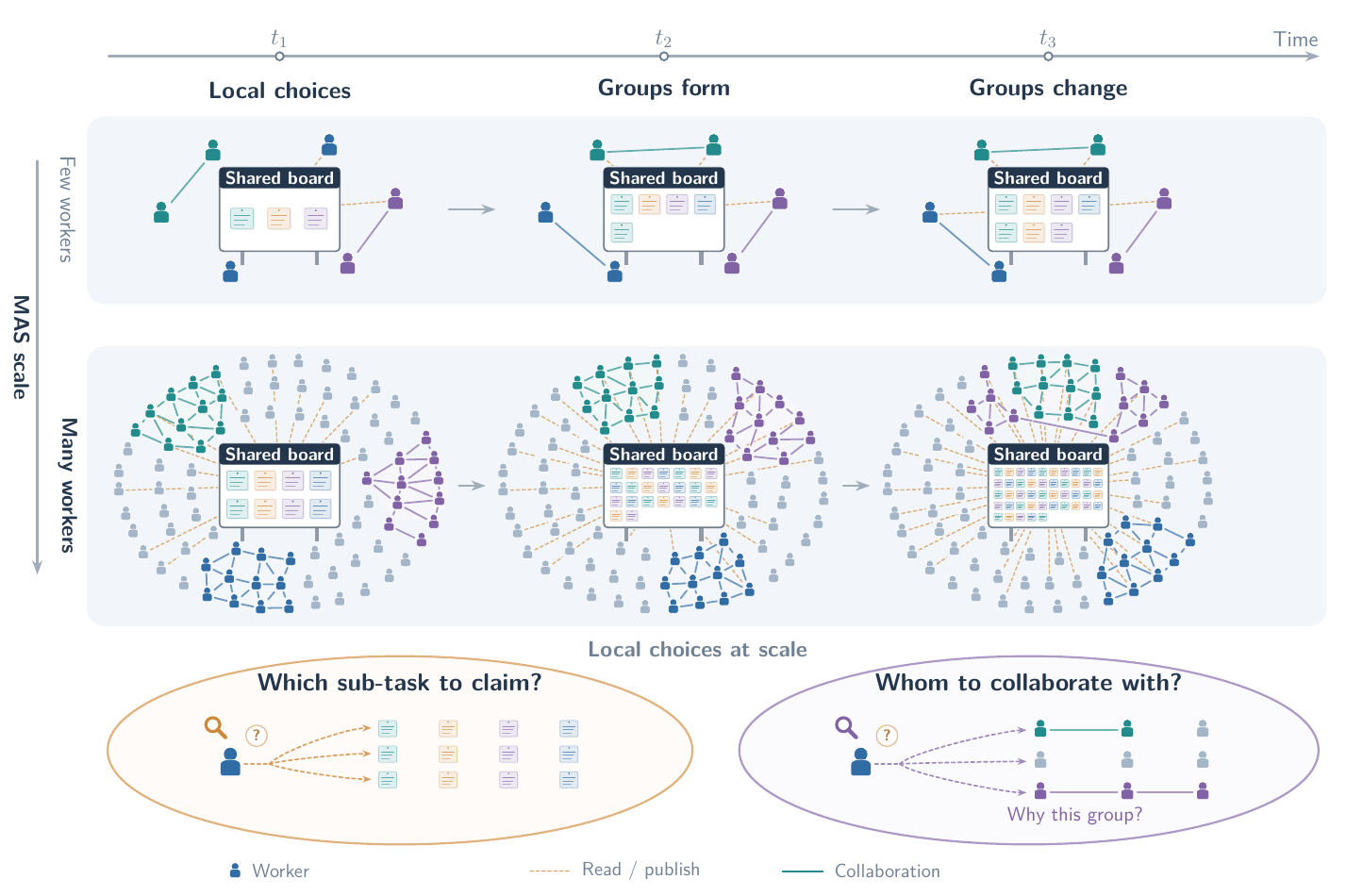}
    \caption{Self-organization in a decentralized MAS. Workers read from and publish to the shared board (dashed lines) and collaborate with one another (solid lines), and groups of collaborators form and change over time through local choices. With few workers (upper panel), the shared board holds few records. With many workers (middle panel), records accumulate quickly, and each local choice (lower panel), which sub-task to claim and whom to collaborate with, requires a search through many records.}
    \label{fig:motivation}
\end{figure}

Several recent studies introduce such external designs. AgentNet \citep{yang2025agentnet} equips each agent with a router, which decides whether the agent executes a task, splits the task, or forwards the task to the agent whose capability best matches the task. Decentralized Expertise-Aware Load Serving (DEALS) \citep{huang2026deals} lets each agent forward a task to a neighbor with fewer unfinished tasks or a higher success rate on tasks of the same category. Economy of Minds (EoM) \citep{qi2026economy} lets agents bid in auctions for the right to act and pay one another, and the wealth of each agent then decides whether the agent persists and mutates or goes bankrupt. Proxifield \citep{tambwekar2026proxifield} builds a sparse communication graph in every round through deterministic routing rules. Self-Organizing Agent Teams (SAT) \citep{pappu2026sat} learns a bank of collaboration strategies offline, which specify roles, conversation phases, participation, and information flow, and applies the bank unchanged to new problems. Each design presets a mechanism of interaction, within which the agents still make their own local choices.

However, AgentNet, DEALS, EoM, Proxifield, and SAT were developed and evaluated in small-scale MAS, and a new problem takes the stage in large-scale MAS. The new problem lies in the participation and collaboration records: participation records state the task types, that is, pre-specified categories of sub-tasks, in which each worker works, and collaboration records state whose published results, that is, results that workers publish on the shared board, each worker uses or reviews. In a large-scale MAS, many workers keep adding participation and collaboration records, and the information explodes: workers spend more and more effort on reading the records of other workers and choose their sub-tasks more and more slowly, so the efficiency of self-organization declines. 
Agensh \citep{zhan2026agensh} shows the burden at scale: with 1,024 workers, every new entry in the shared context goes to all other workers, while a direct read of the shared context returns only the 2,000 latest entries. The more entries the workers write, the more time a worker must spend on keyword searches to trace entries beyond the 2,000 latest entries. The gain from more workers also shrinks: in the codebase reconstruction of Agensh, the final test-pass rate rises by 17.05\% from 1 to 128 workers but by only 4.12\% from 128 to 1,024 workers.
The resulting dilemma, an information explosion without a fast and interpretable summary, is exactly where statistics excels. Picture the workers as a large crowd of ``cowboys'': each cowboy rides a foundation model as a horse and steers the horse through a harness, while a vast list of sub-tasks and potential collaborators awaits selection. What the cowboys need is a statistical estimation procedure\footnote{Coincidentally, one of the most influential estimation procedures in statistics shares its name with the cowboy's rope: the {\it Lasso} \citep{tibshirani1996lasso}, which shrinks and selects among many candidate variables.} that summarizes the information. With such a procedure, the cowboys no longer depend on repeated searches, that is, retrievals of the same unchanged records of the shared board by many workers: the cowboys quickly screen the sub-tasks and potential collaborators, select the sub-tasks and collaborators that fit them best, and reach an efficient self-organization.

In the statistical view, we first characterize the data. A preprocessing step condenses the massive execution history of an MAS into two networks. The first network is a collaboration network, whose nodes are workers and whose edges count the collaborations between two workers, that is, the times one worker uses or reviews a published result of the other worker. The second network is a participation network, a worker--task-type bipartite network whose nodes are the workers and the task types and whose edges count the times each worker performs each task type. 
Latent space models offer a powerful tool for such network data \citep{hoff2002latent,ma2020universal}. A latent space model represents each node by a latent position, a low-dimensional vector to be estimated, and nodes with closer latent positions tend to have closer relations. A rich literature has extended latent space models to various network data. 
For example, \citet{he2025semiparametric} analyze longitudinal networks of interaction counts with timestamps, such as bike trips between stations, and \citet{fan2024generalized} jointly analyze a one-mode network, such as friendships among users, and a two-mode network, such as reviews from users to businesses. 
For a new form of data, which we call continually updated networks because the edge counts accumulate as the MAS runs, we propose a new latent space model, to our knowledge the first latent space model for such networks, and build our method {\bf AgentCue} on the latent space model. In contrast with the longitudinal networks of \citet{he2025semiparametric}, continually updated networks are estimated while the networks grow, and the estimates change the choices that generate later edges. In AgentCue, at each checkpoint, that is, after every fixed number of sessions, where a session is one activation of a worker, the MAS estimates the parameters of the latent space model and publishes the estimates as a statistical summary. Workers then locate suitable collaborators and task types in the latent space, and the next stage of self-organization emerges in an orderly and efficient manner.

Our work makes three main contributions. 
(i) For the problem, we identify a new problem for agents: in large-scale decentralized MAS, the information explosion of participation and collaboration records slows down self-organization, and the problem also provides a stage on which statistical methods can be built. 
(ii) For methodology, we propose AgentCue, which combines a new latent space model for continually updated networks with the MAS: the MAS generates the networks, and the statistical summaries of the latent space model guide the MAS. 
Supplementary~\ref{app:theory} establishes model identifiability and gives statistical and optimization error bounds under stated assumptions.
(iii) For experiments, we evaluate AgentCue on the Multi-Doc QA subset of LongBench v2 \citep{bai2024longbenchv2}, FATE-H for formal algebraic theorem proving \citep{jiang2026fate}, and 10 software engineering tasks from DeepSWE \citep{huang2026deepswev11}. The results show that AgentCue improves the efficiency, accuracy, and interpretability of self-organization. For interpretability, the Choice Records, which store the reason for each choice of a worker, illustrate how workers combine statistical summaries with task requirements.

\section{Background and Related Work}

Large language models support language understanding, reasoning, and code generation. With tools and environmental feedback, agents can plan and act \citep{yao2023react}. MAS let agents divide work and exchange information \citep{wu2024autogen}. Recent decentralized MAS extend collaboration to many workers, who choose sub-tasks and coordinate through a shared board \citep{zhan2026agensh}.

\noindent{\bf Multi-Agent Collaboration.}
Several systems organize collaboration through role assignment and structured communication. {\it CAMEL} pairs a user agent with an assistant agent to solve tasks through role-playing \citep{li2023camel}, while {\it AutoGen} supports conversations among agents with configurable roles and tools \citep{wu2024autogen}. {\it ChatDev} and {\it MetaGPT} organize software development around specialized roles and workflows \citep{qian2024chatdev,hong2024metagpt}. {\it AgentCoder} separates programming, test design, and test execution across three agents \citep{huang2023agentcoder}. For data analysis, {\it LAMBDA} pairs a programmer with an inspector for code correction \citep{sun2026lambda}. AgentCoder and LAMBDA illustrate collaboration in single-digit teams. Our concern is how large-scale MAS can use accumulated participation and collaboration records to guide subsequent local choices.

\noindent{\bf Decentralized Multi-Agent Systems.}
Decentralized MAS allow workers to choose sub-tasks and collaborators through shared context and communication \citep{li2026beyondmoreagents}. {\it DeLM} combines an asynchronous task queue with verified shared context, so that workers claim sub-tasks and build on earlier findings \citep{mao2026delm}. {\it AutoScientists} lets workers form teams around hypotheses and share successes and failures \citep{gao2026autoscientists}. {\it Agensh}, the closest setting to ours, scales to 1,024 workers in the reconstruction of the pandoc codebase \citep{zhan2026agensh}. The workers of Agensh claim, execute, verify, and merge work through a shared workspace; reported execution logs illustrate reviewer selection from prior experience. As the number of workers grows and participation and collaboration records accumulate, workers face more potential collaborators and more records to search. The statistical challenge is to jointly summarize participation across task types and collaboration through the use or review of published results, while accounting for overall activity.

\noindent{\bf Self-Organization Learning.}
Several approaches adapt collaboration through experience. {\it MorphAgent} updates textual profiles \citep{lu2024morphagent}, while {\it SAT} learns strategies for roles, participation, and information flow \citep{pappu2026sat}. Auctions and accumulated payments govern the allocation of the right to act, agent mutation, and replacement in {\it EoM} \citep{qi2026economy}; {\it AgentLance} adapts bidding strategies while an allocator selects workers using bids and reputation \citep{liu2026markets}. {\it TheBotCompany} uses managers to hire, assign, and retire workers \citep{lyu2026thebotcompany}; {\it LATTE} lets workers propose task-graph changes coordinated by a lead agent \citep{mieczkowski2026latte}. In robotics, {\it SwarmCF} learns robot and task vectors from outcomes for decentralized allocation \citep{apartsin2026swarmcf}. We study the cumulative participation and collaboration records that the local choices of workers produce.

\noindent{\bf Graph-Guided Multi-Agent Collaboration.}
Graph-based methods select participants or communication links. {\it DyLAN} uses contribution scores to select participants \citep{liu2024dylan}, {\it G-Designer} learns task-dependent communication graphs \citep{zhang2025gdesigner}, and {\it SelfOrg} rebuilds graphs from evolving answers \citep{tastan2025selforg}. {\it AgentNet} combines local routing and retrieval-based memory with adaptive connections \citep{yang2025agentnet}, and {\it DEALS} routes tasks between neighbors by weighing unfinished tasks against success rates \citep{huang2026deals}; {\it MeshHeal} uses peer verification and routing adjustments to detect degradation and reintegrate recovered agents \citep{chen2026meshheal}. {\it DyTopo} and {\it Proxifield} construct communication links through semantic matching \citep{luo2026dytopo,tambwekar2026proxifield}. In our setting, the collaboration network records the use and review of published results.

\noindent{\bf Statistical Studies of Human Collaboration and Latent Space Models.}
Statistical models of human collaboration analyze collaboration records. {\it Relational event models} connect interaction rates to past activity, and {\it relational hyperevent models} describe group interactions such as coauthorship \citep{butts2008relational,perry2013point,lerner2023coauthor}. Latent space models describe network relationships through low-dimensional latent positions, with methods for large networks \citep{ma2020universal}, temporal variation, multiple layers, and hierarchical node categories \citep{durante2014dynamic,loyal2023eigenmodel,pu2026tree}. \citet{fan2024generalized} jointly model one-mode and two-mode networks through shared latent positions and network-specific node effects, that is, node-level parameters for overall activity. We apply the construction of \citet{fan2024generalized} to the collaboration network and the participation network and account for overall activity in each network. We return the statistical summaries to workers to inform their subsequent local choices.

\section{Methodology}
\subsection{Network Data Collection}
\label{sec:Network-Data-Collection}

We collect network data from the execution history of an Agensh-based MAS \citep{zhan2026agensh}. In an Agensh-based MAS, no central orchestrator assigns sub-tasks, and workers pursue a common goal by choosing their own sub-tasks and collaborators. The shared board makes sub-tasks in progress, published results, and accumulated findings available to all workers. Workers use the shared board to identify useful work, publish their results, and continue from their own results or the results of other workers. Self-organization therefore leaves records of who performs each piece of work and how published results are used.

A piece of published work, which we simply call a \emph{work}, is the basic observation unit. The record of a work, which we simply call a \emph{work record}, identifies the work, the worker who produced the work, and the exact versions of the published results that the work uses or reviews. The work records thus contain the participation and collaboration records of the MAS. Preprocessing removes duplicate work records and assigns task types to each work according to the published results of the work. For network analysis, we pre-specify task types as categories of sub-tasks for the common goal. For example, in cross-document reasoning, we define four task types: \emph{extract}, \emph{compare}, \emph{inference}, and \emph{review}. A work can belong to multiple task types when the work delivers several independently reusable results.

We index the resulting work records by $r$ in publication order, the $n$ workers by $i,\ell\in[n]$, and the $m$ task types by $k\in[m]$, where $[n]=\{1,\ldots,n\}$. Each work record $r$ produces network increments $\Delta\bm{A}_r\in\{0,1\}^{n\times n}$ and $\Delta\bm{B}_r\in\{0,1\}^{n\times m}$. The entry $(\Delta\bm{A}_r)_{i\ell}$ is one when the work of work record $r$ is produced by worker $i$ and uses or reviews a published result of worker $\ell$, or vice versa. We aggregate collaboration over worker--worker pairs, so $\Delta\bm{A}_r$ is symmetric with zero diagonal. Each work contributes at most once to a worker--worker pair, even if the work uses several results of the same worker. The entry $(\Delta\bm{B}_r)_{ik}$ is one when the work of work record $r$ is produced by worker $i$ and  belongs to task type $k$. Multiple collaborators or task types can therefore update several entries, while each work contributes at most once to each worker--task-type pair. $\Delta\bm{A}_r$ and $\Delta\bm{B}_r$ record participation and collaboration in published work, whether or not the work is correct.

A worker acts in sessions, that is, separate activations of the worker. A session is invalid when a technical failure stops the session and valid otherwise. A checkpoint occurs after every fixed number of valid sessions of the MAS, 32 in our experiments, instead of after a fixed elapsed time, because the elapsed time of a session differs greatly across the host machines that run the MAS. At checkpoint $t\in[T]$, where $T$ is the number of checkpoints, let $S_t$ be the number of work records published up to checkpoint $t$. The cumulative network observations are:
\begin{equation*}
    \bm{A}^{(t)}=\sum_{r=1}^{S_t}\Delta\bm{A}_r,\qquad\bm{B}^{(t)}=\sum_{r=1}^{S_t}\Delta\bm{B}_r.
\end{equation*}
The entries $A_{i\ell}^{(t)}$ and $B_{ik}^{(t)}$ count collaboration and participation up to checkpoint $t$, and we retain all worker--worker and worker--task-type pairs, including those with zero counts. As illustrated in Figure~\ref{fig:agentcue-framework}, $\bm{A}^{(t)}$ and $\bm{B}^{(t)}$ provide complementary observations on the same workers: $\bm{A}^{(t)}$ describes whose results support the work of other workers, and $\bm{B}^{(t)}$ describes participation across task types. We jointly model these networks below.

\subsection{Latent Space Model and Estimation}
\label{sec:Joint-Latent-Space-Model-and-Estimation}

We connect the two networks through shared worker latent positions, as in the one-mode and two-mode construction of \citet{fan2024generalized}. Let $\bm{z}_i^{(t)},\bm{w}_k^{(t)}\in\mathbb{R}^d$ denote worker and task-type latent positions, where $d$ is the latent dimension. The node effects $\alpha_i^{(t)}$ and $\beta_i^{(t)}$ describe the overall collaboration and participation activity of worker $i$, and the node effect $\gamma_k^{(t)}$ describes the overall participation that task type $k$ attracts. The latent space model treats $A_{i\ell}^{(t)}$ and $B_{ik}^{(t)}$ as Poisson variables with cumulative means $\theta_{i\ell}^{A,(t)}>0$ and $\theta_{ik}^{B,(t)}>0$:
\begin{equation}
    \begin{aligned}
         A_{i\ell}^{(t)}&\sim\operatorname{Poisson}(\theta_{i\ell}^{A,(t)}),
         &\log\theta_{i\ell}^{A,(t)}&=\alpha_i^{(t)}+\alpha_\ell^{(t)}+(\bm{z}_i^{(t)})^{\top}\bm{z}_\ell^{(t)},\\
         B_{ik}^{(t)}&\sim\operatorname{Poisson}(\theta_{ik}^{B,(t)}),
         &\log\theta_{ik}^{B,(t)}&=\beta_i^{(t)}+\gamma_k^{(t)}+(\bm{z}_i^{(t)})^{\top}\bm{w}_k^{(t)}.
        \end{aligned}
    \label{eq:joint}
\end{equation}
Both $\bm{A}^{(t)}$ and $\bm{B}^{(t)}$ inform the shared $\bm{z}_i^{(t)}$. All parameters vary with $t$ and describe cumulative experience at checkpoint $t$. Define the parameter tuple $\Theta^{(t)}=(\bm{\alpha}^{(t)},\bm{\beta}^{(t)},\bm{\gamma}^{(t)},\bm{Z}^{(t)},\bm{W}^{(t)})$, where $\bm{\alpha}^{(t)}=(\alpha_1^{(t)},\ldots,\alpha_n^{(t)})^\top$, $\bm{\beta}^{(t)}=(\beta_1^{(t)},\ldots,\beta_n^{(t)})^\top$, and $\bm{\gamma}^{(t)}=(\gamma_1^{(t)},\ldots,\gamma_m^{(t)})^\top$ are vectors, and $\bm{Z}^{(t)}=[\bm{z}_1^{(t)},\ldots,\bm{z}_n^{(t)}]^\top\in\mathbb{R}^{n\times d}$ and $\bm{W}^{(t)}=[\bm{w}_1^{(t)},\ldots,\bm{w}_m^{(t)}]^\top\in\mathbb{R}^{m\times d}$ contain the latent positions. The two Poisson negative log-likelihoods $\mathcal{L}_A^{(t)}$ and $\mathcal{L}_B^{(t)}$ omit terms independent of the parameters:
\begin{equation}
    \mathcal{L}_A^{(t)}(\Theta^{(t)})=\sum_{1\leq i<\ell\leq n}\left\{\theta_{i\ell}^{A,(t)}-A_{i\ell}^{(t)}\log\theta_{i\ell}^{A,(t)}\right\},\ \mathcal{L}_B^{(t)}(\Theta^{(t)})=\sum_{i=1}^{n}\sum_{k=1}^{m}\left\{\theta_{ik}^{B,(t)}-B_{ik}^{(t)}\log\theta_{ik}^{B,(t)}\right\}.
    \label{eq:network-likelihoods}
\end{equation}
The joint objective $\mathcal{L}^{(t)}$ is their sum:
\begin{equation}
    \mathcal{L}^{(t)}(\Theta^{(t)})=\mathcal{L}_A^{(t)}(\Theta^{(t)})+\mathcal{L}_B^{(t)}(\Theta^{(t)}). \label{eq:main-likelihood}
\end{equation}
Let $\bm{1}_n$ and $\bm{1}_m$ denote vectors of ones. We minimize \eqref{eq:main-likelihood} subject to $\bm{1}_{n}^{\top}\bm{Z}^{(t)}=\bm{0}^{\top}$, $\bm{1}_{m}^{\top}\bm{W}^{(t)}=\bm{0}^{\top}$, and $\bm{1}_{n}^{\top}\bm{\beta}^{(t)}=0$. Under these centering constraints and full column rank of $\bm{Z}^{(t)}$, the matrices $(\log\theta_{i\ell}^{A,(t)})$ and $(\log\theta_{ik}^{B,(t)})$ identify $\bm{\alpha}^{(t)},\bm{\beta}^{(t)},\bm{\gamma}^{(t)}$ uniquely and $\bm{Z}^{(t)},\bm{W}^{(t)}$ up to a common orthogonal transformation. Proposition~\ref{prop:identifiability} in Supplementary~\ref{sec:identifiability} gives the result and proof.

For numerical optimization, we follow the projected-gradient implementation of \citet{he2025semiparametric} and add a small ridge penalty on $\bm{Z}^{(t)}$ and $\bm{W}^{(t)}$. Supplementary~\ref{app:fit} gives the penalized objective, the numerical scaling, and the stopping rule. Substituting the numerical estimate $\widehat{\Theta}^{(t)}$, that is, the output of the optimization at checkpoint $t$, into \eqref{eq:joint} gives fitted means $\widehat{\theta}_{i\ell}^{A,(t)}$ and $\widehat{\theta}_{ik}^{B,(t)}$. At the first checkpoint, we initialize $\bm{Z}^{(t)}$ and $\bm{W}^{(t)}$ with independent standard normal entries and $\bm{\alpha}^{(t)},\bm{\beta}^{(t)},\bm{\gamma}^{(t)}$ from $\bm{A}^{(t)}$ and $\bm{B}^{(t)}$, then apply the centering constraints. Each later checkpoint $t$ starts the optimization from the preceding estimate $\widehat{\Theta}^{(t-1)}$.

\subsection{Guiding Self-Organization with the Statistical Summary}
\label{sec:Summary-Guiding-in-Self-Organization}

In Agensh, workers search the work records on the shared board before they choose sub-tasks and collaborators \citep{zhan2026agensh}. Each session of a worker has a fixed budget of board requests, that is, requests to the shared board, and every read of the work records of another worker costs one board request. Many workers also retrieve the same unchanged work records to compare the same sub-tasks and collaborators, a pattern that Section~\ref{sec:introduction} calls repeated searches. In AgentCue, the MAS instead estimates, once per checkpoint, the parameters of the latent space model \eqref{eq:joint} from the work records on behalf of all workers and publishes the estimates $\widehat{\Theta}^{(t)}=(\widehat{\bm\alpha}^{(t)},\widehat{\bm\beta}^{(t)},\widehat{\bm\gamma}^{(t)},\widehat{\bm Z}^{(t)},\widehat{\bm W}^{(t)})$, which we call the \emph{statistical summary} at checkpoint $t$, together with a summary version that identifies the network observations behind the estimates. When a session first requests the statistical summary, the worker reads the latest statistical summary with one board request. Every later \emph{local comparison}, that is, a calculation of fitted means in \eqref{eq:joint} from the statistical summary that the worker holds, costs no board request. Figure~\ref{fig:agentcue-framework} gives an overview of AgentCue.

\begin{figure}[!ht]
    \centering
    \includegraphics[width=\textwidth]{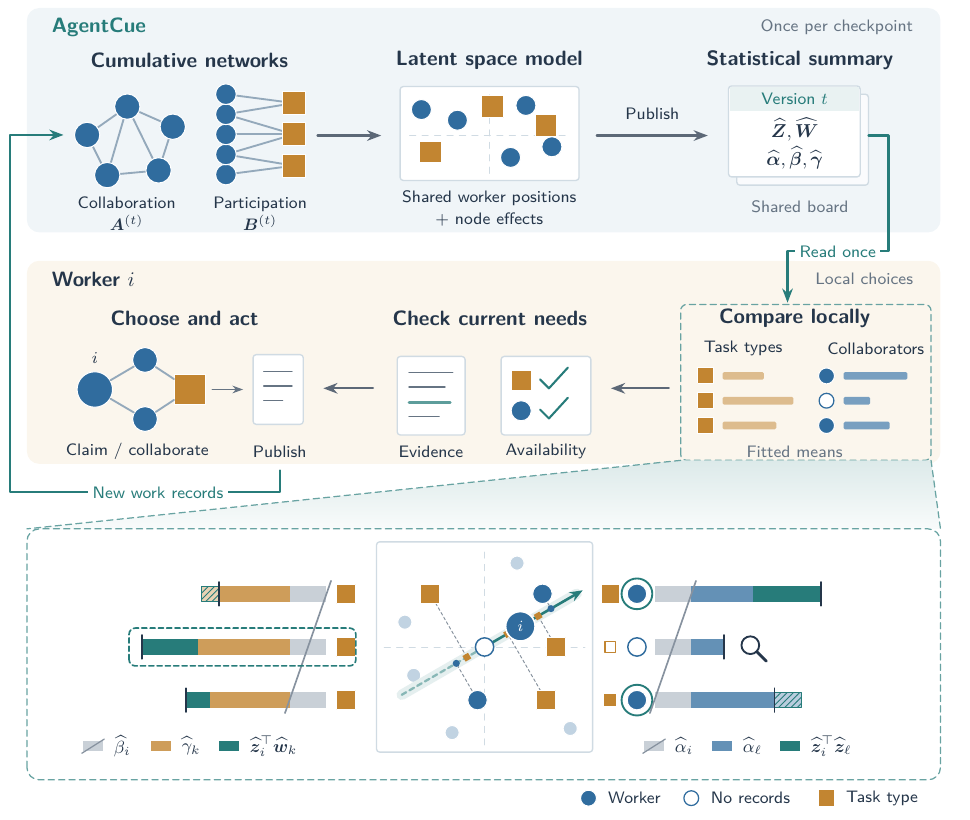}
    \caption{{\bf Overview of AgentCue.} Upper and middle bands: the MAS publishes the statistical summary once per checkpoint, and worker $i$ reads the statistical summary once. Lower band: local comparisons, enlarged. Left and right bars: $\log\widehat{\theta}_{ik}^{B}$ and $\log\widehat{\theta}_{i\ell}^{A}$, split as in \eqref{eq:fitted-means}; struck-through $\widehat{\beta}_i$ and $\widehat{\alpha}_i$ cancel in \eqref{eq:comparisons}; hatched segments: negative affinity terms. Map: projections of $\widehat{\bm w}_k$ and $\widehat{\bm z}_\ell$ onto the line of sight of worker $i$. Small squares: $\widehat{\theta}_{\ell k_{\mathrm{in}}}^{B}$, with $k_{\mathrm{in}}$ the task type of the inputs that the sub-task needs. Dashed box and rings: shortlist; magnifier: work-record search. Superscripts $(t)$ are omitted.}
    \label{fig:agentcue-framework}
\end{figure}

Each fitted mean comes with an explanation, because the logarithm of each fitted mean in \eqref{eq:joint} splits into two kinds of terms:
\begin{equation}
    \log\widehat{\theta}_{i\ell}^{A,(t)}=\underbrace{\widehat{\alpha}_i^{(t)}+\widehat{\alpha}_\ell^{(t)}}_{\text{activity}}+\underbrace{(\widehat{\bm z}_i^{(t)})^{\top}\widehat{\bm z}_\ell^{(t)}}_{\text{affinity}},\qquad\log\widehat{\theta}_{ik}^{B,(t)}=\underbrace{\widehat{\beta}_i^{(t)}+\widehat{\gamma}_k^{(t)}}_{\text{activity}}+\underbrace{(\widehat{\bm z}_i^{(t)})^{\top}\widehat{\bm w}_k^{(t)}}_{\text{affinity}}. \label{eq:fitted-means}
\end{equation}
We call $\widehat{\theta}_{i\ell}^{A,(t)}$ the \emph{fitted collaboration mean} between workers $i$ and $\ell$ and $\widehat{\theta}_{ik}^{B,(t)}$ the \emph{fitted participation mean} of worker $i$ in task type $k$. We call the node effects \emph{activity terms} and the inner products of latent positions \emph{affinity terms}. $(\widehat{\bm z}_i^{(t)})^{\top}\widehat{\bm z}_\ell^{(t)}$ measures whether workers $i$ and $\ell$ collaborate more often than $\widehat{\alpha}_i^{(t)}+\widehat{\alpha}_\ell^{(t)}$ implies, and $(\widehat{\bm z}_i^{(t)})^{\top}\widehat{\bm w}_k^{(t)}$ measures whether worker $i$ leans toward task type $k$. A large $A_{i\ell}^{(t)}$ alone cannot tell an active worker from a close collaborator, and the split in \eqref{eq:fitted-means} separates the two. If all $(\widehat{\bm z}_i^{(t)})^{\top}\widehat{\bm w}_k^{(t)}$ vanish, every worker spreads participation over the task types in the same proportions, so the division of labor appears in $(\widehat{\bm z}_i^{(t)})^{\top}\widehat{\bm w}_k^{(t)}$.

A local comparison of worker $i$ compares $\widehat{\theta}_{i\ell}^{A,(t)}$ and $\widehat{\theta}_{i\ell'}^{A,(t)}$ with two candidates $\ell$ and $\ell'$, that is, other workers whom worker $i$ considers as collaborators, or $\widehat{\theta}_{ik}^{B,(t)}$ and $\widehat{\theta}_{ik'}^{B,(t)}$. In the two ratios, $\widehat{\alpha}_i^{(t)}$ and $\widehat{\beta}_i^{(t)}$ cancel, respectively:
\begin{equation}
    \begin{aligned}
    \frac{\widehat{\theta}_{i\ell}^{A,(t)}}{\widehat{\theta}_{i\ell'}^{A,(t)}}&=\underbrace{\exp\bigl(\widehat{\alpha}_\ell^{(t)}-\widehat{\alpha}_{\ell'}^{(t)}\bigr)}_{\text{candidates}}\,\underbrace{\exp\bigl\{(\widehat{\bm z}_i^{(t)})^{\top}(\widehat{\bm z}_\ell^{(t)}-\widehat{\bm z}_{\ell'}^{(t)})\bigr\}}_{\text{pair}},\\
    \frac{\widehat{\theta}_{ik}^{B,(t)}}{\widehat{\theta}_{ik'}^{B,(t)}}&=\underbrace{\exp\bigl(\widehat{\gamma}_k^{(t)}-\widehat{\gamma}_{k'}^{(t)}\bigr)}_{\text{task types}}\,\underbrace{\exp\bigl\{(\widehat{\bm z}_i^{(t)})^{\top}(\widehat{\bm w}_k^{(t)}-\widehat{\bm w}_{k'}^{(t)})\bigr\}}_{\text{worker }i}.
    \end{aligned}
    \label{eq:comparisons}
\end{equation}
In each row of \eqref{eq:comparisons}, the first factor does not involve worker $i$, so every worker obtains the same value: $\exp(\widehat{\alpha}_\ell^{(t)}-\widehat{\alpha}_{\ell'}^{(t)})$ compares the overall collaboration activity of candidates $\ell$ and $\ell'$, and $\exp(\widehat{\gamma}_k^{(t)}-\widehat{\gamma}_{k'}^{(t)})$ compares the overall participation that task types $k$ and $k'$ attract. 
The second factor in each row involves $\widehat{\bm z}_i^{(t)}$ and therefore differs across workers. The second factor in the first row can separate even two equally active candidates ($\widehat{\alpha}_\ell^{(t)}=\widehat{\alpha}_{\ell'}^{(t)}$). Because the centering constraint $\bm{1}_n^{\top}\widehat{\bm Z}^{(t)}=\bm{0}^{\top}$ makes $\sum_{i=1}^{n}(\widehat{\bm z}_i^{(t)})^{\top}\widehat{\bm w}_k^{(t)}=0$ for every $k$, the average worker has zero affinity terms, and the second factor in the second row is close to one when worker $i$ resembles the average worker on task types $k$ and $k'$.

The affinity terms have a geometric reading. For any $\bm x\in\mathbb{R}^d$, the inner product $(\widehat{\bm z}_i^{(t)})^{\top}\bm x$ equals the length of $\widehat{\bm z}_i^{(t)}$ times the signed length of the projection of $\bm x$ onto the line through the origin along $\widehat{\bm z}_i^{(t)}$, which we call the \emph{line of sight} of worker $i$, as shown in Figure~\ref{fig:agentcue-framework}. The projections of $\widehat{\bm w}_k^{(t)}$ and $\widehat{\bm z}_\ell^{(t)}$ ($\ell\ne i$) onto the line of sight order the task types and the other workers by affinity with worker $i$, so $\widehat{\bm z}_i^{(t)}$ serves both the sub-task choice and the collaborator choice of worker $i$. Under the ridge penalty of Section~\ref{sec:Joint-Latent-Space-Model-and-Estimation}, a worker $i$ without work records has $\widehat{\bm z}_i^{(t)}$ near the origin, so the small $\widehat{\theta}_{i\ell}^{A,(t)}$ and $\widehat{\theta}_{ik}^{B,(t)}$ come from $\widehat{\alpha}_i^{(t)}$ and $\widehat{\beta}_i^{(t)}$ and reflect missing work records, not missing ability.

The lower band of Figure~\ref{fig:agentcue-framework} enlarges the local comparisons of worker $i$. To choose a sub-task, worker $i$ compares $\widehat{\theta}_{ik}^{B,(t)}$ across the task types $k$ of the open sub-tasks, as in the second row of \eqref{eq:comparisons}. To choose a collaborator for a sub-task that needs inputs of task type $k_{\mathrm{in}}$, for example $k_{\mathrm{in}}$ is \emph{extract} for a sub-task of task type \emph{inference}, worker $i$ compares candidates $\ell$ by $\widehat{\theta}_{\ell k_{\mathrm{in}}}^{B,(t)}$ and $\widehat{\theta}_{i\ell}^{A,(t)}$, the latter as in the first row of \eqref{eq:comparisons}. The local comparisons leave a shortlist of sub-tasks and candidates. Worker $i$ then spends board requests only to check the \emph{current needs}, which the statistical summary cannot settle: whether a sub-task remains open, whether a candidate is occupied, and whether a published result supplies the required evidence. The statistical summary answers questions about the past, and the shared board answers questions about the present.

The statistical summary proposes, and the worker chooses. When the statistical summary gives weak support, for example for a candidate without work records, when candidates tie, when a sub-task is new, or when the current needs contradict the statistical summary, worker $i$ searches only the work records that the choice requires and states a concrete reason for the choice. Each choice leaves a Choice Record, which stores the summary version, whether the statistical summary informed the choice, and the reason for the choice. The statistical summary guides the choices, the choices create new work records, and the new work records revise the statistical summary at the next checkpoint; along the loop, the Choice Records trace every choice to a summary version and a reason.

\subsection{Feedback of the Statistical Summary}
\label{sec:Feedback-of-the-Statistical-Summary}

In the usual analysis with latent space models, the network does not respond to the estimates. In AgentCue, the statistical summary guides the choices that produce the next work records, so later estimates partly describe the effect of earlier estimates. This section analyzes the feedback and asks what the statistical summary describes in the long run. We follow one worker $i$ through four steps of the loop in the participation network.

\noindent{\bf Step 1: the statistical summary reports past shares.} The statistical summary gives worker $i$ a \emph{fitted participation share} for each task type $k$:
\begin{equation}
    \widehat{\pi}_{ik}^{(t)}=\frac{\widehat{\theta}_{ik}^{B,(t)}}{\sum_{k'=1}^{m}\widehat{\theta}_{ik'}^{B,(t)}}=\frac{\exp\{\widehat{\gamma}_k^{(t)}+(\widehat{\bm z}_i^{(t)})^{\top}\widehat{\bm w}_k^{(t)}\}}{\sum_{k'=1}^{m}\exp\{\widehat{\gamma}_{k'}^{(t)}+(\widehat{\bm z}_i^{(t)})^{\top}\widehat{\bm w}_{k'}^{(t)}\}}. \label{eq:share}
\end{equation}
$\widehat{\beta}_i^{(t)}$ cancels in \eqref{eq:share}, as in \eqref{eq:comparisons}. The \emph{observed participation share} $p_{ik}^{(t)}=B_{ik}^{(t)}/\sum_{k'=1}^{m}B_{ik'}^{(t)}$ is the fraction of the participation counts of worker $i$ that falls on task type $k$ up to checkpoint $t$. Let $\widehat{\bm\pi}_i^{(t)}$ and $\bm p_i^{(t)}$ denote the vectors of the two shares over task types. At a minimizer of the joint objective \eqref{eq:main-likelihood}, $\sum_{k=1}^{m}\widehat{\theta}_{ik}^{B,(t)}=\sum_{k=1}^{m}B_{ik}^{(t)}$, so $\widehat{\bm\pi}_i^{(t)}$ differs from $\bm p_i^{(t)}$ only through $(\widehat{\bm z}_i^{(t)})^{\top}\widehat{\bm w}_k^{(t)}$: the matrix of these inner products has rank at most $d$, and $\widehat{\bm z}_i^{(t)}$ also enters $\widehat{\theta}_{i\ell}^{A,(t)}$.

\noindent{\bf Step 2: workers act on the shares.} How a worker uses the statistical summary depends on the foundation model and the prompt, so we adopt a stylized choice rule. Worker $i$ chooses task type $k$ with the following probability:
\begin{equation}
    \phi_{ik}(\widehat{\bm\pi}_i^{(t)})=\frac{(\widehat{\pi}_{ik}^{(t)})^{\omega}\exp(u_{ik})}{\sum_{k'=1}^{m}(\widehat{\pi}_{ik'}^{(t)})^{\omega}\exp(u_{ik'})}. \label{eq:choice-rule}
\end{equation}
The \emph{baseline preference} $u_{ik}$ is the log-propensity of worker $i$ for task type $k$ without the statistical summary and combines the suitability of worker $i$ for task type $k$ with the current needs. The \emph{reliance} $\omega\ge0$, common to all workers, measures how strongly workers follow the statistical summary: in the odds of two task types, the participation ratio in the second row of \eqref{eq:comparisons} enters with power $\omega$. With $\omega=0$, a worker ignores the statistical summary; with $\omega=1$, a worker treats $\widehat{\pi}_{ik}^{(t)}$ as a prior, as in Bayes' rule.

\begin{figure}[!t]
    \centering
    \includegraphics[width=\textwidth]{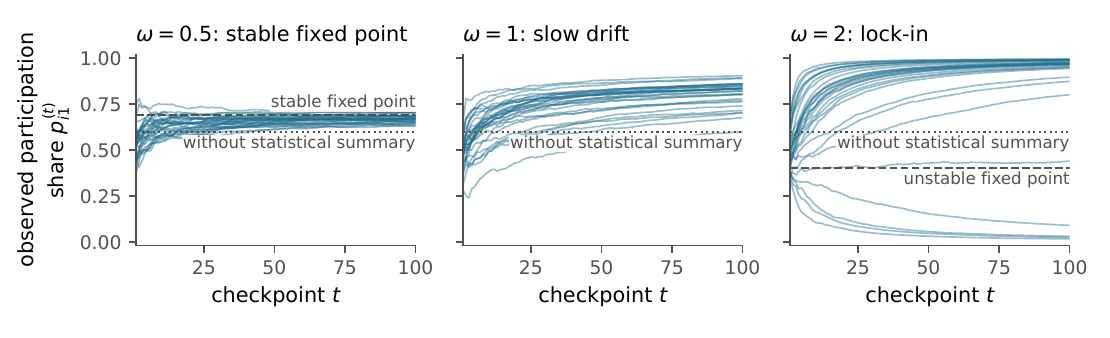}
    \caption{Simulation of \eqref{eq:recursion} in the setting of Section~\ref{sec:Feedback-of-the-Statistical-Summary}; each line is one run. Dotted lines: fixed point without the statistical summary ($\omega=0$). Dashed lines: interior fixed point, stable for $\omega=0.5$ and unstable for $\omega=2$; for $\omega=1$, no interior fixed point exists.}
    \label{fig:toy-loop}
\end{figure}

\noindent{\bf Step 3: cumulative counts feed the choices back.} We hold $u_{ik}$ fixed, take $\widehat{\bm\pi}_i^{(t)}=\bm p_i^{(t)}$, and let worker $i$ add the same number of participation counts before the first checkpoint and between consecutive checkpoints. Each checkpoint then moves $\bm p_i^{(t)}$ a fraction $1/(t+1)$ of the way toward $\bm\phi_i(\bm p_i^{(t)})$:
\begin{equation}
    \bm p_i^{(t+1)}=\bm p_i^{(t)}+\frac{1}{t+1}\left\{\bm\phi_i(\bm p_i^{(t)})-\bm p_i^{(t)}+\bm\xi_i^{(t+1)}\right\}, \label{eq:recursion}
\end{equation}
where $\bm\phi_i=(\phi_{i1},\ldots,\phi_{im})^{\top}$, and the noise $\bm\xi_i^{(t+1)}$, the share of the added participation counts minus $\bm\phi_i(\bm p_i^{(t)})$, has mean zero given the history. The recursion \eqref{eq:recursion} is a stochastic approximation with mean drift $\bm\phi_i(\bm p_i^{(t)})-\bm p_i^{(t)}$ \citep{borkar2023stochastic}. We examine its fixed points $\bm p_i^{\circ}=\bm\phi_i(\bm p_i^{\circ})$ and their local stability. At any fixed point, stable or not, the statistical summary describes the very choices that the statistical summary induces, a property known as performative stability \citep{perdomo2020performative,mendlerdunner2020stochastic}. Figure~\ref{fig:toy-loop} simulates \eqref{eq:recursion} in an illustrative setting: one worker, two task types with $u_{i1}-u_{i2}=0.4$, and $\omega\in\{0.5,1,2\}$. Each of 30 independent runs of \eqref{eq:recursion} starts from 32 participation counts split with probability $1/2$ and adds 32 participation counts per checkpoint for 100 checkpoints.

\noindent{\bf Step 4: the fixed point in the latent space.} A fixed point is interior when all $p_{ik}^{\circ}$ are positive; for $\omega>0$, each vertex, where worker $i$ performs one task type only, is also a fixed point, and a run locks in when the run converges to a vertex. For $\omega\ne1$, the interior fixed point satisfies $\log p_{ik}^{\circ}=u_{ik}/(1-\omega)+c_i$, where $c_i$ does not depend on $k$. Let $\bar u_k$ denote the average of $u_{ik}$ over workers, and let $\tilde u_{ik}$ denote $u_{ik}$ centered over workers and over task types. If every $\bm p_i^{(t)}$ equals the interior fixed point $\bm p_i^{\circ}$ and the matrix $(\tilde u_{ik})$ has rank at most $d$, the centering constraints give the following identities:
\begin{equation}
    \widehat{\gamma}_k^{(t)}=\frac{\bar u_k}{1-\omega}+c^{(t)},\qquad(\widehat{\bm z}_i^{(t)})^{\top}\widehat{\bm w}_k^{(t)}=\frac{\tilde u_{ik}}{1-\omega}, \label{eq:fixed-point}
\end{equation}
where $c^{(t)}$ absorbs the growth of the counts and does not depend on $k$. For $0\le\omega<1$, the interior fixed point is locally stable, and \eqref{eq:fixed-point} magnifies $\bar u_k$ and $\tilde u_{ik}$ by $1/(1-\omega)\ge1$ and keeps the order and the signs of $\tilde u_{ik}$: the statistical summary sharpens the division of labor. Near the interior fixed point, the mean update closes a fraction $(1-\omega)/(t+1)$ of the gap to first order, compared with $1/(t+1)$ when $\omega=0$. In the left panel of Figure~\ref{fig:toy-loop}, most runs remain below the interior fixed point at $t=100$. For $\omega=1$, if $u_{ik}>u_{ik'}$ for all $k'\ne k$, the mean drift of $p_{ik}^{(t)}$ is positive whenever $0<p_{ik}^{(t)}<1$. For $\omega>1$, the factor $1/(1-\omega)$ is negative, so the interior fixed point reverses the order of $u_{ik}$ over $k$ and is locally unstable. When a run locks in, both the early $\bm p_i^{(t)}$ and subsequent random choices can influence the selected task type. In the right panel of Figure~\ref{fig:toy-loop}, 4 of 30 runs head to the task type with the smaller $u_{ik}$.

For $\omega=0$, \eqref{eq:fixed-point} recovers $\bar u_k$ and $\tilde u_{ik}$, as in the usual analysis with latent space models. For $0\le\omega<1$, the feedback only rescales $\widehat{\gamma}_k^{(t)}$ and $(\widehat{\bm z}_i^{(t)})^{\top}\widehat{\bm w}_k^{(t)}$ by $1/(1-\omega)$, so the order and the signs of $(\widehat{\bm z}_i^{(t)})^{\top}\widehat{\bm w}_k^{(t)}$ keep the usual interpretation. In the MAS, the current needs counteract lock-in, because the open sub-tasks of a crowded task type run out and $u_{ik}$ for the crowded task type $k$ falls. Section~\ref{sec:experiment} illustrates changes in workers' participation across task types.

\section{Experiments}
\label{sec:experiment}
\begingroup
\renewcommand{\topfraction}{.9}
\renewcommand{\bottomfraction}{.8}
\renewcommand{\textfraction}{.1}
\renewcommand{\floatpagefraction}{.85}
\clubpenalty=10000
\widowpenalty=10000

\subsection{Benchmarks and Evaluation}
\label{sec:benchmarks-evaluation}

We study autonomous choices of sub-tasks and collaborators in three settings. LongBench v2 Multi-Doc QA, which we call LongBench for short, contains 125 questions that require evidence from academic, financial, governmental, legal, or multi-news documents \citep{bai2024longbenchv2}.\footnote{\url{https://github.com/THUDM/LongBench}} FATE-H contains 100 algebra theorems with formal statements in Lean, which we call original theorems to distinguish them from intermediate lemmas \citep{jiang2026fate}.\footnote{\url{https://github.com/frenzymath/FATE-H}} DeepSWE10 contains 10 long-horizon tasks of DeepSWE~v1.1, each of which asks for a patch to an open-source repository, and hidden tests check the patch \citep{huang2026deepswev11}.\footnote{\url{https://github.com/datacurve-ai/deep-swe}} DeLM uses LongBench and selects the 10 tasks of DeepSWE10 \citep{mao2026delm}, ScaleMAS uses FATE-H \citep{qi2026scalemas}, and we follow their choices of tasks and evaluation. The three benchmarks let us examine how workers combine published results into an answer, a complete proof, or a patch. LongBench records also show how review changes an existing answer, while FATE-H distinguishes intermediate lemmas from proofs of original theorems.

We compare three guidelines, that is, prompt variants for the workers: \emph{History only}, \emph{AgentCue Base}, and \emph{AgentCue Pro}. All three guidelines share the instructions for answers, proofs, patches, and review, and all three let workers specialize, that is, continue work in the task types where the earlier work of the worker has proved useful, and change direction when the current needs call for a change. History only uses the conversation history of each worker and the shared board. Base adds the statistical summary to guide the choices of sub-tasks and collaborators. Pro uses the same statistical summary and adds three components to Base: guidance on the interpretation of activity and affinity terms and on the checks of current needs in Section~\ref{sec:Summary-Guiding-in-Self-Organization}; guidance on specialization under the feedback in Section~\ref{sec:Feedback-of-the-Statistical-Summary}, where past participation informs a choice and the resulting work informs later statistical summaries; and additional checks of task requirements and supporting evidence during review. The comparison of Pro with Base therefore evaluates the three components together. Supplementary~\ref{app:prompts}, \ref{app:fateh-prompts}, and~\ref{app:deepswe-prompts} give the prompts.

For each benchmark, we conduct five matched replicates, R1--R5, each with 32 workers labeled W1--W32 that use GPT-5.6 Sol at high reasoning effort and keep their conversations across sessions. All session counts below count valid sessions. Each replicate begins with 64 warm-up sessions without a statistical summary. The shared board and the worker conversations at the end of warm-up then initialize three branches, one per guideline, and each branch continues with 192 continuation sessions. Each branch therefore has 256 sessions, and the three branches of a replicate share the 64 warm-up sessions. On DeepSWE10, each branch instead continues for at most 96 continuation sessions. Checkpoints fall every 32 sessions from the end of warm-up, that is, after 64, 96, \ldots, 256 valid sessions; on DeepSWE10, checkpoints end by 160 valid sessions. Workers choose their own sub-tasks and collaborators throughout. Sources, tools, and session limits are the same across guidelines within each benchmark. Supplementary~\ref{app:design}, \ref{app:fateh-design}, and~\ref{app:deepswe-design} specify the limits and the counting of sessions. Realized model usage and computational cost can differ across guidelines.

LongBench accuracy uses the latest submitted answers and the withheld benchmark key, with all 125 questions in the denominator. FATE-H reports the proportion of the 100 original theorems with a complete proof that fresh Lean verification of the saved proof certifies under the fixed environment. Intermediate lemmas and review verdicts alone do not count as certified original theorems. DeepSWE10 runs the hidden tests of each task on the final candidate patch, that is, the patch that workers last designate for the task. A task counts as resolved when the candidate patch passes all hidden tests, and the F2P pass rate averages over the 10 tasks the proportion of passed fail-to-pass tests, which fail at the base commit and check the behavior that the task requires. Unanswered questions, uncertified original theorems, and tasks without a candidate patch remain in the denominators, and a task without a candidate patch has an F2P pass rate of zero. For FATE-H, we also examine recorded proof repairs and reads of published results of other workers during the 192 continuation sessions (Supplementary~\ref{app:fateh-results}).

\subsection{Main Results}
\label{sec:main-results}

AgentCue Pro achieves the highest mean final accuracy on LongBench, the highest mean rate of certified original theorems on FATE-H, and the highest mean rate of resolved tasks on DeepSWE10 (Table~\ref{tab:outcomes}). Pro corrects more existing answers through revision on LongBench, certifies more original theorems through proof repairs on FATE-H, and leaves fewer tasks without a candidate patch on DeepSWE10.

\begin{table}[!htbp]
\centering
\spacingset{1}\small\setlength{\tabcolsep}{3pt}
\caption{Results by guideline. Entries are means $\pm$ sample standard deviations across five replicates, except the pooled ratio of board requests. Reads count only published results of other workers. Bold marks the best values.}
\label{tab:outcomes}
\begin{tabularx}{\textwidth}{@{}>{\raggedright\arraybackslash}Xccc@{}}
\toprule
Evaluation Metrics & AgentCue Pro & AgentCue Base & History only\\
\midrule
\multicolumn{4}{@{}l}{\textit{LongBench v2 Multi-Doc QA}}\\[2pt]
Final accuracy (\%) $\uparrow$ & \textbf{59.52} $\pm$ 0.72 & 57.44 $\pm$ 1.73 & 57.12 $\pm$ 0.44\\
Wrong-to-right changes after review $\uparrow$ & \textbf{4.20} $\pm$ 3.77 & 3.00 $\pm$ 2.65 & 2.00 $\pm$ 1.22\\
Right-to-wrong changes after review $\downarrow$ & \textbf{0.80} $\pm$ 0.84 & 2.00 $\pm$ 1.22 & 1.80 $\pm$ 0.84\\
\mbox{Board requests per wrong-to-right change $\downarrow$} & \textbf{389.8} & 485.2 & 787.8\\
\midrule
\multicolumn{4}{@{}l}{\textit{FATE-H}}\\[2pt]
Certified original theorems (\%) $\uparrow$ & \textbf{81.40} $\pm$ 1.52 & 78.60 $\pm$ 2.30 & 78.20 $\pm$ 2.17\\
Failed proofs repaired and certified $\uparrow$ & \textbf{3.20} $\pm$ 1.48 & 2.80 $\pm$ 0.84 & 1.80 $\pm$ 1.10\\
\mbox{Reads of published results per session $\downarrow$} & \textbf{4.38} $\pm$ 0.26 & 5.27 $\pm$ 0.33 & 4.74 $\pm$ 0.44\\
\midrule
\multicolumn{4}{@{}l}{\textit{DeepSWE10}}\\[2pt]
Resolved tasks (\%) $\uparrow$ & \textbf{70.00} $\pm$ 10.00 & \textbf{70.00} $\pm$ 10.00 & 56.00 $\pm$ 16.73\\
F2P pass rate (\%) $\uparrow$ & \textbf{87.43} $\pm$ 10.09 & 84.77 $\pm$ 8.03 & 73.25 $\pm$ 18.99\\
Tasks without a candidate patch $\downarrow$ & \textbf{0.40} $\pm$ 0.89 & 0.60 $\pm$ 0.89 & 1.60 $\pm$ 2.07\\
\bottomrule
\end{tabularx}
\end{table}

On LongBench, every branch has submitted an answer to all 125 questions by 160 valid sessions. Pro's accuracy continues to improve after this point, while the other guidelines change little on average (Figure~\ref{fig:progress}(b)). These later gains therefore come from revisions to answers already on the shared board.
Revisions that adopt review findings show more wrong-to-right changes and fewer right-to-wrong changes under Pro on average (Table~\ref{tab:outcomes}), in line with the review checks of Pro. For example, in a question about video pretraining in R2, a reviewer finds in the sources that the answer confuses an interface that represents human actions without loss with a model of their full joint distribution, and a later worker uses the finding to change the answer from A to the correct D (Supplementary~\ref{app:review}). Pro also completes fewer reviews but spends more total board requests on review and revision, so the lower pooled ratio of board requests per wrong-to-right change under Pro reflects a greater yield of corrections from the review effort. The practical benefit is a more favorable balance between correcting errors and preserving answers that were already correct.

\begin{figure}[!t]
\centering
\spacingset{1}\captionsetup{font=small}
\includegraphics[width=\textwidth]{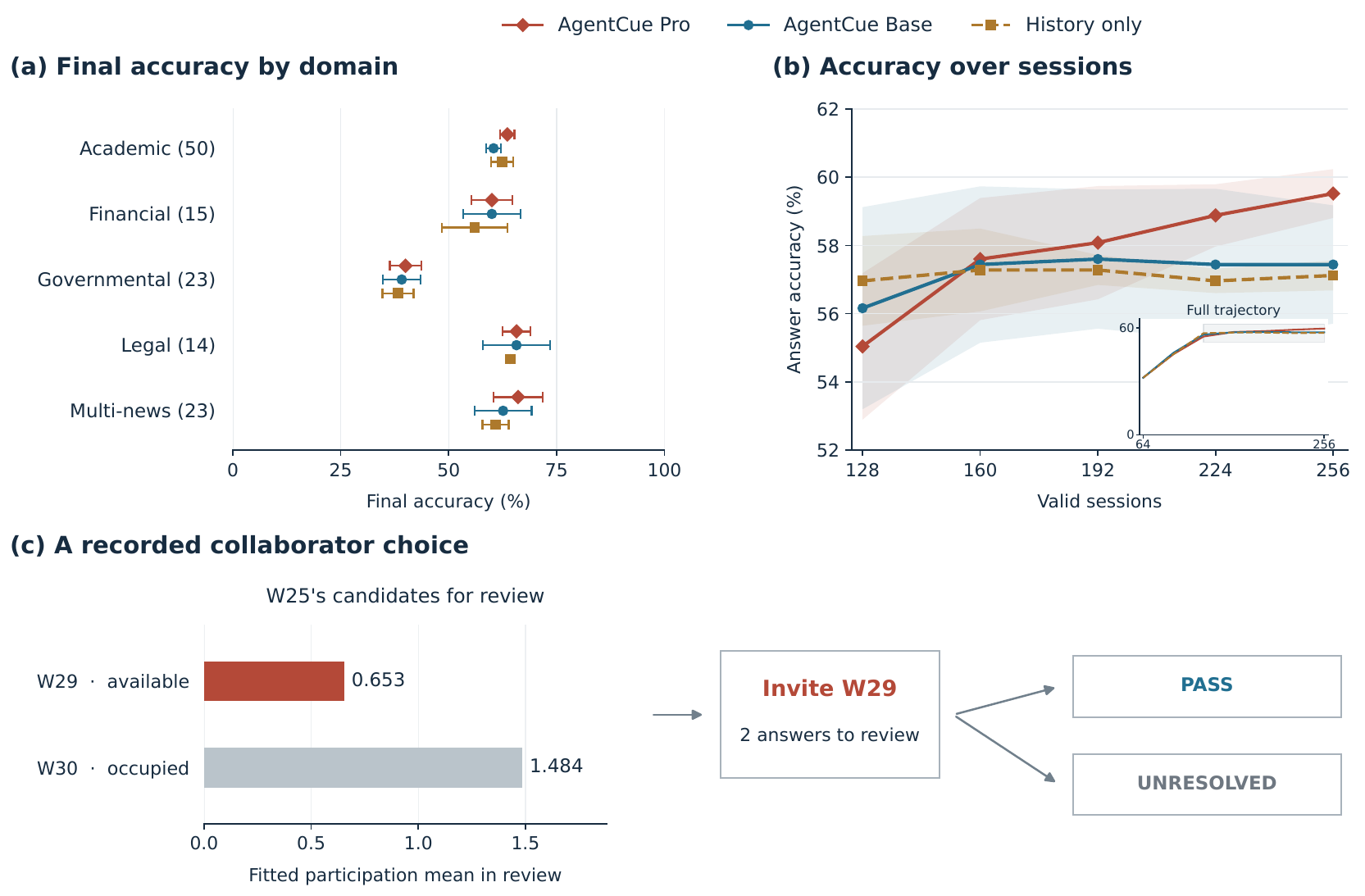}
\caption{LongBench accuracy and collaborator choice. (a) Final accuracy by domain, with question counts in parentheses. (b) Accuracy over valid sessions, with the full trajectory in the inset. Accuracy uses all 125 questions at every checkpoint, with unsubmitted answers counted as wrong. Points and lines show means across five replicates, and error bars and shading show one sample standard deviation. (c) A collaborator choice in the Pro branch of R1, based on fitted participation means in review and current availability.}
\label{fig:progress}
\end{figure}

The Choice Records show how the statistical summary supports the choice of a collaborator. In Figure~\ref{fig:progress}(c), W25 seeks a reviewer for two answers and compares candidates by their fitted participation means in review and their fitted collaboration means with W25. W30 has larger fitted means but is occupied, so W25 invites W29, who is available. W29 gives one answer a PASS verdict and the other an UNRESOLVED verdict, and neither answer changes. The statistical summary narrows the candidates, while current availability and source checks settle the choice and the review. Figure~\ref{fig:workflow} follows the Pro branch of R1: workers repeat task types across sessions and also change task types, and every worker contributes to at least three task types. Section~\ref{sec:Feedback-of-the-Statistical-Summary} explains why these observed participation shares also reflect earlier choices that relied on the statistical summary.

\begin{figure}[p]
\centering
\spacingset{1}\captionsetup{font=small}
\includegraphics[width=\textwidth,height=.80\textheight,keepaspectratio]{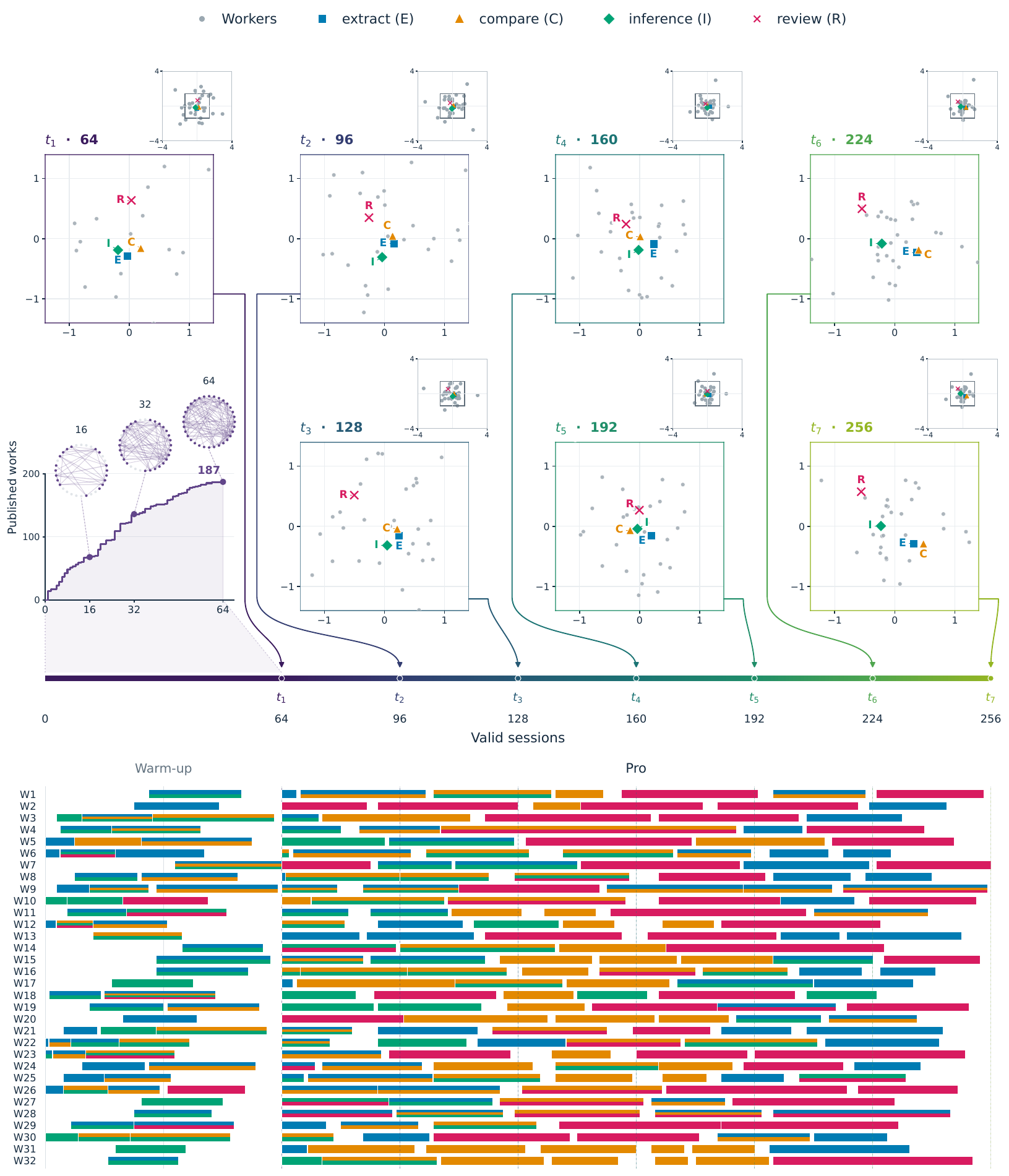}
\caption{LongBench participation and statistical summaries in the Pro branch of R1. The seven latent-space panels show the latent positions at checkpoints $t_1,\ldots,t_7$, aligned to the estimate at $t_7$ by one orthogonal transformation, and the small square above each panel shows the full range of the latent positions. The panel at the left shows the warm-up sessions: the cumulative number of published works and, at 16, 32, and 64 sessions, the worker pairs, with workers placed on a circle for display only. The lower stripes mark the task types of the published work in each session of each worker; stripe height does not measure task duration, and blank intervals have no task type.}
\label{fig:workflow}
\end{figure}

On FATE-H, Pro and Base have the same mean certification rate at 192 valid sessions, and their final difference emerges over the remaining 64 sessions (Figure~\ref{fig:fateh-progress}(b)). Field theory accounts for much of the gain over Base, and the advantage varies across mathematical subjects (Figure~\ref{fig:fateh-progress}(a)). A checked intermediate lemma can be useful without a proof of its original theorem, so progress requires workers to close the remaining proof obligations. The recorded repairs provide one route: a worker returns to a published proof that failed a Lean check and supplies a corrected proof of the unchanged original theorem that passes fresh verification. Pro certifies more original theorems through this route on average and reads fewer published results of other workers per continuation session (Table~\ref{tab:outcomes}).

\begin{figure}[!t]
\centering
\spacingset{1}\captionsetup{font=small}
\includegraphics[width=\textwidth]{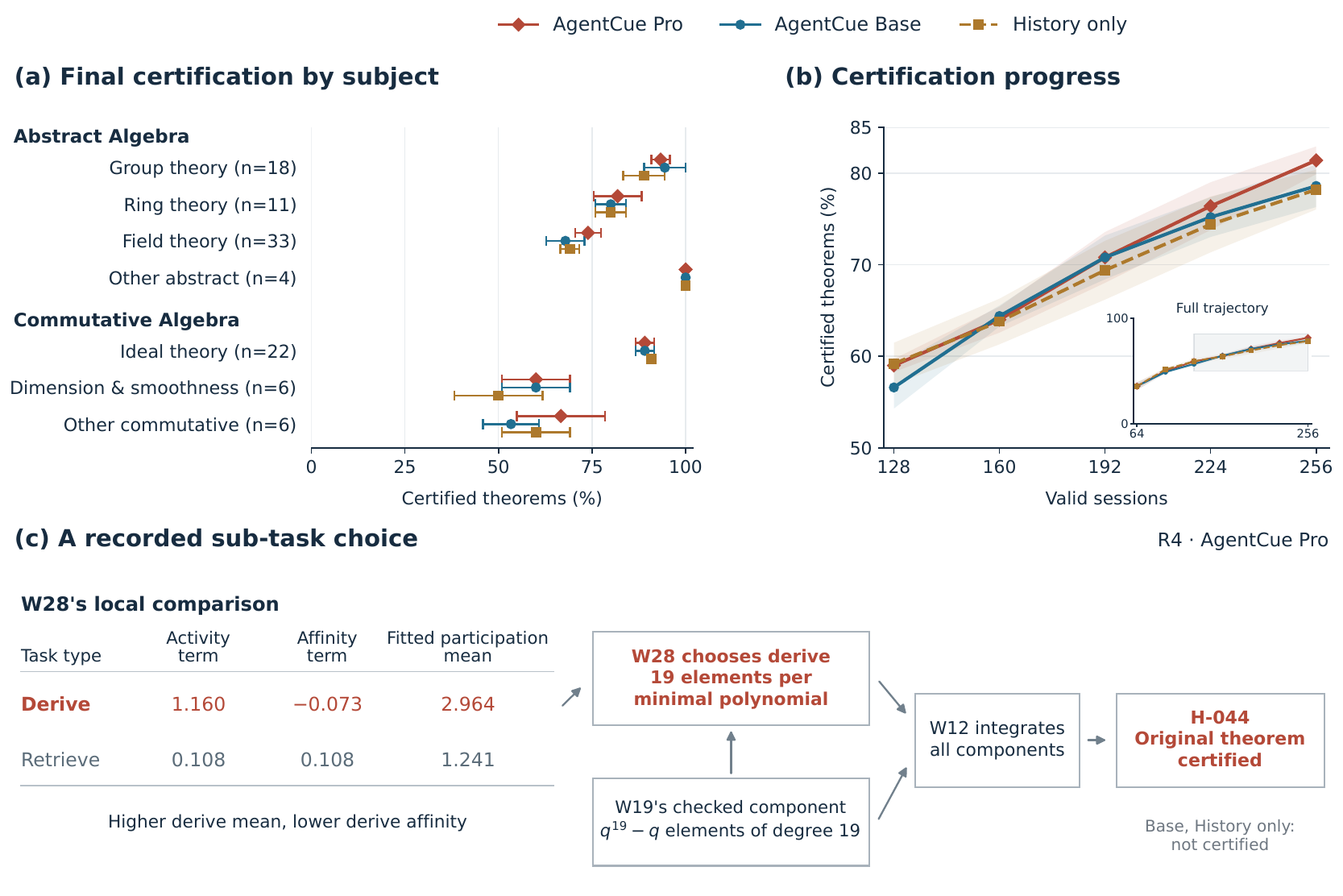}
\caption{FATE-H certification and sub-task choice. (a) Final certification by mathematical subject. (b) Certification progress, with the full trajectory in the inset. Error bars and shading show one sample standard deviation across five replicates. (c) A sub-task choice in the Pro branch of R4, based on activity and affinity terms and the remaining proof obligation.}
\label{fig:fateh-progress}
\end{figure}

The Choice Records show how Pro reads the activity and affinity terms in a local comparison. In Figure~\ref{fig:fateh-progress}(c), the original theorem H-044 states that a finite field with $q$ elements has $(q^{19}-q)/19$ monic irreducible polynomials of degree 19. W19 has published a checked proof that $q^{19}-q$ elements of a degree-19 extension have minimal polynomials of degree 19, and W16 posts the remaining link between these elements and the polynomials as an open sub-task, which W28 takes up. For W28, the fitted participation mean in derive exceeds that in retrieve (2.964 versus 1.241), while the affinity term of derive is lower ($-0.073$ versus $0.108$), so the activity terms (1.160 versus 0.108) account for the gap. The Choice Record of W28 states this split and chooses derive because the sub-task needs a new proof rather than a retrieved declaration. W28 proves that each monic irreducible polynomial of degree 19 is the minimal polynomial of exactly 19 elements of an extension in which the polynomial splits, and W12 supplies the remaining lemmas and combines the checked results into a proof of H-044 that passes fresh verification. In R4, neither Base nor History only certifies H-044, and across the five replicates, Pro, Base, and History only certify H-044 in five, three, and two replicates, respectively. The split shows W28 that the larger fitted mean comes from the participation that derive attracts across all workers rather than from a lean of W28 toward derive, and the remaining proof obligation settles the choice. Supplementary~\ref{app:fateh-results} gives the proof and the recorded comparison.

On DeepSWE10, Pro resolves 70.00\% of the tasks on average, and History only resolves 56.00\% (Table~\ref{tab:outcomes}). Pro also has the highest F2P pass rate and leaves the fewest tasks without a candidate patch. Across the five replicates, Pro resolves a task in eight cases where History only does not, and History only has no candidate patch in five of these cases (Supplementary~\ref{app:deepswe-results}). The first candidate patch of a task must use a published result of another worker on the same task (Supplementary~\ref{app:deepswe-prompts}), so these five cases end without a patch that combines the work of different workers.

\endgroup

\section{Discussion}

As an application of statistics to agent systems, our method replaces repeated searches of the shared board in large-scale decentralized MAS with a statistical summary of participation and collaboration records. In the LongBench experiment, AgentCue Pro has the highest mean accuracy, with more wrong-to-right and fewer right-to-wrong changes after review than Base, and Base and History only have similar mean accuracies. On FATE-H, Pro attains the highest mean certification rate for the original theorems, with more original theorems certified through recorded proof repairs and fewer reads of published results of other workers per continuation session. On DeepSWE10, Pro has the highest mean rate of resolved tasks and the highest F2P pass rate. The Choice Records show how Pro combines fitted participation and collaboration means with current needs, including how a worker distinguishes activity from affinity when choosing the contribution needed to complete a proof. The theory in Supplementary~\ref{app:theory} gives error bounds for the estimator and its prescribed updates under stated assumptions. More broadly, our method builds a stage on which statistics acts inside agent systems: the MAS generates network data, the statistical summary guides the MAS, and other statistical methods can join the loop.

We plan two extensions of our method. First, owing to the cost of foundation-model calls, our experiments use moderate-scale MAS, and a direct next step is to scale the experiments up to 1,024 workers, as in Agensh \citep{zhan2026agensh}. Second, as workers accumulate private memories, early collaboration records may no longer reflect current collaboration. Drawing on continual learning \citep{wang2024continual} and agent memory \citep{park2023generative}, the latent space model can down-weight the early collaboration records through discounted edge counts or time-varying latent positions \citep{durante2014dynamic}.



\noindent{\bf Disclosure Statement.}
The authors report there are no competing interests to declare.

\noindent{\bf Declaration of AI Use.}
OpenAI's GPT-6 Astra was used to assist with literature searches and language editing, and GPT-5.6 Sol was used as the foundation model in the experiments.

\noindent{\bf Data Availability Statement.}
LongBench v2, FATE-H, and DeepSWE are publicly available at \url{https://github.com/THUDM/LongBench}, \url{https://github.com/frenzymath/FATE-H}, and \url{https://github.com/datacurve-ai/deep-swe}, respectively.

\noindent{\bf Supplementary Material.}
The supplementary file contains theoretical details and proofs, experimental procedures, and additional results. (PDF) The archive contains the code, prompts, final results, and scripts for reproducing the reported tables and figures. (ZIP) The code and results are also available at \url{https://github.com/zgj19stat/AgentCue}.

\begingroup
\hypersetup{urlcolor=black}
\setlength{\bibsep}{1pt plus 1pt minus 1pt}
\bibliography{reference}

@misc{apartsin2026swarmcf,
  author = {Apartsin, Alexander and Meshulam, Yigal and Aperstein, Yehudit},
  title = {Acting on the Unseen: Communication-Free Collaborative Filtering for Decentralized Multi-Robot Task Allocation},
  year = {2026},
  howpublished = {arXiv:2605.25584},
  doi = {10.48550/arXiv.2605.25584}
}

@inproceedings{bai2024longbenchv2,
  author = {Bai, Yushi and Tu, Shangqing and Zhang, Jiajie and Peng, Hao and Wang, Xiaozhi and Lv, Xin and Cao, Shulin and Xu, Jiazheng and Hou, Lei and Dong, Yuxiao and Tang, Jie and Li, Juanzi},
  title = {{LongBench v2}: Towards Deeper Understanding and Reasoning on Realistic Long-Context Multitasks},
  booktitle = {Proceedings of the 63rd Annual Meeting of the Association for Computational Linguistics (Volume 1: Long Papers)},
  year = {2025},
  pages = {3639--3664},
  publisher = {Association for Computational Linguistics},
  doi = {10.18653/v1/2025.acl-long.183}
}

@book{borkar2023stochastic,
  author = {Borkar, Vivek S.},
  title = {Stochastic Approximation: A Dynamical Systems Viewpoint},
  series = {Texts and Readings in Mathematics},
  year = {2023},
  volume = {48},
  edition = {2nd},
  publisher = {Springer},
  address = {Singapore},
  doi = {10.1007/978-981-99-8277-6}
}

@article{butts2008relational,
  author = {Butts, Carter T.},
  title = {A Relational Event Framework for Social Action},
  journal = {Sociological Methodology},
  year = {2008},
  volume = {38},
  number = {1},
  pages = {155--200},
  doi = {10.1111/j.1467-9531.2008.00203.x}
}

@misc{chen2026meshheal,
  author = {Chen, Keru and Lin, Sen and Liang, Yingbin and Bastian, Nathaniel D. and Zou, Shaofeng},
  title = {{MeshHeal}: Two-Timescale Self-Healing for Gray Failures in Decentralized {LLM} Agent Networks},
  year = {2026},
  howpublished = {arXiv:2609.29015},
  doi = {10.48550/arXiv.2609.29015}
}

@inproceedings{durante2014dynamic,
  author = {Durante, Daniele and Dunson, David B.},
  title = {{Bayesian} Logistic {Gaussian} Process Models for Dynamic Networks},
  booktitle = {Proceedings of the Seventeenth International Conference on Artificial Intelligence and Statistics},
  series = {Proceedings of Machine Learning Research},
  year = {2014},
  volume = {33},
  pages = {194--201},
  publisher = {PMLR}
}

@article{fan2024generalized,
  author = {Fan, Xinyan and Fang, Kuangnan and Pu, Dan and Qin, Ruixuan},
  title = {Generalized Latent Space Model for One-Mode Networks With Awareness of Two-Mode Networks},
  journal = {Computational Statistics \& Data Analysis},
  year = {2024},
  volume = {193},
  pages = {107915},
  doi = {10.1016/j.csda.2023.107915}
}

@inproceedings{jiang2026fate,
  author = {Jiang, Jiedong and He, Wanyi and Wang, Yuefeng and Gao, Guoxiong and Hu, Yongle and Wang, Jingting and Guan, Nailin and Wu, Peihao and Dai, Chunbo and Xiao, Liang and Dong, Bin},
  title = {{FATE}: A Formal Benchmark Series for Frontier Algebra of Multiple Difficulty Levels},
  booktitle = {International Conference on Learning Representations},
  year = {2026},
  pages = {10543--10580}
}

@misc{gao2026autoscientists,
  author = {Gao, Shanghua and Fang, Ada and Zitnik, Marinka},
  title = {{AutoScientists}: Self-Organizing Agent Teams for Long-Running Scientific Experimentation},
  year = {2026},
  howpublished = {arXiv:2605.28655},
  doi = {10.48550/arXiv.2605.28655},
  note = {Accepted at NeurIPS 2026; to appear}
}

@article{he2025semiparametric,
  author = {He, Yinqiu and Sun, Jiajin and Tian, Yuang and Ying, Zhiliang and Feng, Yang},
  title = {Semiparametric Modeling and Analysis for Longitudinal Network Data},
  journal = {The Annals of Statistics},
  year = {2025},
  volume = {53},
  number = {4},
  pages = {1406--1430},
  doi = {10.1214/25-AOS2506}
}

@article{hoff2002latent,
  author = {Hoff, Peter D. and Raftery, Adrian E. and Handcock, Mark S.},
  title = {Latent Space Approaches to Social Network Analysis},
  journal = {Journal of the American Statistical Association},
  year = {2002},
  volume = {97},
  number = {460},
  pages = {1090--1098},
  doi = {10.1198/016214502388618906}
}

@inproceedings{hong2024metagpt,
  author = {Hong, Sirui and Zhuge, Mingchen and Chen, Jonathan and Zheng, Xiawu and Cheng, Yuheng and Zhang, Ceyao and Wang, Jinlin and Wang, Zili and Yau, Steven Ka Shing and Lin, Zijuan and Zhou, Liyang and Ran, Chenyu and Xiao, Lingfeng and Wu, Chenglin and Schmidhuber, J{\"u}rgen},
  title = {{MetaGPT}: Meta Programming for a Multi-Agent Collaborative Framework},
  booktitle = {The Twelfth International Conference on Learning Representations},
  year = {2024},
  pages = {23247--23275}
}

@misc{huang2023agentcoder,
  author = {Huang, Dong and Zhang, Jie M. and Luck, Michael and Bu, Qingwen and Qing, Yuhao and Cui, Heming},
  title = {{AgentCoder}: Multi-Agent-Based Code Generation With Iterative Testing and Optimisation},
  year = {2023},
  howpublished = {arXiv:2312.13010},
  doi = {10.48550/arXiv.2312.13010}
}

@misc{huang2026deals,
  author = {Huang, Jingjuan and Wang, Wenbin and Yin, Yanchuan and Velasquez, Alvaro and Liu, Jia},
  title = {{DEALS}: Decentralized Expertise-Aware Load Serving for Multi-Agent {LLM} Systems},
  year = {2026},
  howpublished = {arXiv:2609.33768},
  doi = {10.48550/arXiv.2609.33768}
}

@misc{huang2026deepswe,
  author = {Huang, Wenqi and Lee, Charley and Tng, Leonard and Ge, Serena},
  title = {{DeepSWE}: Measuring Frontier Coding Agents on Original, Long-Horizon Engineering Tasks},
  year = {2026},
  howpublished = {Datacurve benchmark report},
  url = {https://deepswe.datacurve.ai/blog/deepswe}
}

@misc{huang2026deepswev11,
  author = {Huang, Wenqi and Jiang, Peter},
  title = {{DeepSWE} v1.1: A Cleaner, More Reproducible Benchmark for Frontier Coding Agents},
  year = {2026},
  howpublished = {GitHub repository},
  url = {https://github.com/datacurve-ai/deep-swe}
}

@misc{kim2025scaling,
  author = {Kim, Yubin and Gu, Ken and Park, Chanwoo and Park, Chunjong and Schmidgall, Samuel and Heydari, A. Ali and Yan, Yao and Zhang, Zhihan and Zhuang, Yuchen and Liu, Yun and Malhotra, Mark and Liang, Paul Pu and Park, Hae Won and Yang, Yuzhe and Xu, Xuhai and Du, Yilun and Patel, Shwetak and Althoff, Tim and McDuff, Daniel and Liu, Xin},
  title = {Towards a Science of Scaling Agent Systems},
  year = {2025},
  howpublished = {arXiv:2512.08296},
  doi = {10.48550/arXiv.2512.08296}
}

@article{lerner2023coauthor,
  author = {Lerner, J{\"u}rgen and H{\^a}ncean, Marian-Gabriel},
  title = {Micro-Level Network Dynamics of Scientific Collaboration and Impact: Relational Hyperevent Models for the Analysis of Coauthor Networks},
  journal = {Network Science},
  year = {2023},
  volume = {11},
  number = {1},
  pages = {5--35},
  doi = {10.1017/nws.2022.29}
}

@inproceedings{li2023camel,
  author = {Li, Guohao and Hammoud, Hasan Abed Al Kader and Itani, Hani and Khizbullin, Dmitrii and Ghanem, Bernard},
  title = {{CAMEL}: Communicative Agents for `Mind' Exploration of Large Language Model Society},
  booktitle = {Advances in Neural Information Processing Systems},
  year = {2023},
  volume = {36},
  pages = {51991--52008},
  doi = {10.52202/075280-2264}
}

@misc{li2026beyondmoreagents,
  author = {Li, Songyuan and Jia, Xiaolong and Tam, Wai Fong and Tabaro, Leon and Li, Qilei and Liu, Gaoyang and Wang, Chen and Abdelmoniem, Ahmed M.},
  title = {Beyond More Agents: A Survey of Collaboration Mechanisms in Multi-Agent {LLM} Systems},
  year = {2026},
  howpublished = {SSRN 7243979},
  doi = {10.2139/ssrn.7243979}
}

@inproceedings{liu2024dylan,
  author = {Liu, Zijun and Zhang, Yanzhe and Li, Peng and Liu, Yang and Yang, Diyi},
  title = {A Dynamic {LLM}-Powered Agent Network for Task-Oriented Agent Collaboration},
  booktitle = {First Conference on Language Modeling},
  year = {2024}
}

@misc{liu2026markets,
  author = {Liu, Xiao and Li, Haoyang and Li, Songwei and Fang, Hongbo and Xu, Fengli and Shi, Feng and Evans, James},
  title = {Markets, Not Planners: Decentralized Orchestration of {LLM} Agents With Private Information},
  year = {2026},
  howpublished = {arXiv:2608.23867},
  doi = {10.48550/arXiv.2608.23867}
}

@article{loyal2023eigenmodel,
  author = {Loyal, Joshua Daniel and Chen, Yuguo},
  title = {An Eigenmodel for Dynamic Multilayer Networks},
  journal = {Journal of Machine Learning Research},
  year = {2023},
  volume = {24},
  number = {128},
  pages = {1--69}
}

@misc{lu2024morphagent,
  author = {Lu, Siyuan and Shao, Jiaqi and Luo, Bing and Lin, Tao},
  title = {{MorphAgent}: Empowering Agents Through Self-Evolving Profiles and Decentralized Collaboration},
  year = {2024},
  howpublished = {arXiv:2410.15048},
  doi = {10.48550/arXiv.2410.15048}
}

@misc{luo2026dytopo,
  author = {Lu, Yuxing and Hu, Yucheng and Zhao, Xukai and Cao, Jiuxin},
  title = {{DyTopo}: Dynamic Topology Routing for Multi-Agent Reasoning via Semantic Matching},
  year = {2026},
  howpublished = {arXiv:2602.06039},
  doi = {10.48550/arXiv.2602.06039}
}

@inproceedings{lyu2026thebotcompany,
  author = {Lyu, Wenhan and Xiao, Yue and Zhang, Yixuan and Sun, Yifan},
  title = {{TheBotCompany}: Self-Organizing Multi-Agent Systems for Continuous Software Development},
  booktitle = {Proceedings of the 41st IEEE/ACM International Conference on Automated Software Engineering},
  year = {2026},
  pages = {349--361},
  publisher = {Association for Computing Machinery},
  doi = {10.1145/3832783.3834358}
}

@article{ma2020universal,
  author = {Ma, Zhuang and Ma, Zongming and Yuan, Hongsong},
  title = {Universal Latent Space Model Fitting for Large Networks With Edge Covariates},
  journal = {Journal of Machine Learning Research},
  year = {2020},
  volume = {21},
  number = {4},
  pages = {1--67}
}

@misc{qi2026scalemas,
  author = {Qi, Zhenting and Zhao, Zhuokai and Shao, Rulin and Oktar, Kerem and Yao, Zhiyuan and Lakkaraju, Himabindu and Yih, Wen-tau and Ni, Ansong and Paranjape, Bhargavi and Du, Yilun},
  title = {Scaling Multi-Agent Intelligence with Substrate-Centric Coordination},
  year = {2026},
  howpublished = {arXiv:2610.16120},
  doi = {10.48550/arXiv.2610.16120}
}

@misc{mao2026delm,
  author = {Mao, Yuzhen and Gu, Jerry and Chauhan, Aadi and Zhang, Qizheng and Kang, Hangoo and Mirhoseini, Azalia},
  title = {Decentralized Multi-Agent Systems With Shared Context},
  year = {2026},
  howpublished = {arXiv:2606.10662v3},
  doi = {10.48550/arXiv.2606.10662}
}

@inproceedings{mendlerdunner2020stochastic,
  author = {Mendler-D{\"u}nner, Celestine and Perdomo, Juan C. and Zrnic, Tijana and Hardt, Moritz},
  title = {Stochastic Optimization for Performative Prediction},
  booktitle = {Advances in Neural Information Processing Systems},
  year = {2020},
  volume = {33},
  pages = {4929--4939}
}

@misc{mieczkowski2026latte,
  author = {Mieczkowski, Elizabeth and Ku, Alexander and Eisape, Tiwalayo and Arumugam, Dilip and Matters, John and Collins, Katherine M. and Sucholutsky, Ilia and Griffiths, Thomas L.},
  title = {Improving the Efficiency of Language Agent Teams With Adaptive Task Graphs},
  year = {2026},
  howpublished = {arXiv:2605.06320},
  doi = {10.48550/arXiv.2605.06320},
  note = {Accepted at NeurIPS 2026; to appear}
}

@misc{pappu2026sat,
  author = {Pappu, Aneesh and Suzgun, Mirac and Kwon, Yongchan and Bianchi, Federico and El, Batu and Kochenderfer, Mykel J. and Cao, Hancheng and Zou, James},
  title = {Self-Organizing Agent Teams Learn to Reason Together},
  year = {2026},
  howpublished = {arXiv:2609.22682},
  doi = {10.48550/arXiv.2609.22682}
}

@inproceedings{park2023generative,
  author = {Park, Joon Sung and O'Brien, Joseph C. and Cai, Carrie J. and Morris, Meredith Ringel and Liang, Percy and Bernstein, Michael S.},
  title = {Generative Agents: Interactive Simulacra of Human Behavior},
  booktitle = {Proceedings of the 36th Annual ACM Symposium on User Interface Software and Technology},
  year = {2023},
  pages = {1--22},
  publisher = {Association for Computing Machinery},
  doi = {10.1145/3586183.3606763}
}

@inproceedings{perdomo2020performative,
  author = {Perdomo, Juan C. and Zrnic, Tijana and Mendler-D{\"u}nner, Celestine and Hardt, Moritz},
  title = {Performative Prediction},
  booktitle = {Proceedings of the 37th International Conference on Machine Learning},
  series = {Proceedings of Machine Learning Research},
  year = {2020},
  volume = {119},
  pages = {7599--7609},
  publisher = {PMLR}
}

@article{perry2013point,
  author = {Perry, Patrick O. and Wolfe, Patrick J.},
  title = {Point Process Modelling for Directed Interaction Networks},
  journal = {Journal of the Royal Statistical Society: Series B (Statistical Methodology)},
  year = {2013},
  volume = {75},
  number = {5},
  pages = {821--849},
  doi = {10.1111/rssb.12013}
}

@article{pu2026tree,
  author = {Pu, Dan and Fan, Xinyan and Fang, Kuangnan},
  title = {Tree-Enhanced Latent Space Models for Two-Mode Networks},
  journal = {Journal of Computational and Graphical Statistics},
  year = {2026},
  volume = {35},
  number = {1},
  pages = {295--303},
  doi = {10.1080/10618600.2025.2527295}
}

@misc{qi2026economy,
  author = {Qi, Zhenting and Su, Huangyuan and Qu, Ao and Wang, Chenyu and Yao, Yu and Zheng, Han and Chattopadhyay, Kushal and Xu, Guowei and Wang, Zihan and Ye, Weirui and Reddi, Vijay Janapa and Li, Ju and Liang, Paul Pu and Lakkaraju, Himabindu and Kakade, Sham and Du, Yilun},
  title = {Economy of Minds: Emerging Multi-Agent Intelligence With Economic Interactions},
  year = {2026},
  howpublished = {arXiv:2606.02859},
  doi = {10.48550/arXiv.2606.02859},
  note = {Accepted at NeurIPS 2026; to appear}
}

@inproceedings{qian2024chatdev,
  author = {Qian, Chen and Liu, Wei and Liu, Hongzhang and Chen, Nuo and Dang, Yufan and Li, Jiahao and Yang, Cheng and Chen, Weize and Su, Yusheng and Cong, Xin and Xu, Juyuan and Li, Dahai and Liu, Zhiyuan and Sun, Maosong},
  title = {{ChatDev}: Communicative Agents for Software Development},
  booktitle = {Proceedings of the 62nd Annual Meeting of the Association for Computational Linguistics (Volume 1: Long Papers)},
  year = {2024},
  pages = {15174--15186},
  publisher = {Association for Computational Linguistics},
  doi = {10.18653/v1/2024.acl-long.810}
}

@article{sun2026lambda,
  author = {Sun, Maojun and Han, Ruijian and Jiang, Binyan and Qi, Houduo and Sun, Defeng and Yuan, Yancheng and Huang, Jian},
  title = {{LAMBDA}: A Large Model Based Data Agent},
  journal = {Journal of the American Statistical Association},
  year = {2026},
  volume = {121},
  number = {553},
  pages = {1--13},
  doi = {10.1080/01621459.2025.2510000}
}

@misc{tambwekar2026proxifield,
  author = {Tambwekar, Pradyumna and Feng, Yenchia and Patel, Deep and Maamari, Karime},
  title = {{Proxifield}: Decentralized Multi-Agent Communication Through Semantic Proximity},
  year = {2026},
  howpublished = {arXiv:2609.20889},
  doi = {10.48550/arXiv.2609.20889}
}

@inproceedings{tastan2025selforg,
  author = {Tastan, Nurbek and Horv{\'a}th, Samuel and Nandakumar, Karthik},
  title = {Stochastic Self-Organization in Multi-Agent Systems},
  booktitle = {The Fourteenth International Conference on Learning Representations},
  year = {2026},
  pages = {40242--40271}
}

@article{tibshirani1996lasso,
  author = {Tibshirani, Robert},
  title = {Regression Shrinkage and Selection via the Lasso},
  journal = {Journal of the Royal Statistical Society: Series B (Methodological)},
  year = {1996},
  volume = {58},
  number = {1},
  pages = {267--288},
  doi = {10.1111/j.2517-6161.1996.tb02080.x}
}

@article{tropp2011,
  author = {Tropp, Joel A.},
  title = {Freedman's Inequality for Matrix Martingales},
  journal = {Electronic Communications in Probability},
  year = {2011},
  volume = {16},
  pages = {262--270},
  doi = {10.1214/ECP.v16-1624}
}

@inproceedings{tu2016,
  author = {Tu, Stephen and Boczar, Ross and Simchowitz, Max and Soltanolkotabi, Mahdi and Recht, Ben},
  title = {Low-Rank Solutions of Linear Matrix Equations via {Procrustes} Flow},
  booktitle = {Proceedings of the 33rd International Conference on Machine Learning},
  series = {Proceedings of Machine Learning Research},
  year = {2016},
  volume = {48},
  pages = {964--973},
  publisher = {PMLR}
}

@article{wang2024continual,
  author = {Wang, Liyuan and Zhang, Xingxing and Su, Hang and Zhu, Jun},
  title = {A Comprehensive Survey of Continual Learning: Theory, Method and Application},
  journal = {IEEE Transactions on Pattern Analysis and Machine Intelligence},
  year = {2024},
  volume = {46},
  number = {8},
  pages = {5362--5383},
  doi = {10.1109/TPAMI.2024.3367329}
}

@inproceedings{wu2024autogen,
  author = {Wu, Qingyun and Bansal, Gagan and Zhang, Jieyu and Wu, Yiran and Li, Beibin and Zhu, Erkang and Jiang, Li and Zhang, Xiaoyun and Zhang, Shaokun and Liu, Jiale and Awadallah, Ahmed Hassan and White, Ryen W. and Burger, Doug and Wang, Chi},
  title = {{AutoGen}: Enabling Next-Gen {LLM} Applications via Multi-Agent Conversations},
  booktitle = {First Conference on Language Modeling},
  year = {2024}
}

@inproceedings{yang2025agentnet,
  author = {Yang, Yingxuan and Chai, Huacan and Shao, Shuai and Song, Yuanyi and Qi, Siyuan and Rui, Renting and Zhang, Weinan},
  title = {{AgentNet}: Decentralized Evolutionary Coordination for {LLM}-Based Multi-Agent Systems},
  booktitle = {Advances in Neural Information Processing Systems},
  year = {2025},
  volume = {38},
  pages = {118758--118785},
  doi = {10.52202/085713-3578}
}

@misc{yang2026programbench,
  author = {Yang, John and Lieret, Kilian and Ma, Jeffrey and Thakkar, Parth and Pedchenko, Dmitrii and Sootla, Sten and McMilin, Emily and Yin, Pengcheng and Hou, Rui and Synnaeve, Gabriel and Yang, Diyi and Press, Ofir},
  title = {{ProgramBench}: Can Language Models Rebuild Programs From Scratch?},
  year = {2026},
  howpublished = {arXiv:2605.03546},
  doi = {10.48550/arXiv.2605.03546},
  note = {Accepted at NeurIPS 2026; to appear}
}

@inproceedings{yao2023react,
  author = {Yao, Shunyu and Zhao, Jeffrey and Yu, Dian and Du, Nan and Shafran, Izhak and Narasimhan, Karthik and Cao, Yuan},
  title = {{ReAct}: Synergizing Reasoning and Acting in Language Models},
  booktitle = {The Eleventh International Conference on Learning Representations},
  year = {2023}
}

@misc{zhan2026agensh,
  author = {Zhan, Zhihao and Song, Ting and Dong, Li and Huang, Shaohan and Lian, Jianxun and Xia, Yan and Wei, Furu},
  title = {{Agensh}: Scaling Organizational Intelligence to 1,024 Agents},
  year = {2026},
  howpublished = {arXiv:2609.26781},
  doi = {10.48550/arXiv.2609.26781}
}

@inproceedings{zhang2025gdesigner,
  author = {Zhang, Guibin and Yue, Yanwei and Sun, Xiangguo and Wan, Guancheng and Yu, Miao and Fang, Junfeng and Wang, Kun and Chen, Tianlong and Cheng, Dawei},
  title = {{G-Designer}: Architecting Multi-Agent Communication Topologies via Graph Neural Networks},
  booktitle = {Proceedings of the 42nd International Conference on Machine Learning},
  series = {Proceedings of Machine Learning Research},
  year = {2025},
  volume = {267},
  pages = {76678--76692},
  publisher = {PMLR}
}
\endgroup

\clearpage

\appendix

\begin{center}{\renewcommand{\baselinestretch}{1.1}\large\bfseries Supplementary File of ``\manuscripttitle''\par}\end{center}

This supplementary file is organized as follows. Section~\ref{app:theory} presents model identifiability, the assumptions and results for statistical estimation and optimization, and their proofs. Section~\ref{app:details} provides additional details of the LongBench v2 Multi-Doc QA experiment, covering the experimental design, three guidelines, network counts, supplementary results, review outcomes, numerical estimation, and computational resources. Section~\ref{app:fateh} presents the corresponding FATE-H details, including proof verification and the scope of the retained runs. Section~\ref{app:deepswe} presents the DeepSWE10 details, including patch evaluation and the results for each task.

\numberwithin{equation}{section}
\numberwithin{figure}{section}
\numberwithin{table}{section}
\numberwithin{theorem}{section}
\numberwithin{assumption}{section}
\numberwithin{proposition}{section}
\numberwithin{corollary}{section}
\numberwithin{lemma}{section}
\numberwithin{remark}{section}
\renewcommand{\theHequation}{\theHsection.\arabic{equation}}
\renewcommand{\theHfigure}{\theHsection.\arabic{figure}}
\renewcommand{\theHtable}{\theHsection.\arabic{table}}

\section{Theoretical Development}
\label{app:theory}

This section gives the theoretical results for the latent space model and its estimation. Section~\ref{sec:identifiability} establishes model identifiability up to a common orthogonal transformation of the latent positions. Section~\ref{sec:theory} defines the estimation targets and assumptions and states the statistical error bounds. Section~\ref{sec:optimization-results} gives the warm-started updates and the error bounds for optimization and computed estimates. Section~\ref{sec:lemmas} develops the supporting lemmas. Sections~\ref{sec:statproof}--\ref{sec:computedproof} prove the statistical and optimization results and the corollary for computed estimates.

\subsection{Identifiability}\label{sec:identifiability}

As in \citet{fan2024generalized}, identifiability follows from equal complete linear predictor matrices.

\begin{proposition}[Identifiability]\label{prop:identifiability}
Fix a checkpoint $t$. Consider $\Theta^{(t)}=(\bm{\alpha}^{(t)},\allowbreak\bm{\beta}^{(t)},\allowbreak\bm{\gamma}^{(t)},\allowbreak\bm{Z}^{(t)},\allowbreak\bm{W}^{(t)})$ and $\Theta^{\prime(t)}=(\bm{\alpha}^{\prime(t)},\allowbreak\bm{\beta}^{\prime(t)},\allowbreak\bm{\gamma}^{\prime(t)},\allowbreak\bm{Z}^{\prime(t)},\allowbreak\bm{W}^{\prime(t)})$, where $\bm{\alpha}^{(t)},\bm{\beta}^{(t)}\in\mathbb{R}^n$, $\bm{\gamma}^{(t)}\in\mathbb{R}^m$, $\bm{Z}^{(t)}\in\mathbb{R}^{n\times d}$, $\bm{W}^{(t)}\in\mathbb{R}^{m\times d}$, and the primed parameters have the same dimensions. Suppose $\rank(\bm{Z}^{(t)})=d$ and both tuples satisfy:
\begin{equation*}
    \bm{1}_n^\top\bm{Z}^{(t)}=\bm{1}_n^\top\bm{Z}^{\prime(t)}=\bm{0}^\top,\qquad\bm{1}_m^\top\bm{W}^{(t)}=\bm{1}_m^\top\bm{W}^{\prime(t)}=\bm{0}^\top,\qquad\bm{1}_n^\top\bm{\beta}^{(t)}=\bm{1}_n^\top\bm{\beta}^{\prime(t)}=0.
\end{equation*}
If their complete linear predictor matrices are equal, that is:
\begin{equation}
  \begin{aligned}
    \bm{\alpha}^{(t)}\bm{1}_n^\top+\bm{1}_n(\bm{\alpha}^{(t)})^\top+\bm{Z}^{(t)}(\bm{Z}^{(t)})^\top=\bm{\alpha}^{\prime(t)}\bm{1}_n^\top+\bm{1}_n(\bm{\alpha}^{\prime(t)})^\top+\bm{Z}^{\prime(t)}(\bm{Z}^{\prime(t)})^\top,
  \end{aligned}
  \label{eq:identifiability-A}
\end{equation}
\begin{equation}
  \begin{aligned}
    \bm{\beta}^{(t)}\bm{1}_m^\top+\bm{1}_n(\bm{\gamma}^{(t)})^\top+\bm{Z}^{(t)}(\bm{W}^{(t)})^\top=\bm{\beta}^{\prime(t)}\bm{1}_m^\top+\bm{1}_n(\bm{\gamma}^{\prime(t)})^\top+\bm{Z}^{\prime(t)}(\bm{W}^{\prime(t)})^\top,
  \end{aligned}
  \label{eq:identifiability-B}
\end{equation}
then there is an orthogonal matrix $\bm{Q}^{(t)}\in\mathbb{R}^{d\times d}$ such that:
\begin{equation*}
  (\bm{\alpha}^{\prime(t)},\bm{\beta}^{\prime(t)},\bm{\gamma}^{\prime(t)})=(\bm{\alpha}^{(t)},\bm{\beta}^{(t)},\bm{\gamma}^{(t)}), \qquad\bm{Z}^{\prime(t)}=\bm{Z}^{(t)}\bm{Q}^{(t)},\qquad\bm{W}^{\prime(t)}=\bm{W}^{(t)}\bm{Q}^{(t)}.
\end{equation*}
\end{proposition}

\begin{proof}
Right-multiplying \eqref{eq:identifiability-A} by $\bm{1}_n$ and using $\bm{1}_n^\top\bm{Z}^{(t)}=\bm{1}_n^\top\bm{Z}^{\prime(t)}=\bm{0}^\top$ gives:
\begin{equation*}
  n(\bm{\alpha}^{(t)}-\bm{\alpha}^{\prime(t)})+\{\bm{1}_n^\top(\bm{\alpha}^{(t)}-\bm{\alpha}^{\prime(t)})\}\bm{1}_n=\bm{0}.
\end{equation*}
Summing its entries yields $\bm{1}_n^\top(\bm{\alpha}^{(t)}-\bm{\alpha}^{\prime(t)})=0$. Back-substitution gives $\bm{\alpha}^{(t)}=\bm{\alpha}^{\prime(t)}$.

Left-multiplying \eqref{eq:identifiability-B} by $\bm{1}_n^\top$ and using the centering of $\bm{\beta}^{(t)},\bm{\beta}^{\prime(t)},\bm{Z}^{(t)},\bm{Z}^{\prime(t)}$ gives $n(\bm{\gamma}^{(t)}-\bm{\gamma}^{\prime(t)})^\top=\bm{0}^\top$, so $\bm{\gamma}^{(t)}=\bm{\gamma}^{\prime(t)}$. Right-multiplying the same equation by $\bm{1}_m$, then using this equality and the centering of $\bm{W}^{(t)},\bm{W}^{\prime(t)}$, gives $m(\bm{\beta}^{(t)}-\bm{\beta}^{\prime(t)})=\bm{0}$. Hence $\bm{\beta}^{(t)}=\bm{\beta}^{\prime(t)}$. Removing the equal terms in $\bm{\alpha}^{(t)},\bm{\beta}^{(t)},\bm{\gamma}^{(t)}$ from \eqref{eq:identifiability-A} and \eqref{eq:identifiability-B} leaves:
\begin{equation*}
  \bm{Z}^{(t)}(\bm{Z}^{(t)})^\top=\bm{Z}^{\prime(t)}(\bm{Z}^{\prime(t)})^\top,\quad\bm{Z}^{(t)}(\bm{W}^{(t)})^\top=\bm{Z}^{\prime(t)}(\bm{W}^{\prime(t)})^\top.
\end{equation*}

The Gram-matrix equality implies that $\bm{Z}^{(t)}$ and $\bm{Z}^{\prime(t)}$ have the same column space. Since $\bm{Z}^{(t)}$ has full column rank, there is a matrix $\bm{Q}^{(t)}\in\mathbb{R}^{d\times d}$ such that $\bm{Z}^{\prime(t)}=\bm{Z}^{(t)}\bm{Q}^{(t)}$. Substituting this relation into the Gram-matrix equality and multiplying by a left inverse of $\bm{Z}^{(t)}$ and its transpose gives $\bm{Q}^{(t)}(\bm{Q}^{(t)})^\top=\bm{I}_d$, where $\bm{I}_d$ is the identity matrix. Thus $\bm{Q}^{(t)}$ is orthogonal. Substituting the same relation into the cross-product equality and multiplying by a left inverse of $\bm{Z}^{(t)}$ gives $(\bm{W}^{(t)})^\top=\bm{Q}^{(t)}(\bm{W}^{\prime(t)})^\top$. Therefore $\bm{W}^{\prime(t)}=\bm{W}^{(t)}\bm{Q}^{(t)}$.
\end{proof}

\subsection{Statistical Framework and Error Bounds}
\label{sec:theory}

We study estimation of cumulative means and the accuracy of warm-started gradient updates, which start each checkpoint from the preceding estimate. The statistical results control error relative to the current estimation target, defined below. The optimization results control error relative to the current exact minimizer, that is, a minimizer of the joint objective \eqref{eq:main-likelihood} over the feasible parameter set in Assumption~\ref{as:geom}. 

Throughout, $n\ge3$, $m\ge1$, and $d\ge1$ is fixed. For the work records indexed by $r$ in Section~\ref{sec:Network-Data-Collection}, let $\FF_r$ denote the information available up to and including work record $r$. With increments $\Delta \bm{A}_r,\Delta \bm{B}_r$ and the cumulative number $S_t$ of work records, we have $\bm{A}^{(t)}=\sum_{r=1}^{S_t}\Delta \bm{A}_r$ and $\bm{B}^{(t)}=\sum_{r=1}^{S_t}\Delta \bm{B}_r$. The increments are $\FF_r$-measurable binary matrices. The finite random counts $S_t$ are nondecreasing in $t$. We define the estimation targets by summing the expected network increments conditional on the preceding history:
\begin{equation*}
  \theta_{\star,i\ell}^{A,(t)}=\sum_{r=1}^{S_t}\E\{(\Delta \bm{A}_r)_{i\ell}\mid\FF_{r-1}\},\qquad \theta_{\star,ik}^{B,(t)}=\sum_{r=1}^{S_t}\E\{(\Delta \bm{B}_r)_{ik}\mid\FF_{r-1}\}.
\end{equation*}
$\theta_{\star,i\ell}^{A,(t)}$ and $\theta_{\star,ik}^{B,(t)}$ can depend on the realized history and generally differ from the entries of $\E \bm{A}^{(t)}$ and $\E \bm{B}^{(t)}$. Write the true parameter tuple as $\Theta_{\star}^{(t)}=(\bm{\alpha}_{\star}^{(t)},\allowbreak \bm{\beta}_{\star}^{(t)},\allowbreak \bm{\gamma}_{\star}^{(t)},\allowbreak \bm{Z}_{\star}^{(t)},\allowbreak \bm{W}_{\star}^{(t)})$, where $\log\theta_{\star,i\ell}^{A,(t)}=\alpha_{\star,i}^{(t)}+\alpha_{\star,\ell}^{(t)}+(\bm{z}_{\star,i}^{(t)})^\top \bm{z}_{\star,\ell}^{(t)}$ and $\log\theta_{\star,ik}^{B,(t)}=\beta_{\star,i}^{(t)}+\gamma_{\star,k}^{(t)}+(\bm{z}_{\star,i}^{(t)})^\top \bm{w}_{\star,k}^{(t)}$. The parameters satisfy the centering constraints $\one_n^\top\bm{Z}_{\star}^{(t)}=\bm{0}^\top$, $\one_m^\top\bm{W}_{\star}^{(t)}=\bm{0}^\top$, and $\one_n^\top\bm{\beta}_{\star}^{(t)}=0$. Section~\ref{sec:lemmas} shows that the parameters can be centered without changing $\theta_{\star,i\ell}^{A,(t)}$ and $\theta_{\star,ik}^{B,(t)}$.

\begin{assumption}\label{as:model}
    There is a constant $K\ge1$, independent of $n,m,T$, such that, for every work record $r$, $1\le\sum_{i<\ell}(\Delta\bm{A}_r)_{i\ell}+\sum_{i,k}(\Delta\bm{B}_r)_{ik}\le K$.
\end{assumption}

Assumption~\ref{as:model} bounds the contribution of each work record to $\Delta\bm{A}_r$ and $\Delta\bm{B}_r$. A work record may involve multiple collaborators and task types but updates at least one and at most $K$ worker--worker and worker--task-type pairs in total. The lower bound holds because each work increases $B_{ik}$ for its task type (Section~\ref{sec:Network-Data-Collection}).

\begin{assumption}\label{as:geom}
    There are constants $R_z,R_w,\kappa>0$, independent of $n,m,T$, such that, for every $t\in[T]$, the centered true parameters satisfy $\max_i\norm{\bm{z}_{\star,i}^{(t)}}_2\le R_z$, $\max_k\norm{\bm{w}_{\star,k}^{(t)}}_2\le R_w$, and $\lambda_{\min}\{(\bm{Z}_{\star}^{(t)})^\top \bm{Z}_{\star}^{(t)}/n\}\ge\kappa$. The feasible parameter set consists of tuples $\widetilde{\Theta}$ with finite node effects that satisfy the same centering constraints and $\max_i\norm{\widetilde{\bm{z}}_i}_2\le 2R_z$, $\max_k\norm{\widetilde{\bm{w}}_k}_2\le 2R_w$.
\end{assumption}
Assumption~\ref{as:geom} bounds $\norm{\bm{z}_{\star,i}^{(t)}}_2$ and $\norm{\bm{w}_{\star,k}^{(t)}}_2$. The lower bound on $\lambda_{\min}\{(\bm{Z}_{\star}^{(t)})^\top \bm{Z}_{\star}^{(t)}/n\}$ ensures that the rows of $\bm{Z}_{\star}^{(t)}$ vary sufficiently in every latent dimension. The larger radii $2R_z$ and $2R_w$ define the radius constraints, that is, the norm bounds of the feasible parameter set, and keep $\bm{z}_{\star,i}^{(t)}$ and $\bm{w}_{\star,k}^{(t)}$ strictly inside the radius constraints.

For a feasible $\widetilde{\Theta}$, let $\theta_{i\ell}^{A}(\widetilde{\Theta})$ and $\theta_{ik}^{B}(\widetilde{\Theta})$ denote the cumulative means in \eqref{eq:joint} evaluated at $\widetilde{\Theta}$. We consider feasible estimates $\widehat{\Theta}^{(t)}$ satisfying:
\begin{equation}
  \label{eq:opt} 
  \mathcal{L}^{(t)}(\widehat{\Theta}^{(t)})\le\inf_{\widetilde{\Theta}\ \text{feasible}}\mathcal{L}^{(t)}(\widetilde{\Theta})+\varepsilon_t,\qquad 0\le\varepsilon_t\le\varepsilon_0,
\end{equation}
where $\varepsilon_t$ is the optimization error allowed at checkpoint $t$ and $\varepsilon_0$ is a fixed constant.

For $\Theta=(\bm{\alpha},\bm{\beta},\bm{\gamma},\bm{Z},\bm{W})$ and $\Theta'=(\bm{\alpha}',\ldots,\bm{W}')$, define the normalized distance:
\begin{equation}
  \label{eq:distance} \begin{aligned}
    \dist^2(\Theta,\Theta')&={} \frac{\norm{\bm{\alpha}-\bm{\alpha}'}_2^2+\norm{\bm{\beta}-\bm{\beta}'}_2^2}{n} +\frac{\norm{\bm{\gamma}-\bm{\gamma}'}_2^2}{m}\\
     &+\min_{\bm{Q}^\top\bm{Q}=\bm{I}_d}\left\{ \frac{\norm{\bm{Z}-\bm{Z}'\bm{Q}}_{\fro}^2}{n} +\frac{\norm{\bm{W}-\bm{W}'\bm{Q}}_{\fro}^2}{m}\right\},
  \end{aligned}
\end{equation}
where $\norm{\cdot}_{\fro}$ is the Frobenius norm and $\bm{I}_d$ is the $d\times d$ identity matrix. The square root $\dist$ is a metric on equivalence classes under a common orthogonal transformation. The same orthogonal $d\times d$ matrix $\bm{Q}$ acts on $\bm{Z},\bm{W}$, and may differ across checkpoints. We write $\Theta \bm{Q}=(\bm{\alpha},\bm{\beta},\bm{\gamma},\bm{Z}\bm{Q},\bm{W}\bm{Q})$.

\begin{assumption}\label{as:scale}
There are constants $0<c_0\le C_0<\infty$, independent of $n,m,T$, such that, for every $t\in[T]$ and all workers $i$ and task types $k$:
\begin{equation}
  \label{eq:expectedmargins} c_0t\le\sum_{\ell\ne i}\theta_{\star,i\ell}^{A,(t)}\le C_0t,\quad c_0t\le\sum_k\theta_{\star,ik}^{B,(t)}\le C_0t,\quad c_0nt/m\le\sum_i\theta_{\star,ik}^{B,(t)}\le C_0nt/m.
\end{equation}
\end{assumption}

Assumption~\ref{as:scale} specifies the scale of cumulative network activity. For each worker $i$, $\sum_{\ell\ne i}\theta_{\star,i\ell}^{A,(t)}$ and $\sum_k\theta_{\star,ik}^{B,(t)}$ are of order $t$. Therefore $\sum_{i,k}\theta_{\star,ik}^{B,(t)}$ is of order $nt$, and each task type $k$ is required to have $\sum_i\theta_{\star,ik}^{B,(t)}$ of order $nt/m$.

Throughout the theoretical results and proofs in Section~\ref{app:theory}, $C,c,n_0$ denote positive constants that may depend only on $d,K,c_0,C_0,R_z,R_w,\kappa,\varepsilon_0$. Under Assumption~\ref{as:geom}, Lemma~\ref{lem:meansfrommargins} shows that Assumption~\ref{as:scale} implies the following bounds for $n\ge n_0$:
\begin{equation}
  \label{eq:scale} c\frac{t}{n}\le\theta_{\star,i\ell}^{A,(t)}\le C\frac{t}{n}\quad(i<\ell),\qquad c\frac{t}{m}\le\theta_{\star,ik}^{B,(t)}\le C\frac{t}{m}\quad(i,k).
\end{equation}
Assumptions~\ref{as:model} and~\ref{as:scale} imply that the cumulative number of work records satisfies $S_t\asymp nt$, that is, $cnt\le S_t\le Cnt$. The proof of $S_t\asymp nt$ is given at the beginning of Section~\ref{sec:statproof}.

\begin{theorem}\label{thm:stat}
    Suppose Assumptions~\ref{as:model}--\ref{as:scale} hold and $n\ge n_0$. For every $\delta\in(0,1)$, there is an event of probability at least $1-\delta$ on which all checkpoints with $t\ge C(1+m/n)\{\log((2n+m)/\delta)+\log\log(e^e+t)\}$ admit exact minimizers, and every feasible estimate with \eqref{eq:opt} obeys:
    \begin{equation}
      \label{eq:statrate} 
      \dist^2(\widehat{\Theta}^{(t)},\Theta_{\star}^{(t)})\le C\frac{n+m}{nt}\left\{\log\frac{2n+m}{\delta}+\log\log(e^e+t)\right\}.
    \end{equation}
    \begin{equation}
      \label{eq:lograte} 
      \begin{aligned}
      &\frac{1}{n^2}\sum_{i<\ell}\left\{\log\frac{\theta_{i\ell}^{A}(\widehat{\Theta}^{(t)})}{\theta_{\star,i\ell}^{A,(t)}}\right\}^2+\frac{1}{nm}\sum_{i,k}\left\{\log\frac{\theta_{ik}^{B}(\widehat{\Theta}^{(t)})}{\theta_{\star,ik}^{B,(t)}}\right\}^2\\
      &\qquad\le C\frac{n+m}{nt}\left\{\log\frac{2n+m}{\delta}+\log\log(e^e+t)\right\}.
      \end{aligned}
    \end{equation}
    For a fixed $a\in(0,1)$, let $C_a$ denote a positive constant that may additionally depend on $a$. For any fixed $T$ satisfying $T\ge C_a(1+m/n)\log\{(2n+m)/\delta\}$, there is an event of probability at least $1-\delta$ on which the maxima of $\dist^2(\widehat{\Theta}^{(t)},\Theta_{\star}^{(t)})$ and the left-hand side of \eqref{eq:lograte} over $\lceil aT\rceil\le t\le T$ are both at most $C_a(n+m)\log\{(2n+m)/\delta\}/(nT)$.
\end{theorem}
Theorem~\ref{thm:stat} quantifies the accuracy of the statistical summaries estimated from $\bm{A}^{(t)}$ and $\bm{B}^{(t)}$. Bound~\eqref{eq:statrate} controls the average squared errors of $\widehat{\bm{\alpha}}^{(t)},\widehat{\bm{\beta}}^{(t)},\widehat{\bm{\gamma}}^{(t)}$ and $\widehat{\bm{Z}}^{(t)},\widehat{\bm{W}}^{(t)}$, with $\widehat{\bm{Z}}^{(t)},\widehat{\bm{W}}^{(t)}$ compared up to a common orthogonal transformation. Bound~\eqref{eq:lograte} controls the average squared errors of the corresponding log-means, that is, the logarithms of the cumulative means, in both networks. Together, Bounds~\eqref{eq:statrate} and~\eqref{eq:lograte} provide statistical guarantees for the estimated parameters and for the participation and collaboration patterns that the estimated parameters describe. Apart from logarithmic factors, both error bounds scale as $(n+m)/(nt)=1/t+m/(nt)$. Under Assumption~\ref{as:scale}, each worker accumulates activity of order $t$, while each task type accumulates participation of order $nt/m$. The two terms in the rate are the reciprocals of the two activity scales. When $m=O(n)$, the rate is of order $1/t$. When $m\gg n$, the term $m/(nt)$ dominates, because less participation data is available for each task type. Thus, a larger number $m$ of task types requires more work records to obtain the same error guarantee. For fixed $n$ and $m$, both bounds tend to zero as $t$ increases.
\begin{remark}
    The first conclusion of Theorem~\ref{thm:stat} holds simultaneously at all checkpoints satisfying the stated lower bound on $t$, which we call the information threshold, and the simultaneous control adds the $\log\log(e^e+t)$ term. For a fixed $T$, the second conclusion bounds the largest error over $\lceil aT\rceil\le t\le T$ without the $\log\log(e^e+t)$ term and therefore gives a bound of order $(n+m)/(nT)$, up to logarithmic factors, uniformly over $\lceil aT\rceil\le t\le T$.
\end{remark}

\subsection{Warm-Started Optimization and Computed Estimates}
\label{sec:optimization-results}

Let $\widehat{\Theta}_\star^{(t)}$ denote an exact minimizer of $\mathcal{L}^{(t)}$ over the feasible parameter set, whose existence at the checkpoints covered by Theorem~\ref{thm:stat} is established by Lemma~\ref{lem:fitmeans}. We analyze the optimization updates from a checkpoint $t_0$ at which $\widehat{\Theta}^{(t_0)}$ satisfies \eqref{eq:opt}. The required lower bound on $t_0$ is specified in Theorem~\ref{thm:optimization}. For $t\ge2$, define the checkpoint increments $\Delta\bm{A}^{(t)}=\bm{A}^{(t)}-\bm{A}^{(t-1)}$ and $\Delta\bm{B}^{(t)}=\bm{B}^{(t)}-\bm{B}^{(t-1)}$. $\Delta\bm{A}^{(t)}$ and $\Delta\bm{B}^{(t)}$ aggregate the contributions of the work records added between checkpoints $t-1$ and $t$.

\begin{assumption}\label{as:opt}
    There is a constant $K_1>0$, independent of $n,m,T$, such that, for every checkpoint $t=t_0+1,\ldots,T$:
    \begin{equation}
      \label{eq:batch} \frac1n\sum_i\left\{\sum_{\ell\ne i}\Delta A_{i\ell}^{(t)}+\sum_k\Delta B_{ik}^{(t)}\right\}^{\!2}+\frac m{n^2}\sum_k\left\{\sum_i\Delta B_{ik}^{(t)}\right\}^{\!2}\le K_1^2\left(1+\frac mn\right).
    \end{equation}
\end{assumption}
Assumption~\ref{as:model} bounds each work record's contribution, whereas Assumption~\ref{as:opt} controls $\Delta\bm{A}^{(t)}$ and $\Delta\bm{B}^{(t)}$. The two terms in \eqref{eq:batch} average squared increments over workers and task types, with the task-type term scaled by $m/n$ to reflect the count scale $nt/m$ of each task type under Assumption~\ref{as:scale}. The factor $1+m/n$ accommodates sparse integer counts: when $m\ge n$, one new participation per worker, each assigned to a distinct task type, with no new collaborations gives the two terms $1$ and $m/n$. 

Following \citet{he2025semiparametric}, we prescribe parameter-specific step sizes for the theoretical analysis, scaled according to the cumulative count scales in Assumption~\ref{as:scale}. Let $\bm{P}$ denote the orthogonal projection enforcing the centering constraints: it subtracts the column means from $\bm{Z}$ and $\bm{W}$ and the mean from $\bm{\beta}$, while leaving $\bm{\alpha}$ and $\bm{\gamma}$ unchanged. Starting from a feasible estimate $\widehat{\Theta}^{(t_0)}$ satisfying \eqref{eq:opt}, initialize each subsequent checkpoint at $\widehat{\Theta}^{(t,0)}=\widehat{\Theta}^{(t-1)}$. For a step-size constant $\eta>0$ and an iteration count $J$, both independent of $n,m,T$, perform the following updates for $j=0,\ldots,J-1$, with all gradients evaluated at $\widehat{\Theta}^{(t,j)}$:
\begin{equation}
  \label{eq:algorithm} \widehat{\Theta}^{(t,j+1)}=\bm{P}\!\left[\widehat{\Theta}^{(t,j)}-\frac{\eta}{t}\left(\nabla_{\bm{\alpha}}\mathcal{L}^{(t)},\nabla_{\bm{\beta}}\mathcal{L}^{(t)},\frac mn\nabla_{\bm{\gamma}}\mathcal{L}^{(t)},\nabla_{\bm{Z}}\mathcal{L}^{(t)},\frac mn\nabla_{\bm{W}}\mathcal{L}^{(t)}\right)^\top\right].
\end{equation}
\begin{remark}
Following the numerical implementation of \citet{he2025semiparametric}, we use adaptive Barzilai--Borwein step sizes in numerical optimization. Our implementation in Section~\ref{app:fit} uses an adaptive scalar step size, obtained from Barzilai--Borwein step sizes and Armijo backtracking, to minimize the penalized objective \eqref{eq:empirical-objective}. 
\end{remark}

Set $\widehat{\Theta}^{(t)}=\widehat{\Theta}^{(t,J)}$, which we call the computed estimate at checkpoint $t$. Theorem~\ref{thm:optimization} analyzes the step sizes and fixed iteration count in \eqref{eq:algorithm}.

\begin{theorem}\label{thm:optimization}
    Suppose Assumptions~\ref{as:model}--\ref{as:opt} hold and $n\ge n_0$. For fixed $\delta\in(0,1)$, choose a deterministic integer $t_0<T$ satisfying $t_0\ge C(n+m)^2\log\{(2n+m)/\delta\}/n$, and let $\widehat{\Theta}^{(t_0)}$ be a feasible estimate satisfying \eqref{eq:opt}. There exist $\eta>0$, an integer $J\ge1$, and $q\in(0,1)$, independent of $n,m,T$, such that, with probability at least $1-\delta$, all iterates of \eqref{eq:algorithm} remain feasible and the following bounds hold simultaneously for $t_0\le t\le T$:
    \begin{equation}
      \label{eq:optrate} 
      \dist^2(\widehat{\Theta}^{(t)},\widehat{\Theta}_\star^{(t)})\le C\frac{\varepsilon_0}{nt_0}q^{2(t-t_0)}+C\frac{n+m}{nt^2},
    \end{equation}
    \begin{equation}
      \label{eq:gapderived} 
      0\le\mathcal{L}^{(t)}(\widehat{\Theta}^{(t)})-\mathcal{L}^{(t)}(\widehat{\Theta}_\star^{(t)})\le C(1+\varepsilon_0),
    \end{equation}
    where the constants may additionally depend on $K_1$.
\end{theorem}
Theorem~\ref{thm:optimization} shows that a fixed number $J$ of updates per checkpoint suffices to maintain a bounded optimization error. In \eqref{eq:optrate}, the contribution of the initial error decays geometrically, while the term $(n+m)/(nt^2)$ accounts for changes in $\widehat{\Theta}_\star^{(t)}$ as new work records are added. Bound~\eqref{eq:gapderived} ensures that $\widehat{\Theta}^{(t)}$ satisfies a fixed optimization-error bound throughout $t_0\le t\le T$. Combining Bound~\eqref{eq:gapderived} with Theorem~\ref{thm:stat} gives the following corollary, which bounds the statistical error of $\widehat{\Theta}^{(t)}$ relative to $\Theta_\star^{(t)}$.

\begin{corollary}\label{cor:computed}
    Fix $\delta\in(0,1)$ and $a\in(0,1)$. Suppose the conditions of Theorem~\ref{thm:optimization} hold and $\lceil aT\rceil\ge t_0$. Then $\widehat{\Theta}^{(t)}$ produced by \eqref{eq:algorithm}, with $\eta$ and $J$ chosen as in Theorem~\ref{thm:optimization}, satisfies, with probability at least $1-\delta$:
    \begin{equation}
      \label{eq:computedtail} 
      \max_{\lceil aT\rceil\le t\le T}\dist^2(\widehat{\Theta}^{(t)},\Theta_\star^{(t)})\le C_a\frac{n+m}{nT}\log\frac{2n+m}{\delta}.
    \end{equation}
    The same bound holds for the maximum of the normalized squared log-mean error on the left-hand side of \eqref{eq:lograte} over these checkpoints.
\end{corollary}
Corollary~\ref{cor:computed} shows that the estimates obtained with a fixed number of updates per checkpoint attain the statistical error rate in Theorem~\ref{thm:stat}. Under Assumptions~\ref{as:model}--\ref{as:opt} and condition~\eqref{eq:opt} at $t_0$, with constants uniform along a sequence of $(n,m,T,t_0)$, a sufficient condition for uniform consistency in probability over $\lceil aT\rceil\le t\le T$ is $\{(n+m)^2\log(2n+m)\}/(nt_0)\to0$ with $t_0\le\lceil aT\rceil$. Section~\ref{sec:computedproof} gives the derivation.

Each gradient evaluation in \eqref{eq:algorithm} costs $O\{n(n+m)d\}$ operations when all $\theta_{i\ell}^{A}$ and $\theta_{ik}^{B}$ are computed directly. The $J$ updates at each of the $T-t_0$ subsequent checkpoints therefore cost $O\{J(T-t_0)n(n+m)d\}$ operations in total, excluding initialization.

\subsection{Supporting Lemmas and Proofs}\label{sec:lemmas}

We first verify invariance under centering. Let $\bar{\bm{z}}=n^{-1}\sum_i \bm{z}_i$, $\bar{\bm{w}}=m^{-1}\sum_k \bm{w}_k$, and $\bar{b}=n^{-1}\sum_i(\beta_i+\bm{z}_i^\top\bar{\bm{w}})$. Set:
\begin{equation*}
  \begin{aligned}
     &\bm{z}_i^c=\bm{z}_i-\bar{\bm{z}},\qquad \bm{w}_k^c=\bm{w}_k-\bar{\bm{w}},\qquad \alpha_i^c=\alpha_i+\bm{z}_i^\top\bar{\bm{z}}-\tfrac12\norm{\bar{\bm{z}}}_2^2,\\
     &\beta_i^c=\beta_i+\bm{z}_i^\top\bar{\bm{w}}-\bar{b},\qquad \gamma_k^c=\gamma_k+\bar{\bm{z}}^\top \bm{w}_k-\bar{\bm{z}}^\top\bar{\bm{w}}+\bar{b}. 
  \end{aligned}
\end{equation*}
$\bm{z}_i^c$, $\bm{w}_k^c$, and $\bm{\beta}^c$ have zero sums. Expanding $(\bm{z}_i^c)^\top \bm{z}_\ell^c$ and $(\bm{z}_i^c)^\top \bm{w}_k^c$ gives:
\begin{equation*}
  \alpha_i^c+\alpha_\ell^c+(\bm{z}_i^c)^\top \bm{z}_\ell^c =\alpha_i+\alpha_\ell+\bm{z}_i^\top \bm{z}_\ell,\qquad \beta_i^c+\gamma_k^c+(\bm{z}_i^c)^\top \bm{w}_k^c =\beta_i+\gamma_k+\bm{z}_i^\top \bm{w}_k.
\end{equation*}
Thus all $\theta_{i\ell}^{A}$ and $\theta_{ik}^{B}$ are preserved. Subtracting the row average increases a uniform row-norm bound by at most a factor of two. Assumption~\ref{as:geom} is imposed on the centered $\bm{Z}_\star^{(t)}$ and $\bm{W}_\star^{(t)}$, including $\lambda_{\min}\{(\bm{Z}_{\star}^{(t)})^\top \bm{Z}_{\star}^{(t)}/n\}$.

\begin{lemma}\label{lem:meansfrommargins}
Fix $u>0$. With latent row norms bounded by $2R_z,2R_w$, if every row sum of $(\theta_{i\ell}^{A})$ and $(\theta_{ik}^{B})$ lies between $c_0u$ and $C_0u$, and every column sum of $(\theta_{ik}^{B})$ lies between $c_0nu/m$ and $C_0nu/m$, then for all $n\ge n_0$, where $n_0$ is a sufficiently large fixed constant, $cu/n\le\theta_{i\ell}^{A}\le Cu/n$ and $cu/m\le\theta_{ik}^{B}\le Cu/m$. Constants depend only on the radii and marginal bounds. $\bm{Z},\bm{W}$ may be rank deficient.
\end{lemma}
\begin{proof}
In the $\bm{A}$ block, with $x_i=e^{\alpha_i}$, $x_+=\sum_i x_i$, $s_i=\sum_{\ell\ne i}\theta_{i\ell}^{A}$, bounded positions give:
\begin{equation}
  \label{eq:xS} e^{-4R_z^2}x_i(x_+-x_i)\le s_i \le e^{4R_z^2}x_i(x_+-x_i),\qquad c_0u\le s_i\le C_0u.
\end{equation}
If the largest $x_i$ were at least $x_+/2$, then:
\begin{equation*}
  (n-1)c_0u\le\sum_{j\ne i}s_j \le e^{4R_z^2}x_+(x_+-x_i) \le2e^{4R_z^2}x_i(x_+-x_i) \le2e^{8R_z^2}C_0u.
\end{equation*}
Taking $n_0>1+2e^{8R_z^2}C_0/c_0$ rules out this possibility. Hence $x_+/2<x_+-x_i\le x_+$ for every $i$. Summing \eqref{eq:xS} gives $x_+^2\asymp nu$. The individual bounds then give $x_ix_+\asymp u$, so $x_i\asymp\sqrt{u/n}$. Multiplying by the bounded inner-product exponential proves the $\bm{A}$ claim.

For the $\bm{B}$ block, set $x_i=e^{\beta_i}$, $y_k=e^{\gamma_k}$, $x_+=\sum_i x_i$, and $y_+=\sum_k y_k$. Write $s_i,c_k,s_+$ for the row, column, and total sums, and $a=e^{-4R_zR_w}$, $b=e^{4R_zR_w}$. Bounded inner products give:
\begin{equation*}
  ax_iy_+\le s_i\le bx_iy_+,\qquad ay_kx_+\le c_k\le by_kx_+,\qquad ax_+y_+\le s_+\le bx_+y_+.
\end{equation*}
Using $x_iy_k=(x_iy_+)(y_kx_+)/(x_+y_+)$ and $s_i\asymp u$, $c_k\asymp nu/m$, $s_+\asymp nu$ yields:
\begin{equation*}
  \frac{a^2}{b^2}\frac{s_ic_k}{s_+}\le\theta_{ik}^{B}\le\frac{b^2}{a^2}\frac{s_ic_k}{s_+},\qquad\theta_{ik}^{B}\asymp u/m.
\end{equation*}
\end{proof}

The martingales below are indexed by work record $r$. Checkpoint $t$ enters through the cumulative index $S_t$. Assumption~\ref{as:scale} gives:
\begin{equation}
  \label{eq:marginalscale} \sum_{\ell\ne i}\theta_{\star,i\ell}^{A,(t)}\asymp t,\qquad \sum_k\theta_{\star,ik}^{B,(t)}\asymp t,\qquad \sum_i\theta_{\star,ik}^{B,(t)}\asymp nt/m.
\end{equation}
The entrywise bounds in \eqref{eq:scale} follow from Lemma~\ref{lem:meansfrommargins} and are used to bound curvature. Define the weighted matrix increment and its centered cumulative sum:
\begin{equation}
  \label{eq:H} \bm{H}_r=\left[\frac{\sqrt n}{2}\Delta \bm{A}_r,\ \sqrt m\,\Delta \bm{B}_r\right],\qquad \bm{M}_q=\sum_{r=1}^q\{\bm{H}_r-\E(\bm{H}_r\mid\FF_{r-1})\}.
\end{equation}
Then $\bm{M}_q$ is an $n\times(n+m)$ matrix martingale. Write $\norm{\cdot}_{\op}$ for the operator norm. The factor $1/2$ compensates for symmetric duplication in $\bm{A}$, while $\sqrt n,\sqrt m$ scale the two blocks. At checkpoint $t$, the first block of $\bm{M}_{S_t}$ is $\sqrt n/2$ times $\bm{A}^{(t)}-(\theta_{\star,i\ell}^{A,(t)})$. The second block is the corresponding $\bm{B}$ residual scaled by $\sqrt m$.

\begin{lemma}\label{lem:concentration}
Under Assumptions~\ref{as:model} and~\ref{as:scale}, for every $\delta\in(0,1)$, the following holds simultaneously at all checkpoints under study with probability at least $1-\delta$:
\begin{equation}
  \label{eq:matrixnoise} \norm{\bm{M}_{S_t}}_{\op}^2 \le C(n+m)\left[ t\left\{\log\frac{2n+m}{\delta}+\log\log(e^e+t)\right\} +\left\{\log\frac{2n+m}{\delta}+\log\log(e^e+t)\right\}^{\!2}\right].
\end{equation}
Whenever the information threshold in Theorem~\ref{thm:stat} holds, the same event also gives:
\begin{equation}
  \label{eq:obsmargins} ct\le\sum_{\ell\ne i}A_{i\ell}^{(t)}\le Ct,\qquad ct\le\sum_kB_{ik}^{(t)}\le Ct,\qquad c\frac{nt}{m}\le\sum_iB_{ik}^{(t)}\le C\frac{nt}{m}.
\end{equation}
For fixed $T\ge1$ and $a\in(0,1)$, the following bound holds simultaneously for $\lceil aT\rceil\le t\le T$, with probability at least $1-\delta$:
\begin{equation}
  \label{eq:finiteM} \norm{\bm{M}_{S_t}}_{\op}^2 \le C(n+m)\left\{T\log\frac{2n+m}{\delta} +\log^2\frac{2n+m}{\delta}\right\}.
\end{equation}
If $T\ge C_a(1+m/n)\log\{(2n+m)/\delta\}$, \eqref{eq:obsmargins} also holds on this event.
\end{lemma}
\begin{proof}
We first bound individual increments and conditional variances. Assumption~\ref{as:model} gives at most $2K$ nonzero entries in $\bm{H}_r$, with:
\begin{equation*}
  \norm{\bm{H}_r}_{\fro}^2 =\frac n2\sum_{i<\ell}(\Delta \bm{A}_r)_{i\ell} +m\sum_{i,k}(\Delta \bm{B}_r)_{ik}\le K(n+m).
\end{equation*}
Conditional expectation preserves the corresponding operator-norm bound, so centered increments have norm at most $2\sqrt{K(n+m)}$. For any vector $\bm{x}$, applying Cauchy--Schwarz to the nonzero entries in each column gives:
\begin{equation*}
  \bm{x}^\top \bm{H}_r\bm{H}_r^\top \bm{x} \le2K\sum_i x_i^2\left\{\frac n4\sum_{\ell\ne i}(\Delta \bm{A}_r)_{i\ell} +m\sum_k(\Delta \bm{B}_r)_{ik}\right\}.
\end{equation*}
Set $\bm{X}_r=\bm{H}_r-\E(\bm{H}_r\mid\FF_{r-1})$. Conditional centering gives:
\begin{equation*}
  \begin{aligned}
    \E(\bm{X}_r\bm{X}_r^\top\mid\FF_{r-1})&=\E(\bm{H}_r\bm{H}_r^\top\mid\FF_{r-1})-\E(\bm{H}_r\mid\FF_{r-1})\E(\bm{H}_r\mid\FF_{r-1})^\top\\
    &\preceq\E(\bm{H}_r\bm{H}_r^\top\mid\FF_{r-1}).
  \end{aligned}
\end{equation*}
The same identity holds with the factors transposed. Applying the preceding column bound and its row analogue, then \eqref{eq:marginalscale}, yields:
\begin{equation*}
  \begin{aligned}
    \left\|\sum_{r\le S_t}\E(\bm{X}_r\bm{X}_r^\top\mid\FF_{r-1})\right\|_{\op}&\le2K\max_i\left\{\frac n4\sum_{\ell\ne i}\theta_{\star,i\ell}^{A,(t)}+m\sum_k\theta_{\star,ik}^{B,(t)}\right\}\le C(n+m)t,\\
    \left\|\sum_{r\le S_t}\E(\bm{X}_r^\top\bm{X}_r\mid\FF_{r-1})\right\|_{\op}&\le2K\max\left\{\frac n4\max_\ell\sum_{i\ne\ell}\theta_{\star,i\ell}^{A,(t)},\ m\max_k\sum_i\theta_{\star,ik}^{B,(t)}\right\}\le Cnt.
  \end{aligned}
\end{equation*}
For a scalar row or column contribution $0\le Y_r\le K$:
\begin{equation*}
  \begin{aligned}
    \sum_{r\le S_t}\operatorname{Var}(Y_r\mid\FF_{r-1})&\le\sum_{r\le S_t}\E(Y_r^2\mid\FF_{r-1})\\
    &\le K\sum_{r\le S_t}\E(Y_r\mid\FF_{r-1})\le\begin{cases}Ct,&\text{worker row},\\ Cnt/m,&\text{task-type column}.\end{cases}
  \end{aligned}
\end{equation*}
Define the matrix variance process:
\begin{equation*}
  V_q=\max\left\{\left\|\sum_{r\le q}\E(\bm{X}_r\bm{X}_r^\top\mid\FF_{r-1})\right\|_{\op},\ \left\|\sum_{r\le q}\E(\bm{X}_r^\top\bm{X}_r\mid\FF_{r-1})\right\|_{\op}\right\}.
\end{equation*}

The rectangular matrix Freedman inequality of \citet[Corollary 1.3]{tropp2011} states that, if every $\norm{\bm{X}_r}_{\op}$ is at most $R$, then for any $u,v>0$:
\begin{equation*}
  \PP\left\{\exists q\ge1:\ \norm{\bm{M}_q}_{\op}\ge u,\ V_q\le v\right\} \le(2n+m)\exp\left\{-\frac{u^2}{2(v+Ru/3)}\right\}.
\end{equation*}
For $x>0$ and variance budget $v$, the threshold $\sqrt{2vx}+(2R/3)x$ therefore has failure probability at most $(2n+m)e^{-x}$. This maximal event covers all work-record indices within the budget and hence applies to random $S_t$. The scalar Freedman inequality for both signs controls the $2n+m$ marginal processes of the $\bm{A},\bm{B}$ row sums and $\bm{B}$ column sums.

For the time-uniform result, partition checkpoints into $2^j\le t<2^{j+1}$, $j\ge0$. On interval $j$, the variance budgets are $C(n+m)2^{j+1}$, $C2^{j+1}$, and $C(n/m)2^{j+1}$. Allocate total failure probability at most $C\delta/(j+1)^2$ to its matrix and marginal events, with the allocation constant chosen so that the sum is at most $\delta$. Within this proof, take:
\begin{equation*}
  x_j=\log\frac{C(2n+m)(j+1)^2}{\delta} \le C\left\{\log\frac{2n+m}{\delta}+\log\log(e^e+t)\right\}.
\end{equation*}
The matrix deviation is bounded by $C\sqrt{n+m}\{\sqrt{tx_j}+x_j\}$. Squaring proves \eqref{eq:matrixnoise}. Worker row deviations are bounded by $C\{\sqrt{tx_j}+x_j\}$ and task-type column deviations by $C\{\sqrt{(nt/m)x_j}+x_j\}$. When $\min\{t,nt/m\}\ge Cx_j$, each deviation is at most a fixed fraction of the corresponding sum of $\theta_{\star,i\ell}^{A,(t)}$ or $\theta_{\star,ik}^{B,(t)}$, which gives \eqref{eq:obsmargins}. The information threshold in Theorem~\ref{thm:stat} suffices.

For the bound over $\lceil aT\rceil\le t\le T$, apply the same maximal inequality directly over all work-record indices with variance budgets $C(n+m)T$, $CT$, and $CnT/m$. Taking $x=\log\{C(2n+m)/\delta\}$ gives \eqref{eq:finiteM} and the marginal deviations. On $t\ge aT$, the corresponding sums of $\theta_{\star,i\ell}^{A,(t)}$ and $\theta_{\star,ik}^{B,(t)}$ are at least $caT$ for worker rows and $canT/m$ for task-type columns, so the stated lower bound on $T$ suffices for \eqref{eq:obsmargins}.
\end{proof}

\begin{lemma}\label{lem:fitmeans}
Suppose \eqref{eq:obsmargins} holds, $\min\{t,nt/m\}$ exceeds a fixed constant, and $n\ge n_0$. Every estimate satisfying \eqref{eq:opt} then obeys:
\begin{equation}
  \label{eq:fittedmeans} c\frac tn\le\widehat{\theta}_{i\ell}^{A,(t)}\le C\frac tn,\qquad c\frac tm\le\widehat{\theta}_{ik}^{B,(t)}\le C\frac tm.
\end{equation}
Under these conditions, the infimum in \eqref{eq:opt} is attained at a finite feasible parameter.
\end{lemma}
\begin{proof}
Let $N^{(t)}$ be one observed row sum of $\bm{A}^{(t)}$ or $\bm{B}^{(t)}$, or one column sum of $\bm{B}^{(t)}$, and let $\widehat{N}^{(t)}$ be the corresponding fitted sum. Whenever $N^{(t)}>0$, we claim:
\begin{equation}
  \label{eq:marginKL} N^{(t)}\left\{\frac{\widehat{N}^{(t)}}{N^{(t)}}-1-\log\frac{\widehat{N}^{(t)}}{N^{(t)}}\right\} \le\varepsilon_t.
\end{equation}
For row $i$ of $\bm{A}^{(t)}$, increasing only $\widehat{\alpha}_i^{(t)}$ by $h$ multiplies $\widehat{\theta}_{i\ell}^{A,(t)}$, $\ell\ne i$, by $e^h$ and changes $\mathcal{L}^{(t)}$ by $\widehat{N}^{(t)}(e^h-1)-N^{(t)}h$. Taking $h=\log(N^{(t)}/\widehat{N}^{(t)})$ decreases $\mathcal{L}^{(t)}$ by $\widehat{N}^{(t)}-N^{(t)}+N^{(t)}\log(N^{(t)}/\widehat{N}^{(t)})$. The perturbation is feasible, so the decrease is at most $\varepsilon_t$.

For row $i$ of $\bm{B}^{(t)}$, increase $\widehat{\beta}_i^{(t)}$ by $h(1-1/n)$, decrease every other $\widehat{\beta}_j^{(t)}$ by $h/n$, and increase every $\widehat{\gamma}_k^{(t)}$ by $h/n$. This preserves $\one_n^\top\widehat{\bm{\beta}}^{(t)}=0$ and increases only $\log\widehat{\theta}_{ik}^{B,(t)}$, $k\in[m]$, by $h$. For column $k$, change only $\widehat{\gamma}_k^{(t)}$. These three perturbations prove \eqref{eq:marginKL}.

Since $\inf_{u>0,\,u\notin[1/2,2]}(u-1-\log u)>1/10$, \eqref{eq:marginKL} and $N^{(t)}\ge10\varepsilon_0$ give $N^{(t)}/2\le\widehat{N}^{(t)}\le2N^{(t)}$. Lemma~\ref{lem:meansfrommargins} then yields \eqref{eq:fittedmeans} and $e^{\widehat{\alpha}_i^{(t)}}\asymp\sqrt{t/n}$, without requiring full-rank $\widehat{\bm{Z}}^{(t)}$ or positive individual counts.

For existence, $\mathcal{L}^{(t)}$ is bounded below and finite at $\Theta_\star^{(t)}$, so choose a feasible minimizing sequence. Its tail has optimization error at most $\max\{1,\varepsilon_0\}$ and satisfies the preceding bounds. For each tail element $\widetilde{\Theta}^{(t)}$, centering gives:
\begin{equation*}
  \begin{aligned}
    \widetilde{\gamma}_k^{(t)}=\frac1n\sum_i\log\theta_{ik}^{B}(\widetilde{\Theta}^{(t)}),\quad\widetilde{\beta}_i^{(t)}=\frac1m\sum_k\log\theta_{ik}^{B}(\widetilde{\Theta}^{(t)})-\frac1{nm}\sum_{j,k}\log\theta_{jk}^{B}(\widetilde{\Theta}^{(t)}),\\
    \max_i\left|\widetilde{\alpha}_i^{(t)}-\tfrac12\log(t/n)\right|+\max_i|\widetilde{\beta}_i^{(t)}|\le C,\quad\max_k\left|\widetilde{\gamma}_k^{(t)}-\log(t/m)\right|\le C.
  \end{aligned}
\end{equation*}
At fixed $t$, these bounds place the tail in a compact feasible set. Continuity gives a subsequence attaining the infimum.
\end{proof}

For parameter recovery, compare $\Theta=(\bm{\alpha},\bm{\beta},\bm{\gamma},\bm{Z},\bm{W})$ and $\Theta'=(\bm{\alpha}',\bm{\beta}',\bm{\gamma}',\bm{Z}',\bm{W}')$. Both are centered and satisfy the radius constraints, and $\lambda_{\min}\{(\bm{Z}')^\top \bm{Z}'/n\}\ge\kappa$. Define the differences of the complete linear predictor matrices:
\begin{equation}
  \label{eq:D} \begin{aligned}
     \bm{D}_A&=(\bm{\alpha}-\bm{\alpha}')\one_n^\top+\one_n(\bm{\alpha}-\bm{\alpha}')^\top+\bm{Z}\bm{Z}^\top-\bm{Z}'(\bm{Z}')^\top,\\
     \bm{D}_B&=(\bm{\beta}-\bm{\beta}')\one_m^\top+\one_n(\bm{\gamma}-\bm{\gamma}')^\top+\bm{Z}\bm{W}^\top-\bm{Z}'(\bm{W}')^\top. 
  \end{aligned}
\end{equation}
The diagonal of $\bm{D}_A$ is used for algebraic completion. $\mathcal{L}^{(t)}$ includes only the entries of $\bm{A}$ with $i<\ell$ and all entries of $\bm{B}$, which we call observed entries. For any $\bm{U},\bm{X}\in\RR^{n\times d}$ with $\bm{X}$ of full column rank and smallest singular value $\sigma_d(\bm{X})$, \citet[Lemma 5.4]{tu2016} give:
\begin{equation}
  \label{eq:proc} \min_{\bm{Q}^\top\bm{Q}=\bm{I}_d}\norm{\bm{U}-\bm{X}\bm{Q}}_{\fro}^2 \le\frac{\norm{\bm{U}\bm{U}^\top-\bm{X}\bm{X}^\top}_{\fro}^2}{2(\sqrt2-1)\sigma_d^2(\bm{X})}.
\end{equation}

\begin{lemma}\label{lem:recovery}
For $\bm{D}_A$ in \eqref{eq:D}, we have:
\begin{equation}
  \label{eq:diagonal} \sum_i(\bm{D}_A)_{ii}^2 \le\left\{\frac2n+\frac{8R_z^2}{(\sqrt2-1)n\kappa}\right\}\norm{\bm{D}_A}_{\fro}^2.
\end{equation}
Consequently, for all $n\ge n_0$ with a sufficiently large fixed $n_0$:
\begin{equation}
  \label{eq:offdiag} \norm{\bm{D}_A}_{\fro}^2\le4\sum_{i<\ell}(\bm{D}_A)_{i\ell}^2.
\end{equation}
$\Theta$ and $\Theta'$ also satisfy $\dist^2(\Theta,\Theta')\le C\{\norm{\bm{D}_A}_{\fro}^2/n^2+\norm{\bm{D}_B}_{\fro}^2/(nm)\}$, where $C$ depends only on $R_w,\kappa$. Together with \eqref{eq:offdiag}, this implies that, for $n\ge n_0$, equal means at observed entries identify the parameters up to a common orthogonal transformation.
\end{lemma}
\begin{proof}
Set $\bm{u}=\bm{\alpha}-\bm{\alpha}'$ and $\bm{G}=\bm{Z}\bm{Z}^\top-\bm{Z}'(\bm{Z}')^\top$. Centering gives $\bm{G}\one_n=\bm{0}$, so:
\begin{equation}
  \label{eq:orthA} \norm{\bm{D}_A}_{\fro}^2=2n\norm{\bm{u}}_2^2+2(\one_n^\top \bm{u})^2+\norm{\bm{G}}_{\fro}^2.
\end{equation}
By \eqref{eq:proc}, choose $\bm{Q}$ such that $\norm{\bm{Z}-\bm{Z}'\bm{Q}}_{\fro}\le\norm{\bm{G}}_{\fro}/\sqrt{2(\sqrt2-1)n\kappa}$. Applying $|\norm{\bm{z}_i}_2^2-\norm{\bm{z}_i'}_2^2|\le4R_z\norm{\bm{z}_i-\bm{Q}^\top \bm{z}_i'}_2$ rowwise gives:
\begin{equation*}
  \Big\{\sum_i(\bm{D}_A)_{ii}^2\Big\}^{1/2} \le2\norm{\bm{u}}_2+\frac{4R_z\norm{\bm{G}}_{\fro}}{\sqrt{2(\sqrt2-1)n\kappa}} \le\left\{\frac2n+\frac{8R_z^2}{(\sqrt2-1)n\kappa}\right\}^{1/2} (2n\norm{\bm{u}}_2^2+\norm{\bm{G}}_{\fro}^2)^{1/2}.
\end{equation*}
Equation~\eqref{eq:orthA} proves \eqref{eq:diagonal}. Choose a fixed $n_0$ such that the braced coefficient is at most $1/2$. Subtracting the diagonal contribution from the complete norm proves \eqref{eq:offdiag}. Equations~\eqref{eq:orthA} and~\eqref{eq:proc} allow the same $\bm{Q}$ to satisfy:
\begin{equation*}
  \frac{\norm{\bm{\alpha}-\bm{\alpha}'}_2^2}{n}\le\frac{\norm{\bm{D}_A}_{\fro}^2}{2n^2},\qquad \frac{\norm{\bm{Z}-\bm{Z}'\bm{Q}}_{\fro}^2}{n} \le\frac{\norm{\bm{D}_A}_{\fro}^2}{2(\sqrt2-1)n^2\kappa}.
\end{equation*}
The three $\bm{B}$-block components are mutually orthogonal, yielding:
\begin{equation*}
  \norm{\bm{D}_B}_{\fro}^2 =m\norm{\bm{\beta}-\bm{\beta}'}_2^2+n\norm{\bm{\gamma}-\bm{\gamma}'}_2^2 +\norm{\bm{Z}\bm{W}^\top-\bm{Z}'(\bm{W}')^\top}_{\fro}^2.
\end{equation*}
The identity $\bm{Z}'\bm{Q}(\bm{W}-\bm{W}'\bm{Q})^\top=\bm{Z}\bm{W}^\top-\bm{Z}'(\bm{W}')^\top-(\bm{Z}-\bm{Z}'\bm{Q})\bm{W}^\top$, together with $\sigma_d^2(\bm{Z}'\bm{Q})\ge n\kappa$ and $\norm{\bm{W}}_{\op}^2\le4mR_w^2$, gives:
\begin{equation*}
  \frac{\norm{\bm{W}-\bm{W}'\bm{Q}}_{\fro}^2}{m} \le\frac{2\norm{\bm{Z}\bm{W}^\top-\bm{Z}'(\bm{W}')^\top}_{\fro}^2}{nm\kappa} +\frac{8R_w^2}{n\kappa}\norm{\bm{Z}-\bm{Z}'\bm{Q}}_{\fro}^2.
\end{equation*}
Summing gives dimension-independent coefficients $1/2+1/\{2(\sqrt2-1)\kappa\}+4R_w^2/\{(\sqrt2-1)\kappa^2\}$ for the normalized $\bm{D}_A$ term and $1+2/\kappa$ for the normalized $\bm{D}_B$ term.
\end{proof}

For small $n$, identifiability can fail. For example, take $n=3,d=1,\bm{W}=\bm{0}$, $\bm{Z}=(-1,0,1)^\top$, $\bm{\alpha}=\bm{0}$, and $\bm{Z}'=\bm{Z}/2$, $\bm{\alpha}'=(-3/8,3/8,-3/8)^\top$. The two representations give identical off-diagonal $\bm{A}$ means. With the same $\bm{\beta},\bm{\gamma}$, they also give identical $\bm{B}$ means, but $\bm{Z}$ and $\bm{Z}'$ are not related by an orthogonal transformation.

The next deterministic bound serves estimation and comparison of successive objectives.
\begin{lemma}\label{lem:deterministic}
Fix nonnegative observations $\bm{A},\bm{B}$, with $\bm{A}$ symmetric and zero-diagonal. For $u>0$, suppose their row and column sums satisfy \eqref{eq:obsmargins} with $t$ replaced by $u$, and $\min\{u,nu/m\}$ exceeds a fixed constant. Let the feasible reference $\Theta'$ be centered, satisfy the radius constraints and $\lambda_{\min}\{(\bm{Z}')^\top \bm{Z}'/n\}\ge\kappa$, and have $\theta_{i\ell}^{A}(\Theta')$ between $cu/n$ and $Cu/n$ and $\theta_{ik}^{B}(\Theta')$ between $cu/m$ and $Cu/m$. Let $\widehat{\Theta}$ be feasible with global optimization error at most $\varepsilon\le\varepsilon_0$, where $\varepsilon_0$ is fixed. Define the mean matrices by $(\bm{\Theta}_A(\Theta'))_{i\ell}=\theta_{i\ell}^{A}(\Theta')$ for $i\ne\ell$, $(\bm{\Theta}_A(\Theta'))_{ii}=0$, and $(\bm{\Theta}_B(\Theta'))_{ik}=\theta_{ik}^{B}(\Theta')$, and set:
\begin{equation}
  \label{eq:arbitrary-noise} \bm{M}=\left[\frac{\sqrt n}{2}\{\bm{A}-\bm{\Theta}_A(\Theta')\},\ \sqrt m\{\bm{B}-\bm{\Theta}_B(\Theta')\}\right].
\end{equation}
For $n\ge n_0$:
\begin{equation}
  \label{eq:deterministic-bound} \dist^2(\widehat{\Theta},\Theta') \le C\frac{\norm{\bm{M}}_{\op}^2}{nu^2}+C\frac{\varepsilon}{nu}.
\end{equation}
The normalized squared log-mean error obeys the same bound.
\end{lemma}
\begin{proof}
The feasible marginal perturbations in Lemma~\ref{lem:fitmeans} apply to these observations. Together with Lemma~\ref{lem:meansfrommargins}, they give $\widehat{\theta}_{i\ell}^{A}\asymp u/n$ and $\widehat{\theta}_{ik}^{B}\asymp u/m$.

\textit{Objective difference.} In \eqref{eq:D}, take $\Theta=\widehat{\Theta}$ and use the current reference $\Theta'$. At observed entries, $(\bm{D}_A)_{i\ell}=\log\{\widehat{\theta}_{i\ell}^{A}/\theta_{i\ell}^{A}(\Theta')\}$, with the analogous definition for $\bm{D}_B$. At a given entry, let $Y$ be the observed count, $\widehat{\theta},\theta'$ the fitted and reference means, and $\zeta=\log(\widehat{\theta}/\theta')$. Then:
\begin{equation*}
  \widehat{\theta}-Y\log\widehat{\theta}-(\theta'-Y\log\theta') =\theta'(e^\zeta-1-\zeta)-(Y-\theta')\zeta.
\end{equation*}
Summing, using feasibility of $\Theta'$ and the optimization error bound $\varepsilon$, yields:
\begin{equation}
  \label{eq:basic} \begin{aligned}
     &\sum_{i<\ell}\theta_{i\ell}^{A}(\Theta') \{e^{(\bm{D}_A)_{i\ell}}-1-(\bm{D}_A)_{i\ell}\}+\sum_{i,k}\theta_{ik}^{B}(\Theta') \{e^{(\bm{D}_B)_{ik}}-1-(\bm{D}_B)_{ik}\}\\
     &\le\sum_{i<\ell}\{A_{i\ell}-\theta_{i\ell}^{A}(\Theta')\}(\bm{D}_A)_{i\ell} +\sum_{i,k}\{B_{ik}-\theta_{ik}^{B}(\Theta')\}(\bm{D}_B)_{ik} +\varepsilon. 
  \end{aligned}
\end{equation}

Set $\bm{F}=[\bm{D}_A/\sqrt n,\ \bm{D}_B/\sqrt m]$. The noise sums in \eqref{eq:basic} equal $\langle \bm{M},\bm{F}\rangle=\operatorname{tr}(\bm{M}^\top\bm{F})$: the zero-diagonal $\bm{A}$ block of $\bm{M}$ has factor $1/2$ against double counting. Taylor's formula gives:
\begin{equation*}
  \theta'(e^\zeta-1-\zeta)=\zeta^2\int_0^1(1-v)(\theta')^{1-v}\widehat{\theta}^{\,v}\,dv\ge\tfrac12\min\{\theta',\widehat{\theta}\}\zeta^2.
\end{equation*}
The reference and fitted mean bounds, followed by \eqref{eq:offdiag}, therefore yield:
\begin{equation*}
  cu\norm{\bm{F}}_{\fro}^2\le cu\left\{\frac1n\sum_{i<\ell}(\bm{D}_A)_{i\ell}^2+\frac1m\norm{\bm{D}_B}_{\fro}^2\right\}\le\langle \bm{M},\bm{F}\rangle+\varepsilon,
\end{equation*}
where the first constant is decreased if needed.

\textit{Low-rank noise bound.} Equation~\eqref{eq:D} shows that the column space of $\bm{F}$ is contained in the span of $\widehat{\bm{Z}},\bm{Z}',\widehat{\bm{\alpha}}-\bm{\alpha}',\widehat{\bm{\beta}}-\bm{\beta}',\one_n$. Thus $\rank(\bm{F})\le2d+3$. For a singular value decomposition $\bm{F}=\sum_j\sigma_j \bm{a}_j\bm{b}_j^\top$:
\begin{equation*}
  |\langle \bm{M},\bm{F}\rangle| \le\sum_j\sigma_j|\bm{a}_j^\top \bm{M}\bm{b}_j| \le\norm{\bm{M}}_{\op}\sum_j\sigma_j \le\sqrt{2d+3}\norm{\bm{M}}_{\op}\norm{\bm{F}}_{\fro}.
\end{equation*}
The last step uses Cauchy--Schwarz and fixed $d$.

Young's inequality now gives:
\begin{equation*}
  cu\norm{\bm{F}}_{\fro}^2\le C\norm{\bm{M}}_{\op}\norm{\bm{F}}_{\fro}+\varepsilon\le\frac{cu}{2}\norm{\bm{F}}_{\fro}^2+\frac C{u}\norm{\bm{M}}_{\op}^2+\varepsilon.
\end{equation*}
Thus $\norm{\bm{F}}_{\fro}^2\le C\norm{\bm{M}}_{\op}^2/u^2+C\varepsilon/u$. Lemma~\ref{lem:recovery} gives \eqref{eq:deterministic-bound}. The normalized log-mean error is at most $\norm{\bm{F}}_{\fro}^2/n$.
\end{proof}

For local optimization, write a centered perturbation as $\bm{h}=(\bm{a},\bm{b},\bm{c},\bm{U},\bm{V})$, with rows $\bm{u}_i^\top,\bm{v}_k^\top$ in $\bm{U},\bm{V}$, and $\one_n^\top \bm{b}=0$, $\one_n^\top \bm{U}=\bm{0}^\top$, $\one_m^\top \bm{V}=\bm{0}^\top$. Use the norm:
\begin{equation}
  \label{eq:tangentnorm} \norm{\bm{h}}_{n,m}^2 =\frac{\norm{\bm{a}}_2^2+\norm{\bm{b}}_2^2+\norm{\bm{U}}_{\fro}^2}{n} +\frac{\norm{\bm{c}}_2^2+\norm{\bm{V}}_{\fro}^2}{m}.
\end{equation}
The inner product $\langle\cdot,\cdot\rangle_{n,m}$ is the polarization of the norm $\norm{\cdot}_{n,m}$ on perturbations, and Equation~\eqref{eq:distance} minimizes the discrepancy over common orthogonal transformations.

\begin{lemma}\label{lem:alignment}
The minimum in \eqref{eq:distance} is attained, and $\dist$ satisfies the triangle inequality. For two tuples $\Theta,\Theta'$, let $\bm{Q}$ optimally align $\Theta'$ to $\Theta$ and set $\Theta_0=\Theta'\bm{Q}$. For $\bm{h}=\Theta-\Theta_0$, we have $\norm{\bm{h}}_{n,m}=\dist(\Theta,\Theta_0)$ and:
\begin{equation}
  \label{eq:orthogonality} (\bm{Z}_0)^\top \bm{U}/n+(\bm{W}_0)^\top \bm{V}/m\quad\text{is symmetric}.
\end{equation}
Equivalently, $\bm{h}$ is orthogonal to every tangent direction $(\bm{0},\bm{0},\bm{0},\bm{Z}_0\bm{\Omega},\bm{W}_0\bm{\Omega})$ of common orthogonal transformations, where $\bm{\Omega}^\top=-\bm{\Omega}$.
\end{lemma}
\begin{proof}
Compactness of the orthogonal group and continuity give an optimal orthogonal matrix. For three tuples $\Theta_1,\Theta_2,\Theta_3$, let $\bm{Q}_{12}$ align $\Theta_2$ to $\Theta_1$ and $\bm{Q}_{23}$ align $\Theta_3$ to $\Theta_2$. Invariance of \eqref{eq:tangentnorm} gives:
\begin{equation*}
  \begin{aligned}
    \dist(\Theta_1,\Theta_3)&\le\norm{\Theta_1-\Theta_3\bm{Q}_{23}\bm{Q}_{12}}_{n,m}\\
    &\le\norm{\Theta_1-\Theta_2\bm{Q}_{12}}_{n,m}+\norm{(\Theta_2-\Theta_3\bm{Q}_{23})\bm{Q}_{12}}_{n,m}\\
    &=\dist(\Theta_1,\Theta_2)+\dist(\Theta_2,\Theta_3).
  \end{aligned}
\end{equation*}

For $\bm{\Omega}^\top=-\bm{\Omega}$, optimal alignment implies:
\begin{equation*}
  0=\left.\frac{d}{ds}\norm{\Theta-\Theta_0\exp(s\bm{\Omega})}_{n,m}^2\right|_{s=0}=-2\operatorname{tr}\!\left\{\left(\frac{\bm{Z}_0^\top \bm{U}}{n}+\frac{\bm{W}_0^\top \bm{V}}{m}\right)^\top\bm{\Omega}\right\}.
\end{equation*}
Orthogonality to every skew-symmetric matrix is equivalent to symmetry, proving \eqref{eq:orthogonality}.
\end{proof}

\begin{lemma}\label{lem:jacobian}
Fix a centered reference $\Theta_0$ with row norms bounded by $2R_z,2R_w$ and $\lambda_{\min}(\bm{Z}_0^\top \bm{Z}_0/n)\ge\kappa_1>0$, where $\kappa_1$ is fixed. For centered $\Theta$ and $\bm{h}=(\bm{a},\bm{b},\bm{c},\bm{U},\bm{V})$, define the first-order log-mean changes:
\begin{equation}
  \label{eq:jacobian} \begin{aligned}
    j_{A,i\ell}(\Theta;\bm{h})&=a_i+a_\ell+\bm{u}_i^\top \bm{z}_\ell+\bm{z}_i^\top \bm{u}_\ell,\\
    j_{B,ik}(\Theta;\bm{h})&=b_i+c_k+\bm{u}_i^\top \bm{w}_k+\bm{z}_i^\top \bm{v}_k.
  \end{aligned}
\end{equation}
At the reference $\Theta_0$, every such $\bm{h}$ satisfies:
\begin{equation}
  \label{eq:Jupper} \frac1{n^2}\sum_{i<\ell}j_{A,i\ell}(\Theta_0;\bm{h})^2 +\frac1{nm}\sum_{i,k}j_{B,ik}(\Theta_0;\bm{h})^2\le C\norm{\bm{h}}_{n,m}^2.
\end{equation}
If $\bm{h}$ also satisfies \eqref{eq:orthogonality}, then for $n\ge n_0$ with a sufficiently large fixed $n_0$:
\begin{equation}
  \label{eq:Jlower} \frac1{n^2}\sum_{i<\ell}j_{A,i\ell}(\Theta_0;\bm{h})^2 +\frac1{nm}\sum_{i,k}j_{B,ik}(\Theta_0;\bm{h})^2\ge c\norm{\bm{h}}_{n,m}^2.
\end{equation}
\end{lemma}
\begin{proof}
\textit{Upper bound.} Applying $(x_1+\cdots+x_4)^2\le4\sum_jx_j^2$ and counting each worker at most $n-1$ times gives:
\begin{equation*}
  \begin{aligned}
    \sum_{i<\ell}j_{A,i\ell}(\Theta_0;\bm{h})^2&\le4(n-1)\{\norm{\bm{a}}_2^2+4R_z^2\norm{\bm{U}}_{\fro}^2\},\\
    \sum_{i,k}j_{B,ik}(\Theta_0;\bm{h})^2&\le4\{m\norm{\bm{b}}_2^2+n\norm{\bm{c}}_2^2+4mR_w^2\norm{\bm{U}}_{\fro}^2+4nR_z^2\norm{\bm{V}}_{\fro}^2\}.
  \end{aligned}
\end{equation*}
Dividing by $n^2,nm$ and adding proves \eqref{eq:Jupper}.

\textit{First-order limit of the distance.} Consider $\Theta_0+\varepsilon \bm{h}$ as $\varepsilon\to0$. We first establish:
\begin{equation}
  \label{eq:metricderivative} \lim_{\varepsilon\to0} \frac{\dist^2(\Theta_0+\varepsilon \bm{h},\Theta_0)}{\varepsilon^2} =\min_{\bm{\Omega}^\top=-\bm{\Omega}} \norm{\bm{h}-(\bm{0},\bm{0},\bm{0},\bm{Z}_0\bm{\Omega},\bm{W}_0\bm{\Omega})}_{n,m}^2.
\end{equation}
The choice $\bm{Q}=\bm{I}_d$ bounds the minimum distance by $|\varepsilon|\norm{\bm{h}}_{n,m}$. For optimal $\bm{Q}_\varepsilon$, the triangle inequality gives:
\begin{equation*}
  \frac{\norm{\bm{Z}_0(\bm{I}_d-\bm{Q}_\varepsilon)}_{\fro}}{\sqrt n} \le\frac{|\varepsilon|\norm{\bm{U}}_{\fro}}{\sqrt n} +\dist(\Theta_0+\varepsilon \bm{h},\Theta_0) \le2|\varepsilon|\norm{\bm{h}}_{n,m}.
\end{equation*}
The spectral bound gives $\norm{\bm{Q}_\varepsilon-\bm{I}_d}_{\fro}\le2|\varepsilon|\norm{\bm{h}}_{n,m}/\sqrt{\kappa_1}$. Hence $\bm{\Omega}_\varepsilon=(\bm{Q}_\varepsilon-\bm{I}_d)/\varepsilon$ is bounded, and $\bm{\Omega}_\varepsilon+\bm{\Omega}_\varepsilon^\top=-\varepsilon\bm{\Omega}_\varepsilon^\top\bm{\Omega}_\varepsilon$. Every subsequential limit is skew-symmetric, so substitution into the distance gives:
\begin{equation*}
  \liminf_{\varepsilon\to0}\frac{\dist^2(\Theta_0+\varepsilon \bm{h},\Theta_0)}{\varepsilon^2}\ge\min_{\bm{\Omega}^\top=-\bm{\Omega}}\norm{\bm{h}-(\bm{0},\bm{0},\bm{0},\bm{Z}_0\bm{\Omega},\bm{W}_0\bm{\Omega})}_{n,m}^2.
\end{equation*}
Conversely, each skew-symmetric $\bm{\Omega}$ gives the rotation $\exp(\varepsilon\bm{\Omega})=\bm{I}_d+\varepsilon\bm{\Omega}+O(\varepsilon^2)$ and the matching upper bound. Full column rank of $\bm{Z}_0$ makes the quadratic coercive in $\bm{\Omega}$, so the minimum exists. For $\bm{h}$ satisfying \eqref{eq:orthogonality}, $\bm{\Omega}=\bm{0}$ attains the minimum $\norm{\bm{h}}_{n,m}^2$.

\textit{Lower bound.} Apply Lemma~\ref{lem:recovery} and \eqref{eq:offdiag} to $\Theta_0+\varepsilon\bm{h}$ and $\Theta_0$. For any fixed $\bm{h}$ and sufficiently small $|\varepsilon|$, the row bounds can be enlarged to $3R_z,3R_w$, with constants independent of $\bm{h},n,m$. The observed log-mean differences are:
\begin{equation*}
  (\bm{D}_A)_{i\ell}=\varepsilon j_{A,i\ell}(\Theta_0;\bm{h})+\varepsilon^2\bm{u}_i^\top\bm{u}_\ell,\qquad (\bm{D}_B)_{ik}=\varepsilon j_{B,ik}(\Theta_0;\bm{h})+\varepsilon^2\bm{u}_i^\top\bm{v}_k.
\end{equation*}
For $\bm{h}$ satisfying \eqref{eq:orthogonality}, division by $\varepsilon^2$ and \eqref{eq:metricderivative} give:
\begin{equation*}
  \begin{aligned}
    \norm{\bm{h}}_{n,m}^2&=\lim_{\varepsilon\to0}\frac{\dist^2(\Theta_0+\varepsilon\bm{h},\Theta_0)}{\varepsilon^2}\\
    &\le C\lim_{\varepsilon\to0}\left\{\frac1{n^2}\sum_{i<\ell}\left(\frac{(\bm{D}_A)_{i\ell}}\varepsilon\right)^2+\frac1{nm}\sum_{i,k}\left(\frac{(\bm{D}_B)_{ik}}\varepsilon\right)^2\right\}\\
    &=C\left\{\frac1{n^2}\sum_{i<\ell}j_{A,i\ell}(\Theta_0;\bm{h})^2+\frac1{nm}\sum_{i,k}j_{B,ik}(\Theta_0;\bm{h})^2\right\}.
  \end{aligned}
\end{equation*}
This proves \eqref{eq:Jlower}.
\end{proof}

As in \eqref{eq:arbitrary-noise}, the diagonal of $\bm{\Theta}_A(\Theta)$ is set to zero. At checkpoint $t$, define residuals $\bm{E}_A^{(t)}(\Theta)=\bm{\Theta}_A(\Theta)-\bm{A}^{(t)}$ and $\bm{E}_B^{(t)}(\Theta)=\bm{\Theta}_B(\Theta)-\bm{B}^{(t)}$. Differentiating $\mathcal{L}^{(t)}$ block by block gives:
\begin{equation*}
  \begin{aligned}
     (\nabla_{\bm{\alpha}}\mathcal{L}^{(t)})_i&=\sum_{\ell\ne i}(\bm{E}_A^{(t)}(\Theta))_{i\ell},& \nabla_{\bm{\beta}}\mathcal{L}^{(t)}&=\bm{E}_B^{(t)}(\Theta)\one_m,& \nabla_{\bm{\gamma}}\mathcal{L}^{(t)}&=\{\bm{E}_B^{(t)}(\Theta)\}^\top\one_n,\\
     \nabla_{\bm{Z}}\mathcal{L}^{(t)}&=\bm{E}_A^{(t)}(\Theta)\bm{Z}+\bm{E}_B^{(t)}(\Theta)\bm{W},& \nabla_{\bm{W}}\mathcal{L}^{(t)}&=\{\bm{E}_B^{(t)}(\Theta)\}^\top \bm{Z}. 
  \end{aligned}
\end{equation*}
Each $\bm{A}$ edge appears once in $\mathcal{L}^{(t)}$, and both endpoint derivatives enter the corresponding rows. On the centered parameter space with \eqref{eq:tangentnorm}, $\mathcal{L}^{(t)}/(nt)$ has gradient:
\begin{equation}
  \label{eq:gradientmetric} \nabla_{n,m}\{\mathcal{L}^{(t)}/(nt)\} =\bm{P}\left(\frac{\nabla_{\bm{\alpha}}\mathcal{L}^{(t)}}t, \frac{\nabla_{\bm{\beta}}\mathcal{L}^{(t)}}t, \frac{m\nabla_{\bm{\gamma}}\mathcal{L}^{(t)}}{nt}, \frac{\nabla_{\bm{Z}}\mathcal{L}^{(t)}}t, \frac{m\nabla_{\bm{W}}\mathcal{L}^{(t)}}{nt}\right).
\end{equation}
The block weights make $\bm{P}$ orthogonal and give, for centered $\bm{h}$:
\begin{equation*}
  \left\langle\nabla_{n,m}\{\mathcal{L}^{(t)}/(nt)\},\bm{h}\right\rangle_{n,m}=D\mathcal{L}^{(t)}(\Theta)[\bm{h}]/(nt).
\end{equation*}
Thus \eqref{eq:algorithm} is the gradient update in this weighted norm.

Along $\Theta+s\bm{h}$, the first derivatives of $\log\theta_{i\ell}^{A}$ and $\log\theta_{ik}^{B}$ at $s=0$ are $j_{A,i\ell}(\Theta;\bm{h})$ and $j_{B,ik}(\Theta;\bm{h})$ from \eqref{eq:jacobian}. The corresponding second derivatives are $2\bm{u}_i^\top \bm{u}_\ell$ and $2\bm{u}_i^\top \bm{v}_k$. Differentiating once more yields:
\begin{equation}
  \label{eq:hessianexact} \begin{aligned}
     D^2\{\mathcal{L}^{(t)}/(nt)\}(\Theta)[\bm{h},\bm{h}] ={}&\frac1{nt}\sum_{i<\ell}\theta_{i\ell}^{A}(\Theta)j_{A,i\ell}(\Theta;\bm{h})^2\\
     &+\frac1{nt}\sum_{i,k}\theta_{ik}^{B}(\Theta)j_{B,ik}(\Theta;\bm{h})^2\\
     &+\frac1{nt}\operatorname{tr}(\bm{U}^\top \bm{E}_A^{(t)}(\Theta)\bm{U}) +\frac2{nt}\operatorname{tr}(\bm{U}^\top \bm{E}_B^{(t)}(\Theta)\bm{V}). 
  \end{aligned}
\end{equation}
The mean-weighted terms and residual terms jointly determine local curvature.

\begin{lemma}\label{lem:localgeometry}
Suppose $n\ge n_0$. Let $\Theta_0$ be centered with fixed latent row bounds and $\lambda_{\min}(\bm{Z}_0^\top \bm{Z}_0/n)\ge\kappa_1>0$. Suppose $ct/n\le\theta_{i\ell}^{A}(\Theta_0)\le Ct/n$ and $ct/m\le\theta_{ik}^{B}(\Theta_0)\le Ct/m$. There are fixed $b_0,r_0,\mu,L>0$ such that the following residual condition suffices for the local bounds below:
\begin{equation}
  \label{eq:residualcertificate} \frac{\norm{\bm{\Theta}_A(\Theta_0)-\bm{A}^{(t)}}_{\op}}{t} +\sqrt{\frac mn}\frac{\norm{\bm{\Theta}_B(\Theta_0)-\bm{B}^{(t)}}_{\op}}{t} \le b_0.
\end{equation}
For every centered $\Theta$ with $\norm{\Theta-\Theta_0}_{n,m}\le r_0/\sqrt{n+m}$ and every centered $\bm{h}$:
\begin{equation}
  \label{eq:hessupper} \left|D^2\{\mathcal{L}^{(t)}/(nt)\}(\Theta)[\bm{h},\bm{h}]\right| \le L\norm{\bm{h}}_{n,m}^2.
\end{equation}
If $\bm{h}$ also satisfies \eqref{eq:orthogonality} at $\Theta_0$, then:
\begin{equation}
  \label{eq:hesslower} D^2\{\mathcal{L}^{(t)}/(nt)\}(\Theta)[\bm{h},\bm{h}] \ge\mu\norm{\bm{h}}_{n,m}^2.
\end{equation}
All constants are independent of $n,m,t,T$.
\end{lemma}
\begin{proof}
Write $\bm{h}_0=\Theta-\Theta_0$ within this proof. The norm definition bounds each worker scalar or latent row change by $\sqrt n\norm{\bm{h}_0}_{n,m}$, and each task-type scalar or row change by $\sqrt m\norm{\bm{h}_0}_{n,m}$. All these changes are therefore at most $r_0$. Expanding the inner products around the bounded reference positions gives:
\begin{equation}
  \label{eq:loguniformlocal} \max_{i<\ell}|\log\{\theta_{i\ell}^{A}(\Theta)/\theta_{i\ell}^{A}(\Theta_0)\}| \vee\max_{i,k}|\log\{\theta_{ik}^{B}(\Theta)/\theta_{ik}^{B}(\Theta_0)\}| \le C\sqrt{n+m}\norm{\bm{h}_0}_{n,m}\le Cr_0.
\end{equation}
For example, the $\bm{A}$ inner-product change is $(\bm{z}_i-\bm{z}_{0,i})^\top \bm{z}_\ell+\bm{z}_{0,i}^\top(\bm{z}_\ell-\bm{z}_{0,\ell})$, where the current row $\bm{z}_\ell$ is bounded by the reference radius plus $r_0$. Thus $\theta_{i\ell}^{A}(\Theta)$ and $\theta_{ik}^{B}(\Theta)$ remain within fixed multiples of $t/n$ and $t/m$.

\textit{Mean-weighted terms.} The two mean-weighted sums in \eqref{eq:hessianexact} are comparable to $n^{-2}\sum_{i<\ell}j_{A,i\ell}(\Theta;\bm{h})^2$ and $(nm)^{-1}\sum_{i,k}j_{B,ik}(\Theta;\bm{h})^2$, respectively. The derivative changes from the reference obey:
\begin{equation*}
  \begin{aligned}
    \frac{\norm{\bm{U}(\bm{Z}-\bm{Z}_0)^\top+(\bm{Z}-\bm{Z}_0)\bm{U}^\top}_{\fro}}{n}&\le2\frac{\norm{\bm{U}}_{\fro}}{\sqrt n}\frac{\norm{\bm{Z}-\bm{Z}_0}_{\fro}}{\sqrt n},\\
    \frac{\norm{\bm{U}(\bm{W}-\bm{W}_0)^\top+(\bm{Z}-\bm{Z}_0)\bm{V}^\top}_{\fro}}{\sqrt{nm}}&\le\frac{\norm{\bm{U}}_{\fro}}{\sqrt n}\frac{\norm{\bm{W}-\bm{W}_0}_{\fro}}{\sqrt m}+\frac{\norm{\bm{Z}-\bm{Z}_0}_{\fro}}{\sqrt n}\frac{\norm{\bm{V}}_{\fro}}{\sqrt m}.
  \end{aligned}
\end{equation*}
Their sum is at most $C\norm{\bm{h}_0}_{n,m}\norm{\bm{h}}_{n,m}$. Lemma~\ref{lem:jacobian} and the reverse triangle inequality therefore bound the mean-weighted terms below by $c_F\norm{\bm{h}}_{n,m}^2$ for $\bm{h}$ satisfying \eqref{eq:orthogonality} when $r_0$ is small. The upper bound \eqref{eq:Jupper}, with enlarged fixed radii, gives $C_F\norm{\bm{h}}_{n,m}^2$ for every $\bm{h}$.

\textit{Residual terms.} By $|e^x-1|\le e^{|x|}|x|$ and \eqref{eq:loguniformlocal}, entrywise mean changes are at most $C(t/n)\sqrt{n+m}\norm{\bm{h}_0}_{n,m}$ in $\bm{A}$ and $C(t/m)\sqrt{n+m}\norm{\bm{h}_0}_{n,m}$ in $\bm{B}$. The operator norm of a symmetric matrix is at most its maximum absolute row sum, and the operator norm is always at most the Frobenius norm. Consequently:
\begin{equation}
  \label{eq:residualchange} \frac{\norm{\bm{\Theta}_A(\Theta)-\bm{\Theta}_A(\Theta_0)}_{\op}}t +\sqrt{\frac mn}\frac{\norm{\bm{\Theta}_B(\Theta)-\bm{\Theta}_B(\Theta_0)}_{\op}}t \le C\sqrt{n+m}\norm{\bm{h}_0}_{n,m}\le Cr_0.
\end{equation}
Adding \eqref{eq:residualcertificate} bounds the same standardized residual sum at the current parameter by $b_0+Cr_0$. The operator-norm definition gives:
\begin{equation*}
  \begin{aligned}
    \frac{|\operatorname{tr}(\bm{U}^\top \bm{E}_A^{(t)}(\Theta)\bm{U})|}{nt} &\le\frac{\norm{\bm{E}_A^{(t)}(\Theta)}_{\op}}t\frac{\norm{\bm{U}}_{\fro}^2}{n},\\
    \frac{2|\operatorname{tr}(\bm{U}^\top \bm{E}_B^{(t)}(\Theta)\bm{V})|}{nt} &\le2\sqrt{\frac mn}\frac{\norm{\bm{E}_B^{(t)}(\Theta)}_{\op}}t \frac{\norm{\bm{U}}_{\fro}}{\sqrt n}\frac{\norm{\bm{V}}_{\fro}}{\sqrt m}.
  \end{aligned}
\end{equation*}
Using $2ab\le a^2+b^2$, the absolute residual contribution is at most $C(b_0+Cr_0)\norm{\bm{h}}_{n,m}^2$. Choose $b_0$, then $r_0$, sufficiently small that the residual contribution is at most $c_F\norm{\bm{h}}_{n,m}^2/2$. The lower bound under \eqref{eq:orthogonality} holds with $\mu=c_F/2$, and the full-space absolute upper bound holds with a fixed $L\ge C_F+c_F/2$. The constants depend only on the fixed mean bounds, radii, and $\kappa_1$.
\end{proof}

Self-adjointness and \eqref{eq:hessupper} bound the Hessian operator norm by $L$. Integration along segments gives an $L$-Lipschitz gradient on the local ball. Equation~\eqref{eq:hesslower} applies to directions satisfying \eqref{eq:orthogonality} at the reference.

\subsection{Proof of Theorem~\ref{thm:stat}}\label{sec:statproof}
For the rest of this section, write $x_{t,\delta}=\log\{(2n+m)/\delta\}+\log\log(e^e+t)$ for $t\ge1$ and $\delta\in(0,1)$.
\begin{proof}
We first derive the order of $S_t$. By Assumption~\ref{as:model}, each work record has total contribution between one and $K$. Taking conditional expectations given $\FF_{r-1}$ and summing along each path through $r\le S_t$ gives:
\begin{equation*}
  S_t\le\sum_{i<\ell}\theta_{\star,i\ell}^{A,(t)} +\sum_{i,k}\theta_{\star,ik}^{B,(t)}\le KS_t.
\end{equation*}
$\sum_{i<\ell}\theta_{\star,i\ell}^{A,(t)}$ is half the sum of the row sums of $(\theta_{\star,i\ell}^{A,(t)})$, while $\sum_{i,k}\theta_{\star,ik}^{B,(t)}$ is the sum of the row sums of $(\theta_{\star,ik}^{B,(t)})$. By \eqref{eq:expectedmargins}, the middle quantity lies between $(3c_0/2)nt$ and $(3C_0/2)nt$. Hence $S_t\asymp nt$, pathwise, including random $S_t$.

Work on the time-uniform event of Lemma~\ref{lem:concentration}. At any checkpoint satisfying the information threshold in Theorem~\ref{thm:stat}, \eqref{eq:obsmargins} holds and the smallest marginal sum is sufficiently large. Lemma~\ref{lem:fitmeans} gives a finite $\widehat{\Theta}_\star^{(t)}$. Apply Lemma~\ref{lem:deterministic} with $u=t$ and $\Theta'=\Theta_{\star}^{(t)}$. The matrix in \eqref{eq:arbitrary-noise} is then $\bm{M}_{S_t}$ from \eqref{eq:H}. Since $t\ge C(1+m/n)x_{t,\delta}$, $x_{t,\delta}\ge1$, and $\varepsilon_0$ is fixed, \eqref{eq:matrixnoise} and \eqref{eq:deterministic-bound} give:
\begin{equation*}
  \begin{aligned}
    \dist^2(\widehat{\Theta}^{(t)},\Theta_{\star}^{(t)})&\le C\frac{\norm{\bm{M}_{S_t}}_{\op}^2}{nt^2}+C\frac{\varepsilon_0}{nt}\\
    &\le C\frac{(n+m)x_{t,\delta}}{nt}\left\{1+\frac{x_{t,\delta}}{t}+\frac{\varepsilon_0}{(n+m)x_{t,\delta}}\right\}\le C\frac{(n+m)x_{t,\delta}}{nt}.
  \end{aligned}
\end{equation*}
The same deterministic bound controls \eqref{eq:lograte}. This proves \eqref{eq:statrate} on the common concentration event, without further failure probability. The matrix bound is uniform over low-rank directions and therefore allows history-dependent $\Theta_\star^{(t)}$.

For the result over $\lceil aT\rceil\le t\le T$, use \eqref{eq:finiteM} and the corresponding marginal event. Write $x=\log\{(2n+m)/\delta\}$. For $t\ge\lceil aT\rceil$, Lemma~\ref{lem:deterministic} gives:
\begin{equation*}
  \dist^2(\widehat{\Theta}^{(t)},\Theta_{\star}^{(t)}) \le C_a\frac{n+m}{nT}x +C_a\frac{n+m}{nT^2}x^2 +C_a\frac{\varepsilon_0}{nT}.
\end{equation*}
The condition $T\ge C_a(1+m/n)x$ ensures the marginal bounds and lets the first term absorb the other two. Taking maxima on the same event proves the claimed bounds over these checkpoints. The log-mean bound follows from the same deterministic inequality.
\end{proof}

\emph{Larger cumulative activity.} Replace Assumption~\ref{as:scale} by worker row sums of order $(n+m)t$ and task-type column sums of order $n(n+m)t/m$:
\begin{equation}
  \label{eq:densemargins} \begin{aligned}
     & c_0(n+m)t\le\sum_{\ell\ne i}\theta_{\star,i\ell}^{A,(t)}\le C_0(n+m)t,\quad c_0(n+m)t\le\sum_k\theta_{\star,ik}^{B,(t)}\le C_0(n+m)t,\\
     &c_0n(n+m)t/m\le\sum_i\theta_{\star,ik}^{B,(t)}\le C_0n(n+m)t/m. 
  \end{aligned}
\end{equation}
For $n\ge n_0$, Lemma~\ref{lem:meansfrommargins} gives:
\begin{equation*}
  c\frac{(n+m)t}{n}\le\theta_{\star,i\ell}^{A,(t)}\le C\frac{(n+m)t}{n},\quad c\frac{(n+m)t}{m}\le\theta_{\star,ik}^{B,(t)}\le C\frac{(n+m)t}{m}.
\end{equation*}
Keep the feasible parameter set, model \eqref{eq:joint}, Assumption~\ref{as:model}, and $\varepsilon_0$ unchanged. Under Assumptions~\ref{as:model}, \ref{as:geom}, and these larger activity scales, $S_t\asymp n(n+m)t$. With probability at least $1-\delta$, every estimate satisfying \eqref{eq:opt}, at every checkpoint with $nt\ge C\{\log((2n+m)/\delta)+\log\log(e^e+t)\}$, obeys:
\begin{equation}
  \label{eq:denserate} \dist^2(\widehat{\Theta}^{(t)},\Theta_{\star}^{(t)}) \le\frac C{nt}\left\{\log\frac{2n+m}{\delta}+\log\log(e^e+t)\right\}.
\end{equation}
For fixed $a\in(0,1)$ and $nT\ge C_a\log\{(2n+m)/\delta\}$, the maximum normalized squared parameter error over $\lceil aT\rceil\le t\le T$ is at most $C_a\log\{(2n+m)/\delta\}/(nT)$ on one event of probability at least $1-\delta$.

To verify these claims, summing \eqref{eq:densemargins} gives $S_t\asymp n(n+m)t$. Write $\nu=n+m$ within this argument. The single-increment norm in \eqref{eq:H} remains at most $C\sqrt \nu$, while the variance calculation in Lemma~\ref{lem:concentration} replaces $t$ by $\nu t$. The left and right variance bounds become $C\nu^2t$ and $Cn\nu t$, hence the total bound is $C\nu^2t$. Partitioning checkpoints into $2^j\le t<2^{j+1}$ and applying Freedman gives $\norm{\bm{M}_{S_t}}_{\op}^2\le C\{\nu^2tx_{t,\delta}+\nu x_{t,\delta}^2\}$, with failure probabilities allocated by checkpoint index.

Worker margins now have scale $\nu t$, and task-type margins have scale $n\nu t/m$. Their minimum scale $\nu t\min\{1,n/m\}$ lies between $nt$ and $2nt$: use $n\le \nu\le2n$ when $m\le n$, and $m\le \nu\le2m$ when $m\ge n$. Thus $nt\ge Cx_{t,\delta}$ ensures marginal concentration and finite minimizers. Applying Lemma~\ref{lem:deterministic} with $u=\nu t$ and $\Theta'=\Theta_{\star}^{(t)}$ yields:
\begin{equation*}
  \dist^2(\widehat{\Theta}^{(t)},\Theta_{\star}^{(t)}) \le C\frac{\norm{\bm{M}_{S_t}}_{\op}^2}{n\nu^2t^2} +C\frac{\varepsilon_0}{n\nu t} \le C\frac{x_{t,\delta}}{nt}+C\frac{x_{t,\delta}^2}{n\nu t^2} +C\frac{\varepsilon_0}{n\nu t}.
\end{equation*}
The ratio of the second term in the last bound to the first is $x_{t,\delta}/(\nu t)\le n/(C\nu)$. The fixed optimization error also lets the first term absorb the third. This proves \eqref{eq:denserate}. For fixed $T$, variance budget $C\nu^2T$ gives the matrix bound $C\{\nu^2T\log((2n+m)/\delta)+\nu\log^2((2n+m)/\delta)\}$. The marginal variance budgets are $C\nu T,Cn\nu T/m$. Apply the same deterministic calculation over $t\ge aT$. Finally, with $x=\log\{(2n+m)/\delta\}\ge1$, the function $\log\log(e^e+t)/t$ is bounded for $t\ge1$, so $\sup_{t\ge1}x_{t,\delta}/t\le Cx$. Thus $n\ge Cx$ ensures the corresponding information threshold for every $t\ge1$, and the supremum of the bound in \eqref{eq:denserate} is at most $Cx/n$. The maximum error over $\lceil aT\rceil\le t\le T$ has scale $x/(nT)$, corresponding to cumulative activity of order $n(n+m)T$. An optimization analysis under these larger activity scales requires an analogue of Assumption~\ref{as:opt} and parameter-block step sizes.

\subsection{Proof of Theorem~\ref{thm:optimization}}\label{sec:optproof}
\begin{proof}
Work on the time-uniform event of Lemma~\ref{lem:concentration}. We establish local geometry at $\widehat{\Theta}_\star^{(t)}$, bound $\dist(\widehat{\Theta}_\star^{(t)},\widehat{\Theta}_\star^{(t-1)})$, and then prove finite-step tracking. Write $x=\log\{(2n+m)/\delta\}$, so that $x_{t,\delta}=x+\log\log(e^e+t)$.

\textit{Step 1: Initial information and empirical local geometry.} For any sufficiently small fixed $\rho>0$, increasing the dimension-independent constant in the initial threshold, that is, the lower bound on $t_0$ in Theorem~\ref{thm:optimization}, ensures that, for every $t_0\le t\le T$, all feasible estimates with global optimization error at most $\varepsilon_0$ satisfy:
\begin{equation}
  \label{eq:smallaverage} \dist(\widehat{\Theta}^{(t)},\Theta_{\star}^{(t)}) \le\frac{\rho}{\sqrt{n+m}}.
\end{equation}
If $\rho\le\min\{R_z,R_w\}/4$, Assumption~\ref{as:geom} then bounds their latent row norms by $5R_z/4,5R_w/4$, leaving margins of at least $3R_z/4,3R_w/4$ inside the radius constraints. Every such $\widehat{\Theta}^{(t)}$ also satisfies $\lambda_{\min}\{(\widehat{\bm{Z}}^{(t)})^\top\widehat{\bm{Z}}^{(t)}/n\}\ge\kappa/4$. For every $\widehat{\Theta}_\star^{(t)}$, the left side of \eqref{eq:residualcertificate} can also be made at most $b_0/4$.

The function $x_{t,\delta}/t$ decreases for $t\ge1$, since $tx_{t,\delta}'=t/\{(e^e+t)\log(e^e+t)\}<1/e<x_{t,\delta}$. For $b_*>0$, set $\nu=n+m$, $\tilde\nu=\nu^2/n$, and $t_*=b_*\tilde\nu x$. Here $\nu\le\tilde\nu\le \nu^2$ and $x\ge\log \nu\ge1$, so $e^e+t_*\le(e^e+b_*)e^{3x}$. Consequently:
\begin{equation*}
  \sup_{t\ge t_*}\frac{\tilde\nu x_{t,\delta}}{t}=\frac{x_{t_*,\delta}}{b_*x}\le\frac{2+\log\{\log(e^e+b_*)+3\}}{b_*}\longrightarrow0\qquad(b_*\to\infty).
\end{equation*}
Choosing a fixed large $b_*$ controls $\tilde\nu x_{t,\delta}/t$ uniformly over $t\ge t_0\ge t_*$.

This choice also ensures the information threshold of Theorem~\ref{thm:stat}. Finite $\widehat{\Theta}_\star^{(t)}$ therefore exist, and all feasible estimates with optimization error at most $\varepsilon_0$ satisfy $\dist^2(\widehat{\Theta}^{(t)},\Theta_{\star}^{(t)})\le C\nu x_{t,\delta}/(nt)\le\rho^2/\nu$, proving \eqref{eq:smallaverage}. After optimal common alignment, each individual worker or task-type parameter row has error at most $\sqrt \nu\dist(\widehat{\Theta}^{(t)},\Theta_{\star}^{(t)})\le\rho$. The true row bounds then prove the interior margins. Taking the infimum over unit vectors in $\norm{\bm{X}\bm{v}}_2\ge\norm{\bm{Y}\bm{v}}_2-\norm{(\bm{X}-\bm{Y})\bm{v}}_2$ gives the singular-value perturbation bound:
\begin{equation*}
  \frac{\sigma_d(\widehat{\bm{Z}}^{(t)})}{\sqrt n} \ge\sqrt\kappa-\frac{\norm{\widehat{\bm{Z}}^{(t)}-\bm{Z}_{\star}^{(t)}\bm{Q}}_{\op}}{\sqrt n} \ge\sqrt\kappa-\dist(\widehat{\Theta}^{(t)},\Theta_{\star}^{(t)}).
\end{equation*}
If also $\rho\le\sqrt\kappa/2$, the right side is at least $\sqrt\kappa/2$. Squaring proves the worker Gram bound.

To bound empirical residuals, the first and second blocks of \eqref{eq:H} are the true cumulative noise scaled by $\sqrt n/2$ and $\sqrt m$. Taking a block does not increase the operator norm, so:
\begin{equation}
  \label{eq:truecertificate} \frac{\norm{\bm{A}^{(t)}-\bm{\Theta}_A(\Theta_{\star}^{(t)})}_{\op}}t +\sqrt{\frac mn}\frac{\norm{\bm{B}^{(t)}-\bm{\Theta}_B(\Theta_{\star}^{(t)})}_{\op}}t \le\frac{3\norm{\bm{M}_{S_t}}_{\op}}{\sqrt n\,t} \le C\sqrt{\frac{\nu x_{t,\delta}}{nt}}.
\end{equation}
The last inequality uses \eqref{eq:matrixnoise} and $t\ge Cx_{t,\delta}$ to absorb the linear tail term. Common alignment preserves $\theta_{i\ell}^{A}$ and $\theta_{ik}^{B}$. Expanding at each entry and using the radius constraints gives:
\begin{equation*}
  \max_{i<\ell}\left|\log\frac{\theta_{i\ell}^{A}(\widehat{\Theta}_\star^{(t)})} {\theta_{\star,i\ell}^{A,(t)}}\right| \vee \max_{i,k}\left|\log\frac{\theta_{ik}^{B}(\widehat{\Theta}_\star^{(t)})} {\theta_{\star,ik}^{B,(t)}}\right| \le C\sqrt \nu\dist(\widehat{\Theta}_\star^{(t)},\Theta_{\star}^{(t)}).
\end{equation*}
The entrywise, row-sum, and Frobenius calculations used in \eqref{eq:residualchange} now give:
\begin{equation*}
  \frac{\norm{\bm{\Theta}_A(\widehat{\Theta}_\star^{(t)})-\bm{\Theta}_A(\Theta_{\star}^{(t)})}_{\op}}t +\sqrt{\frac mn}\frac{\norm{\bm{\Theta}_B(\widehat{\Theta}_\star^{(t)})-\bm{\Theta}_B(\Theta_{\star}^{(t)})}_{\op}}t \le C\sqrt \nu\dist(\widehat{\Theta}_\star^{(t)},\Theta_{\star}^{(t)}).
\end{equation*}
Adding \eqref{eq:truecertificate} bounds the empirical residual by $C\nu\sqrt{x_{t,\delta}/(nt)}$. Its square is $C\tilde\nu x_{t,\delta}/t$, which the preceding threshold calculation makes arbitrarily small. In particular, the residual bound can be made at most $b_0/4$.

\textit{Step 2: Stationarity, uniqueness, and initialization error.} Increasing the fixed threshold constant if necessary, we show that every $\widehat{\Theta}_\star^{(t)}$ with $t\ge t_0$ is unique up to a common orthogonal transformation and satisfies:
\begin{equation}
  \label{eq:stationarity} \nabla_{n,m}\{\mathcal{L}^{(t)}/(nt)\}(\widehat{\Theta}_\star^{(t)})=\bm{0}.
\end{equation}
An initial estimate with optimization error at most $\varepsilon_0$ also satisfies:
\begin{equation}
  \label{eq:initialdistance} \dist^2(\widehat{\Theta}^{(t_0)},\widehat{\Theta}_\star^{(t_0)}) \le C\varepsilon_0/(nt_0).
\end{equation}
Step 1 gives positive margins inside the radius constraints at every $\widehat{\Theta}_\star^{(t)}$. Thus, for any centered $\bm{h}$, the perturbation $\widehat{\Theta}_\star^{(t)}+s\bm{h}$ is feasible for sufficiently small positive and negative $s$. $\mathcal{L}^{(t)}$ along this line is minimized at $s=0$, so its derivative in $s$ vanishes there. Since the derivative vanishes for every centered $\bm{h}$, \eqref{eq:stationarity} follows.

Two exact minimizers $\widehat{\Theta}_{\star,1}^{(t)}$ and $\widehat{\Theta}_{\star,2}^{(t)}$ are within $2\rho/\sqrt{n+m}$ by \eqref{eq:smallaverage}. Choose $\rho<r_0/4$, optimally align them, and write $\bm{h}=\widehat{\Theta}_{\star,2}^{(t)}-\widehat{\Theta}_{\star,1}^{(t)}$. Lemma~\ref{lem:alignment} gives \eqref{eq:orthogonality} for $\bm{h}$. Step 1 verifies Lemma~\ref{lem:localgeometry} along the connecting segment. Stationarity and equal minimum values give:
\begin{equation*}
  0=\frac{\mathcal{L}^{(t)}(\widehat{\Theta}_{\star,2}^{(t)})-\mathcal{L}^{(t)}(\widehat{\Theta}_{\star,1}^{(t)})}{nt}\ge\frac\mu2\norm{\bm{h}}_{n,m}^2.
\end{equation*}
Thus $\bm{h}=\bm{0}$, proving uniqueness up to a common orthogonal transformation.

$\widehat{\Theta}^{(t_0)}$ and $\widehat{\Theta}_\star^{(t_0)}$ both satisfy \eqref{eq:smallaverage}. Their optimally aligned difference lies in the same local ball and satisfies \eqref{eq:orthogonality} at $\widehat{\Theta}_\star^{(t_0)}$. Repeating the integral lower bound proves \eqref{eq:initialdistance}:
\begin{equation*}
  \frac\mu2\dist^2(\widehat{\Theta}^{(t_0)},\widehat{\Theta}_\star^{(t_0)}) \le\frac{\mathcal{L}^{(t_0)}(\widehat{\Theta}^{(t_0)}) -\mathcal{L}^{(t_0)}(\widehat{\Theta}_\star^{(t_0)})}{nt_0} \le\frac{\varepsilon_0}{nt_0}.
\end{equation*}

\textit{Step 3: Objective increments and a common local region.} At the same $\Theta$, subtracting $\mathcal{L}^{(t-1)}$ from $\mathcal{L}^{(t)}$ gives:
\begin{equation}
  \label{eq:lossincrement} \mathcal{L}^{(t)}(\Theta)-\mathcal{L}^{(t-1)}(\Theta) =-\sum_{i<\ell}\Delta A_{i\ell}^{(t)}\log\theta_{i\ell}^{A}(\Theta) -\sum_{i,k}\Delta B_{ik}^{(t)}\log\theta_{ik}^{B}(\Theta).
\end{equation}
For $\Theta$ with fixed latent row bounds, the normalized gradient change satisfies:
\begin{equation}
  \label{eq:incrementgradient} \begin{aligned}
     &\left\|\nabla_{n,m}\{\mathcal{L}^{(t)}/(nt)\}(\Theta) -\frac{t-1}{t}\nabla_{n,m}\{\mathcal{L}^{(t-1)}/(n(t-1))\}(\Theta)\right\|_{n,m}^2\\
     &\le\frac C{t^2}\left[\frac1n\sum_i \left\{\sum_{\ell\ne i}\Delta A_{i\ell}^{(t)}+\sum_k\Delta B_{ik}^{(t)}\right\}^2 +\frac m{n^2}\sum_k\left\{\sum_i\Delta B_{ik}^{(t)}\right\}^2\right]. 
  \end{aligned}
\end{equation}
The exponential terms cancel in \eqref{eq:lossincrement}. Up to the minus sign, its gradient has blocks:
\begin{equation*}
  \left(\Delta\bm{A}^{(t)}\one_n,\ \Delta\bm{B}^{(t)}\one_m,\ (\Delta\bm{B}^{(t)})^\top\one_n,\ \Delta\bm{A}^{(t)}\bm{Z}+\Delta\bm{B}^{(t)}\bm{W},\ (\Delta\bm{B}^{(t)})^\top\bm{Z}\right).
\end{equation*}
Write $r_i^{(t)}=\sum_{\ell\ne i}\Delta A_{i\ell}^{(t)}+\sum_k\Delta B_{ik}^{(t)}$ and $c_k^{(t)}=\sum_i\Delta B_{ik}^{(t)}$. The scalar blocks are bounded by $r_i^{(t)}$ and $c_k^{(t)}$, and nonnegativity gives:
\begin{equation*}
  \begin{aligned}
    \left\|\bigl(\Delta\bm{A}^{(t)}\bm{Z}+\Delta\bm{B}^{(t)}\bm{W}\bigr)_{i\cdot}\right\|_2&\le\sum_{\ell\ne i}\Delta A_{i\ell}^{(t)}\norm{\bm{z}_\ell}_2+\sum_k\Delta B_{ik}^{(t)}\norm{\bm{w}_k}_2\le Cr_i^{(t)},\\
    \left\|\bigl((\Delta\bm{B}^{(t)})^\top\bm{Z}\bigr)_{k\cdot}\right\|_2&\le\sum_i\Delta B_{ik}^{(t)}\norm{\bm{z}_i}_2\le Cc_k^{(t)}.
  \end{aligned}
\end{equation*}
The weights in \eqref{eq:gradientmetric} give $(1/n)t^{-2}=1/(nt^2)$ for worker blocks and $(1/m)\{m/(nt)\}^2=m/(n^2t^2)$ for task-type blocks. Projection cannot increase the norm, proving \eqref{eq:incrementgradient}. The factor $(t-1)/t$ puts both objectives over $nt$.

For $t>t_0$, we next use \eqref{eq:batch} and the initial threshold to show $\dist(\widehat{\Theta}_\star^{(t)},\widehat{\Theta}_\star^{(t-1)})\le r_0/(4\sqrt{n+m})$. We also verify Lemma~\ref{lem:localgeometry} for $\mathcal{L}^{(t)}/(nt)$ with $\Theta_0=\widehat{\Theta}_\star^{(t-1)}$. With $\nu=n+m$, \eqref{eq:batch} gives $\sum_i(r_i^{(t)})^2\le K_1^2\nu$ and $\sum_k(c_k^{(t)})^2\le K_1^2n\nu/m$. By nonnegativity:
\begin{equation*}
  \norm{\Delta \bm{A}^{(t)}}_{\fro}^2 =2\sum_{i<\ell}(\Delta A_{i\ell}^{(t)})^2 \le\sum_i (r_i^{(t)})^2,\qquad \norm{\Delta \bm{B}^{(t)}}_{\fro}^2\le\sum_k (c_k^{(t)})^2.
\end{equation*}
The operator norm is at most the Frobenius norm, so:
\begin{equation}
  \label{eq:batchmatrix} \left\|\left[\frac{\sqrt n}{2}\Delta \bm{A}^{(t)},\sqrt m\Delta \bm{B}^{(t)}\right]\right\|_{\op}^2 \le \frac n4\sum_i (r_i^{(t)})^2+m\sum_k (c_k^{(t)})^2 \le Cn(n+m).
\end{equation}
Apply Lemma~\ref{lem:deterministic} to $\bm{A}^{(t)},\bm{B}^{(t)}$ with $u=t-1$ and $\Theta'=\Theta_{\star}^{(t-1)}$. The marginal scales are comparable because $1\le t/(t-1)\le2$, and the residual matrix is $\bm{M}_{S_{t-1}}+[\frac{\sqrt n}{2}\Delta\bm{A}^{(t)},\sqrt m\Delta\bm{B}^{(t)}]$. $\widehat{\Theta}_\star^{(t)}$ has zero optimization error. Equations~\eqref{eq:matrixnoise} and~\eqref{eq:batchmatrix} give:
\begin{equation}
  \label{eq:newoldtruth} \begin{aligned}
    \dist^2(\widehat{\Theta}_\star^{(t)},\Theta_{\star}^{(t-1)})&\le C\frac{\norm{\bm{M}_{S_{t-1}}}_{\op}^2+n\nu}{n(t-1)^2}\\
    &\le C\left\{\frac{\nu x_{t-1,\delta}}{n(t-1)}+\frac{\nu x_{t-1,\delta}^2}{n(t-1)^2}+\frac{\nu}{(t-1)^2}\right\}\\
    &=C\frac{\nu x_{t-1,\delta}}{n(t-1)}\left\{1+\frac{x_{t-1,\delta}}{t-1}+\frac{n}{(t-1)x_{t-1,\delta}}\right\}\le C\frac{\nu x_{t-1,\delta}}{n(t-1)}.
  \end{aligned}
\end{equation}
The last step uses $t-1\ge C\max\{n,x_{t-1,\delta}\}$ and $x_{t-1,\delta}\ge1$. $\widehat{\Theta}_\star^{(t-1)}$ satisfies the same bound relative to $\Theta_{\star}^{(t-1)}$. The triangle inequality and Step 1 give:
\begin{equation*}
  \begin{aligned}
    \dist(\widehat{\Theta}_\star^{(t)},\widehat{\Theta}_\star^{(t-1)})&\le\dist(\widehat{\Theta}_\star^{(t)},\Theta_{\star}^{(t-1)})+\dist(\widehat{\Theta}_\star^{(t-1)},\Theta_{\star}^{(t-1)})\\
    &\le C\sqrt{\frac{\nu x_{t-1,\delta}}{n(t-1)}}\le\frac{r_0}{4\sqrt\nu}.
  \end{aligned}
\end{equation*}

At $\widehat{\Theta}_\star^{(t-1)}$, $\theta_{i\ell}^{A}$ and $\theta_{ik}^{B}$ have scales $(t-1)/n$ and $(t-1)/m$, comparable to the new scales. Step 1 gives the radii, interior margins, and spectral bound. The new residual is the old residual minus $(\Delta\bm{A}^{(t)},\Delta\bm{B}^{(t)})$, so the left side of \eqref{eq:residualcertificate} is at most:
\begin{equation*}
  \frac{t-1}{t}\frac{b_0}{4} +\frac{\norm{\Delta \bm{A}^{(t)}}_{\fro}}{t} +\sqrt{\frac mn}\frac{\norm{\Delta \bm{B}^{(t)}}_{\fro}}{t} \le\frac{b_0}{4}+C\frac{\sqrt{n+m}}{t}.
\end{equation*}
The initial threshold also implies $t\ge C(n+m)\ge C\sqrt{n+m}$, making this bound at most $b_0$. Choose fixed $b_0,r_0,\mu,L$ covering both sets of mean comparison constants if needed. Lemma~\ref{lem:localgeometry} now applies to $\mathcal{L}^{(t)}/(nt)$ at $\Theta_0=\widehat{\Theta}_\star^{(t-1)}$.

\textit{Step 4: Displacement of adjacent minimizers.} Optimally align $\widehat{\Theta}_\star^{(t-1)}$ to $\widehat{\Theta}_\star^{(t)}$, retaining the notation, and set $\bm{h}=\widehat{\Theta}_\star^{(t)}-\widehat{\Theta}_\star^{(t-1)}$. By Lemma~\ref{lem:alignment}, $\bm{h}$ satisfies \eqref{eq:orthogonality} at $\widehat{\Theta}_\star^{(t-1)}$ and $\norm{\bm{h}}_{n,m}=\dist(\widehat{\Theta}_\star^{(t)},\widehat{\Theta}_\star^{(t-1)})$. Step 3 puts the connecting segment in the local ball of $\mathcal{L}^{(t)}/(nt)$ at $\widehat{\Theta}_\star^{(t-1)}$. Integrating the gradient difference gives:
\begin{equation*}
  \begin{aligned}
     &\left\langle \nabla_{n,m}\{\mathcal{L}^{(t)}/(nt)\}(\widehat{\Theta}_\star^{(t)}) -\nabla_{n,m}\{\mathcal{L}^{(t)}/(nt)\}(\widehat{\Theta}_\star^{(t-1)}),\bm{h} \right\rangle_{n,m}\\
     &=\int_0^1 D^2\{\mathcal{L}^{(t)}/(nt)\} (\widehat{\Theta}_\star^{(t-1)}+s\bm{h})[\bm{h},\bm{h}]\,ds \ge\mu\norm{\bm{h}}_{n,m}^2. 
  \end{aligned}
\end{equation*}
$\widehat{\Theta}_\star^{(t)}$ and $\widehat{\Theta}_\star^{(t-1)}$ have zero gradient for $\mathcal{L}^{(t)}$ and $\mathcal{L}^{(t-1)}$, respectively. Combining the zero gradient of $\widehat{\Theta}_\star^{(t-1)}$ with \eqref{eq:incrementgradient}, Cauchy--Schwarz yields:
\begin{equation*}
  \mu\norm{\bm{h}}_{n,m}^2 \le\left\|\nabla_{n,m}\{\mathcal{L}^{(t)}/(nt)\} (\widehat{\Theta}_\star^{(t-1)})\right\|_{n,m}\norm{\bm{h}}_{n,m}.
\end{equation*}
The bound is immediate if $\norm{\bm{h}}_{n,m}=0$. Otherwise, divide the preceding inequality by $\norm{\bm{h}}_{n,m}$ and square. Applying \eqref{eq:incrementgradient} and then Assumption~\ref{as:opt} gives, for every $t_0<t\le T$:
\begin{equation}
  \label{eq:driftmain} 
  \medmuskip=1mu\thickmuskip=2mu\thinmuskip=2mu \dist^2\!(\widehat{\Theta}_\star^{(t)},\widehat{\Theta}_\star^{(t-1)})\le\tfrac C{t^2}\Biggl[\tfrac1n\sum_i\!\left\{\sum_{\ell\ne i}\!\Delta A_{i\ell}^{(t)}+\sum_k\!\Delta B_{ik}^{(t)}\right\}^2+\tfrac m{n^2}\sum_k\!\left\{\sum_i\!\Delta B_{ik}^{(t)}\right\}^2\Biggr]\le\tfrac{C(n+m)}{nt^2}.
\end{equation}

\textit{Step 5: Gradient contraction and the inner iteration count.} In the local region from Steps 1--2, use $\Theta_0=\widehat{\Theta}_\star^{(t)}$. For a centered parameter $\Theta$ with $\dist(\Theta,\widehat{\Theta}_\star^{(t)})\le r_0/\sqrt{n+m}$, we show that one update in \eqref{eq:algorithm} contracts distance by a fixed $q_0\in(0,1)$. Choose $r_0$ below the fixed boundary margin as well.

Optimally align $\Theta$ to $\widehat{\Theta}_\star^{(t)}$, and write $\bm{h}$ for the resulting difference and $\bm{g}^{(t)}=\nabla_{n,m}\{\mathcal{L}^{(t)}/(nt)\}(\Theta)$. The difference $\bm{h}$ satisfies \eqref{eq:orthogonality} at $\widehat{\Theta}_\star^{(t)}$. A common orthogonal transformation preserves $\mathcal{L}^{(t)}$ and the local conditions. By stationarity, integration of \eqref{eq:hesslower} and \eqref{eq:hessupper} gives $\langle \bm{g}^{(t)},\bm{h}\rangle_{n,m}\ge\mu\norm{\bm{h}}_{n,m}^2$ and $\norm{\bm{g}^{(t)}}_{n,m}\le L\norm{\bm{h}}_{n,m}$. Since $\bm{P}$ is included in the gradient, the update is $\Theta-\eta \bm{g}^{(t)}$, and:
\begin{equation*}
  \begin{aligned}
    \dist^2(\Theta-\eta \bm{g}^{(t)},\widehat{\Theta}_\star^{(t)})&\le\norm{\bm{h}-\eta \bm{g}^{(t)}}_{n,m}^2\\
    &\le(1-2\eta\mu+\eta^2L^2)\norm{\bm{h}}_{n,m}^2=q_0^2\norm{\bm{h}}_{n,m}^2,
  \end{aligned}
\end{equation*}
where $\eta=\mu/(2L^2)$ and $q_0=\{1-3\mu^2/(4L^2)\}^{1/2}\in(0,1)$.

After contraction, distance remains below $r_0/\sqrt{n+m}$. Under some optimal orthogonal matrix, every latent row changes by at most $r_0$. The reference lies at least $3R_z/4,3R_w/4$ inside the two boundaries. Choosing $r_0\le\min\{R_z,R_w\}/4$ leaves margins of at least $R_z/2,R_w/2$. The update is therefore feasible. Choose a fixed integer $J$ with $q=q_0^J\le1/2$, for example $J\ge\log(1/2)/\log q_0$. All of $\mu,L,q_0,J,q$ depend only on fixed model constants.

\textit{Step 6: Remaining in the local region.} Write $e_t=\dist(\widehat{\Theta}^{(t)},\widehat{\Theta}_\star^{(t)})$ in this proof and subsequent references. Step 1 and a sufficiently large threshold constant ensure $e_{t_0}\le r_0/(4\sqrt{n+m})$. Bound~\eqref{eq:driftmain} gives adjacent-minimizer distance at most $C\sqrt{(n+m)/n}/t$. The threshold also implies $t_0\ge C(n+m)/\sqrt n$, so this distance can be made at most $r_0/(4\sqrt{n+m})$ by increasing a dimension-independent constant.

Suppose $e_{t-1}\le r_0/(4\sqrt{n+m})$. The previous output initializes the new checkpoint, and the triangle inequality places it within $r_0/(2\sqrt{n+m})$ of the new exact minimizer. Every inner step contracts by Step 5 and stays local. After $J$ steps, the factor is at most $q\le1/2$, restoring $e_t\le r_0/(4\sqrt{n+m})$. Induction proves that every checkpoint output and inner iterate stays in the local region and inside the radius constraints.

\textit{Step 7: Tracking recursion.} Iterating $e_t\le qe_{t-1}+Cq\sqrt{(n+m)/n}/t$ gives:
\begin{equation*}
  e_t\le q^{t-t_0}e_{t_0}+C\sqrt{\frac{n+m}{n}}\sum_{s=t_0+1}^t\frac{q^{t-s}}s\le q^{t-t_0}e_{t_0}+\frac{C\sqrt{(n+m)/n}}{t(1-q)^2}.
\end{equation*}
The second inequality uses $t/s\le t-s+1$ and $\sum_{s=t_0+1}^t(t-s+1)q^{t-s}\le\sum_{j=0}^{\infty}(j+1)q^j=(1-q)^{-2}$. Squaring with $(a+b)^2\le2a^2+2b^2$ and applying \eqref{eq:initialdistance} proves \eqref{eq:optrate}.

\textit{Step 8: Subsequent global optimization errors.} Align $\widehat{\Theta}^{(t)}$ to $\widehat{\Theta}_\star^{(t)}$. Their segment lies in the local region, where the reference gradient is zero. Integrating the Hessian upper bound \eqref{eq:hessupper} gives:
\begin{equation}
  \label{eq:gapquadratic} 0\le\mathcal{L}^{(t)}(\widehat{\Theta}^{(t)}) -\mathcal{L}^{(t)}(\widehat{\Theta}_\star^{(t)}) \le\frac{Lnt}{2}e_t^2.
\end{equation}
Substitute Step 7 and $e_{t_0}^2\le C\varepsilon_0/(nt_0)$ to obtain the first upper bound in \eqref{eq:gapderived}. Since $t_0\ge1$:
\begin{equation*}
  \begin{aligned}
    \frac t{t_0}q^{2(t-t_0)}&=\{1+(t-t_0)/t_0\}q^{2(t-t_0)}\le\{1+(t-t_0)\}q^{2(t-t_0)}\\
    &\le\sum_{j=0}^{\infty}(1+j)q^{2j}=\frac1{(1-q^2)^2}.
  \end{aligned}
\end{equation*}
The initial threshold also gives $(n+m)/t\le(n+m)/t_0\le C$, with a dimension-independent constant. Thus $\mathcal{L}^{(t)}(\widehat{\Theta}^{(t)})-\mathcal{L}^{(t)}(\widehat{\Theta}_\star^{(t)})\le C(1+\varepsilon_0)$ uniformly. All steps are deterministic on the same statistical event, and Assumption~\ref{as:opt} imposes \eqref{eq:batch} at every checkpoint $t_0<t\le T$, so no additional failure probability is incurred.
\end{proof}

\emph{A sufficient condition for Assumption~\ref{as:opt}.} Let a fixed $K_*$ bound all checkpoint increments:
\begin{equation*}
  \max_i\left\{\sum_{\ell\ne i}\Delta A_{i\ell}^{(t)} +\sum_k\Delta B_{ik}^{(t)}\right\}\le K_*, \qquad \max_k\sum_i\Delta B_{ik}^{(t)}\le K_*(1+n/m).
\end{equation*}
Then \eqref{eq:batch} holds with $K_1=\sqrt2K_*$. Using $r_i^{(t)},c_k^{(t)}$ from Step 3, nonnegativity gives $\sum_k c_k^{(t)}=\sum_{i,k}\Delta B_{ik}^{(t)}\le\sum_i r_i^{(t)}\le nK_*$. Consequently:
\begin{equation*}
  \frac1n\sum_i (r_i^{(t)})^2\le K_*^2,\qquad \frac m{n^2}\sum_k (c_k^{(t)})^2 \le\frac m{n^2}(\max_k c_k^{(t)})\sum_k c_k^{(t)} \le K_*^2(1+m/n).
\end{equation*}
The sum is at most $2K_*^2(1+m/n)$, proving \eqref{eq:batch}. The converse fails: when $m=n$, a single $\bm{A}$ edge of weight $\lfloor\sqrt n\rfloor$ gives a weighted squared sum at most two, while its maximum row increment grows with dimension.

\emph{Local stopping error.} In the local region of Step 5, let $\bm{g}^{(t)}=\nabla_{n,m}\{\mathcal{L}^{(t)}/(nt)\}(\widehat{\Theta}^{(t)})$. Then:
\begin{equation*}
  \dist(\widehat{\Theta}^{(t)},\widehat{\Theta}_\star^{(t)})\le\norm{\bm{g}^{(t)}}_{n,m}/\mu, \qquad 0\le\mathcal{L}^{(t)}(\widehat{\Theta}^{(t)})-\mathcal{L}^{(t)}(\widehat{\Theta}_\star^{(t)}) \le\frac{Lnt}{2\mu^2}\norm{\bm{g}^{(t)}}_{n,m}^2.
\end{equation*}
For the aligned difference $\bm{h}$, Step 5 and Cauchy--Schwarz give:
\begin{equation*}
  \mu\norm{\bm{h}}_{n,m}^2\le\langle\bm{g}^{(t)},\bm{h}\rangle_{n,m}\le\norm{\bm{g}^{(t)}}_{n,m}\norm{\bm{h}}_{n,m}.
\end{equation*}
For $\bm{h}\ne\bm{0}$, divide by $\norm{\bm{h}}_{n,m}$ and apply \eqref{eq:gapquadratic}. For $\bm{h}=\bm{0}$, both bounds hold directly.

For the numerical normalization, let $n_{\mathrm{pair}}=n(n-1)/2+nm$. If $\|\bm{P}\nabla\{\mathcal{L}^{(t)}/n_{\mathrm{pair}}\}(\widehat{\Theta}^{(t)})\|_\infty\le\varepsilon$, the $n(d+2)$ worker and $m(d+1)$ task-type coordinates give:
\begin{equation*}
  \norm{\bm{g}^{(t)}}_{n,m}^2\le\frac{n_{\mathrm{pair}}^2\varepsilon^2}{t^2}\left\{d+2+(d+1)\frac{m^2}{n^2}\right\}.
\end{equation*}
These bounds translate the projected gradient of the unpenalized objective into local distance and optimization errors under the stated local conditions. The penalized objective and stopping rule used in the experiments are described in Section~\ref{app:fit}.

\subsection{Proof of Corollary~\ref{cor:computed}}\label{sec:computedproof}
\begin{proof}
Allocate failure probability $\delta/2$ to Theorem~\ref{thm:optimization} and $\delta/2$ to the event of Lemma~\ref{lem:concentration} for $\lceil aT\rceil\le t\le T$, including \eqref{eq:finiteM}. Their intersection has probability at least $1-\delta$, without an independence requirement. Replacing $\log\{(2n+m)/\delta\}$ by $\log\{2(2n+m)/\delta\}$ adds only $\log2$, absorbed into dimension-independent constants. On this intersection, all $\widehat{\Theta}^{(t)}$ are feasible and satisfy $\mathcal{L}^{(t)}(\widehat{\Theta}^{(t)})-\mathcal{L}^{(t)}(\widehat{\Theta}_\star^{(t)})\le C(1+\varepsilon_0)$. The optimization event also gives the marginal bounds \eqref{eq:obsmargins} for $\lceil aT\rceil\le t\le T$. Write $x=\log\{(2n+m)/\delta\}$. Applying Lemma~\ref{lem:deterministic} with $u=t$, $\Theta'=\Theta_{\star}^{(t)}$, and the fixed error bound gives, by \eqref{eq:finiteM} and $t\ge aT$:
\begin{equation*}
  \begin{aligned}
    \max_{\lceil aT\rceil\le t\le T}\dist^2(\widehat{\Theta}^{(t)},\Theta_{\star}^{(t)})&\le C_a\left\{\frac{(n+m)(Tx+x^2)}{nT^2}+\frac{1+\varepsilon_0}{nT}\right\}\\
    &\le C_a\frac{(n+m)x}{nT}.
  \end{aligned}
\end{equation*}
The initial threshold gives $T\ge Cx$, and $x\ge1$ absorbs the fixed optimization error. The same argument controls \eqref{eq:lograte}, proving \eqref{eq:computedtail}.

For consistency, write $\nu=n+m$, $s_*=2n+m$, and let $b_*$ be the fixed threshold constant. For a fixed cutoff $m_*>4$, each fixed $\delta\in(0,1)$ eventually gives, along the stated sequence:
\begin{equation*}
  \frac{\nu\log s_*}{nT}\le\frac{\log s_*}{b_*\nu\log(s_*/\delta)}\le\begin{cases}(b_*m_*)^{-1},&\nu\ge m_*,\\ \log(2m_*)/\{b_*\log(1/\delta)\},&\nu<m_*.\end{cases}
\end{equation*}
Given $\varepsilon>0$, choose fixed $m_*>4$ with $(b_*m_*)^{-1}<\varepsilon$, then fixed $\delta$ making the second bound smaller than $\varepsilon$. Hence $\nu\log s_*/(nT)\to0$ for every such sequence.

The explicit sufficient condition $\{(n+m)^2\log(2n+m)\}/(nt_0)\to0$ ensures the initial threshold for every fixed $\delta$, since $\log(s_*/\delta)\le C_\delta\log s_*$ for a fixed $C_\delta$. It also gives directly:
\begin{equation*}
  \frac{\nu\log s_*}{n T} \le\frac1{\nu}\frac{\nu^2\log s_*}{n t_0} \longrightarrow0.
\end{equation*}
Together with $t_0\le\lceil aT\rceil$, the explicit condition is the sufficient condition stated after Corollary~\ref{cor:computed}. For every fixed $\delta$, $\log\{(2n+m)/\delta\}$ and $\log(2n+m)$ differ by at most a fixed factor, so \eqref{eq:computedtail} gives uniform consistency in probability over $\lceil aT\rceil\le t\le T$. If initialization succeeds with probability only $1-\delta_{\mathrm{init}}$, add its failure probability. When $\delta_{\mathrm{init}}\to0$, allocate half of any target failure probability to concentration. The initialization failure is eventually below the other half, yielding unconditional consistency.
\end{proof}

\section{Additional Experiment Details: LongBench v2 Multi-Doc QA}
\label{app:details}
\begingroup
\renewcommand{\topfraction}{.9}
\renewcommand{\bottomfraction}{.8}
\renewcommand{\textfraction}{.1}
\renewcommand{\floatpagefraction}{.8}
\clubpenalty=10000
\widowpenalty=10000

This section gives the LongBench v2 Multi-Doc QA details supporting Section~\ref{sec:experiment}. Section~\ref{app:design} describes the experimental design and the participation and collaboration counts. Section~\ref{app:prompts} presents the core worker prompts for the three guidelines. Section~\ref{app:outcomes} reports accuracy, published work, and answer changes after review. Section~\ref{app:fit} gives the numerical estimation procedure and computational resources.

\subsection{Experimental Design and Network Counts}
\label{app:design}
\label{app:classification}

Workers begin the warm-up sessions with empty conversations and keep their own conversations in each branch. The branches then develop separately. The order of questions and the initial order of workers are fixed across replicates, while model responses are stochastic. Model calls use Amazon Bedrock. All questions and their sources, divided into 7,191 passages, are accessible from the start. Each session allows 40 board requests and 20 minutes under every guideline. Up to 32 sessions run concurrently, with one active session per worker. At least 10 seconds separate different workers' first activations so later entrants can inspect work already published. Conversations undergo automatic context compaction when needed.

The shared board lists up to 16 open sub-tasks. Workers may select one, initiate or continue work, or invite a collaborator. Publishing an answer submits it for scoring, and a revision identifies the answer it replaces. Section~\ref{app:prompts} gives the shared rules for evidence, use of published results of other workers, and review.

For Base and Pro, the MAS estimates the statistical summaries after 64, 96, 128, 160, 192, 224, and 256 valid sessions. The intermediate estimates include work already published by active sessions and run in the background, and the final estimate follows the end of all sessions. We estimate the History-only networks after completion for descriptive comparison.

Each published work records its author, sources, and the published results it uses or reviews. Table~\ref{tab:types} defines the four task types. A fixed classification procedure checks the required content without access to the benchmark key.

\begin{table}[!htbp]
\centering\small\spacingset{1}
\caption{Published results associated with each task type.}
\label{tab:types}
\begin{tabularx}{\textwidth}{@{}lX@{}}
\toprule
Task type & Result\\
\midrule
Extract & Facts or values located in the sources, or a documented search result.\\
Compare & Evidence-supported similarities, differences, or conflicts between specified objects.\\
Inference & A conclusion with its premises and reasoning. An answer also explains its use of published results of other workers.\\
Review & Checks of a specified version of a published work, supporting evidence, and a PASS, FAIL, or UNRESOLVED verdict.\\
\bottomrule
\end{tabularx}
\end{table}

A published work contributes at most once to each task type in $\bm B^{(t)}$. It may contribute to several types if it contains independently reusable results of those types. Supporting quotations and self-checks add no types, and publishing the same work again adds no counts. Using or reviewing another worker's published result contributes once per work and worker pair to $\bm A^{(t)}$. Reads and invitations alone add no collaboration counts.

\subsection{Core Prompts for the Three Guidelines}
\label{app:prompts}

The following boxes condense the core experimental prompts, and Section~\ref{sec:benchmarks-evaluation} describes the three guidelines. Source-first review permits disclosed prior exposure to an answer and is therefore not necessarily blind. Automated checks enforce the required records and review sequence. Workers judge the evidence. The complete prompts accompany the manuscript. The prompt text calls board requests shared-board actions and describes specialization as the continuation of useful methods across compatible questions.

\newtcolorbox{coreprompt}[1]{%
  enhanced,breakable,
  colback=black!1,colframe=black!45,colbacktitle=black!7,coltitle=black,
  boxrule=.45pt,arc=0pt,
  left=8pt,right=8pt,top=7pt,bottom=7pt,
  before skip=10pt,after skip=10pt,
  title={#1},title after break={#1\ (continued)},
  fonttitle=\small\bfseries,
  fontupper=\small,
  before upper={\spacingset{1.12}\setlength{\parindent}{0pt}\setlength{\parskip}{5pt}},
  pad at break*=2mm
}

\begin{coreprompt}{Shared instructions}
Address the 125 questions using the supplied sources. Each session has at most 40 shared-board actions and 20 minutes. Prioritize unanswered questions and missing contributions. Use extraction, comparison, inference, and review as needed. No worker is assigned a permanent role. Continue useful methods across compatible questions, and change direction when the work or feedback calls for it.

A first answer must use a relevant published result from another worker and explain its contribution. Submit a selected option, confidence, and supporting reasoning. Identify missing evidence or ambiguous options. If the evidence cannot distinguish tied options, choose the first tied option in the original order and report low confidence. Submission does not require a PASS review.

When reviewing a submitted answer, choose a solution that you did not help construct, compare, or correct. Disclose any earlier exposure to the answer or reasoning. Within each answer review, read the original sources and record a preliminary judgment before opening the target answer. Then compare the specified answer version with that judgment and the sources. Retain the preliminary judgment and explain any later change. If your preliminary choice differs from the target but you conclude PASS, identify the error in your earlier premise and explain how the source evidence changes the answer. Use PASS to accept the answer, FAIL for a demonstrated defect, and UNRESOLVED when the evidence does not settle an issue.
\end{coreprompt}

\begin{coreprompt}{History only}
Choose task types, question batches, inputs, and collaborators using actual work records, retained methods, current needs, and workloads. No fitted statistical summary is available.
\end{coreprompt}

\begin{coreprompt}{AgentCue Base}
Before choosing work or collaborators in each session, read the complete current statistical summary once and use that version throughout the session. Use the fitted participation and collaboration means to narrow the choices. Then check current needs, workloads, and the actual inputs.

The summary and its activity/affinity explanations describe historical participation and cooperation. They do not measure competence or answer correctness. Use current needs and work history when the summary gives weak support or does not help the choice.
\end{coreprompt}

\begin{coreprompt}{AgentCue Pro: task and partner choice}
Use the same complete statistical summary as Base. Compare full fitted means and use their activity and affinity components to interpret the differences. Do not replace the fitted mean with an affinity-only score or a distance between latent positions. The largest fitted mean does not require you to choose that task or partner.

First identify the input task type your work needs, which may differ from your current task type. For a candidate who could supply the needed task type, consider both its fitted participation in that type and its fitted collaboration with you. An inference task may need extraction or comparison input. Resolve competing choices through relevance, availability, and the candidate's actual results. No combined score is prescribed.

Use your fitted participation profile alongside your methods and current needs. When a useful focus, the summary, and available work agree, continue with a compatible question. Do not rotate only to equalize counts. Later summaries reflect your own earlier choices, so repetition does not independently confirm quality. Change direction when work is crowded, results are unusable, a method fails, or an uncovered gap needs attention.
\end{coreprompt}

\begin{coreprompt}{AgentCue Pro: additional review checks}
Before opening the target answer, record your own confidence, the original question's critical constraints, and one to three concrete verification questions. Confidence is an uncalibrated self-report. Include the uncertain premise, strongest alternative, or a challenge that could change your choice. Even high confidence requires a check.

After reading the target, return to the sources to check these questions and constraints, low or missing author confidence, material problems, tied options, and disagreement with your preliminary judgment. You may reread an earlier source, but you must perform a new verification. Compare the alternatives against the original question and record source quotations with the findings. For each resolved issue, explain which interpretation the evidence supports and which alternative it rules out.

If you reverse a preliminary judgment made with high confidence, reread and cite an original source for that judgment and explain the error. Retain any remaining ambiguity. An unresolved or contradicted issue prevents PASS. Use FAIL for a demonstrated defect and UNRESOLVED when the evidence does not settle the issue. All checks use the same session budget as the other guidelines.
\end{coreprompt}

\subsection{Additional Results}
\label{app:outcomes}
\label{app:review}

Table~\ref{tab:per-execution} reports final accuracy for each replicate. The intermediate accuracy measurements in Figure~\ref{fig:progress}(b) also include answers submitted by sessions still in progress. Table~\ref{tab:participation} summarizes participation and collaboration, with collaboration counts summed over the upper triangle of $\bm A$. Figure~\ref{fig:published-work} shows published work by replicate. Pro produces fewer works and review results on average than Base. In the Pro branch of R1 in Figure~\ref{fig:workflow}, the warm-up sessions produce 187 works and 101 distinct worker pairs.

\begin{table}[!htbp]
\centering\small
\spacingset{1}\setlength{\tabcolsep}{4.5pt}
\caption{Final correct answers out of 125, with accuracy in parentheses, under each guideline.}
\label{tab:per-execution}
\begin{tabular}{@{}lccc@{}}
\toprule
Replicate & AgentCue Pro & AgentCue Base & History only\\
\midrule
R1 & 74 (59.2\%) & 71 (56.8\%) & 71 (56.8\%)\\
R2 & 76 (60.8\%) & 69 (55.2\%) & 71 (56.8\%)\\
R3 & 74 (59.2\%) & 74 (59.2\%) & 72 (57.6\%)\\
R4 & 74 (59.2\%) & 71 (56.8\%) & 72 (57.6\%)\\
R5 & 74 (59.2\%) & 74 (59.2\%) & 71 (56.8\%)\\
\bottomrule
\end{tabular}
\end{table}

\begin{table}[!htbp]
\centering\small
\spacingset{1}\setlength{\tabcolsep}{4.5pt}
\caption{Participation and collaboration under each guideline. Values are means $\pm$ sample standard deviations across five replicates. Counts cover all 256 sessions unless the 192 continuation sessions are specified.}
\label{tab:participation}
\begin{tabular}{@{}lrrr@{}}
\toprule
Measure & AgentCue Pro & AgentCue Base & History only\\
\midrule
Published works & 548.8 $\pm$ 13.8 & 567.4 $\pm$ 13.1 & 576.8 $\pm$ 20.1\\
Published works (continuation sessions) & 365.0 $\pm$ 12.8 & 383.6 $\pm$ 7.7 & 393.0 $\pm$ 16.4\\
Collaboration count & 542.6 $\pm$ 17.6 & 537.6 $\pm$ 14.0 & 557.2 $\pm$ 17.3\\
Distinct worker pairs & 317.0 $\pm$ 17.4 & 316.0 $\pm$ 7.4 & 320.6 $\pm$ 5.1\\
Extract results & 145.8 $\pm$ 12.6 & 156.6 $\pm$ 12.9 & 156.4 $\pm$ 11.2\\
Compare results & 135.8 $\pm$ 7.5 & 135.2 $\pm$ 7.0 & 136.6 $\pm$ 8.8\\
Inference results & 139.4 $\pm$ 2.9 & 135.8 $\pm$ 4.2 & 136.2 $\pm$ 1.6\\
Review results & 129.6 $\pm$ 16.6 & 141.8 $\pm$ 13.1 & 149.8 $\pm$ 11.6\\
\bottomrule
\end{tabular}
\end{table}

\begin{figure}[!t]
\centering
\spacingset{1}\captionsetup{font=small}
\includegraphics[width=\textwidth,height=.80\textheight,keepaspectratio]{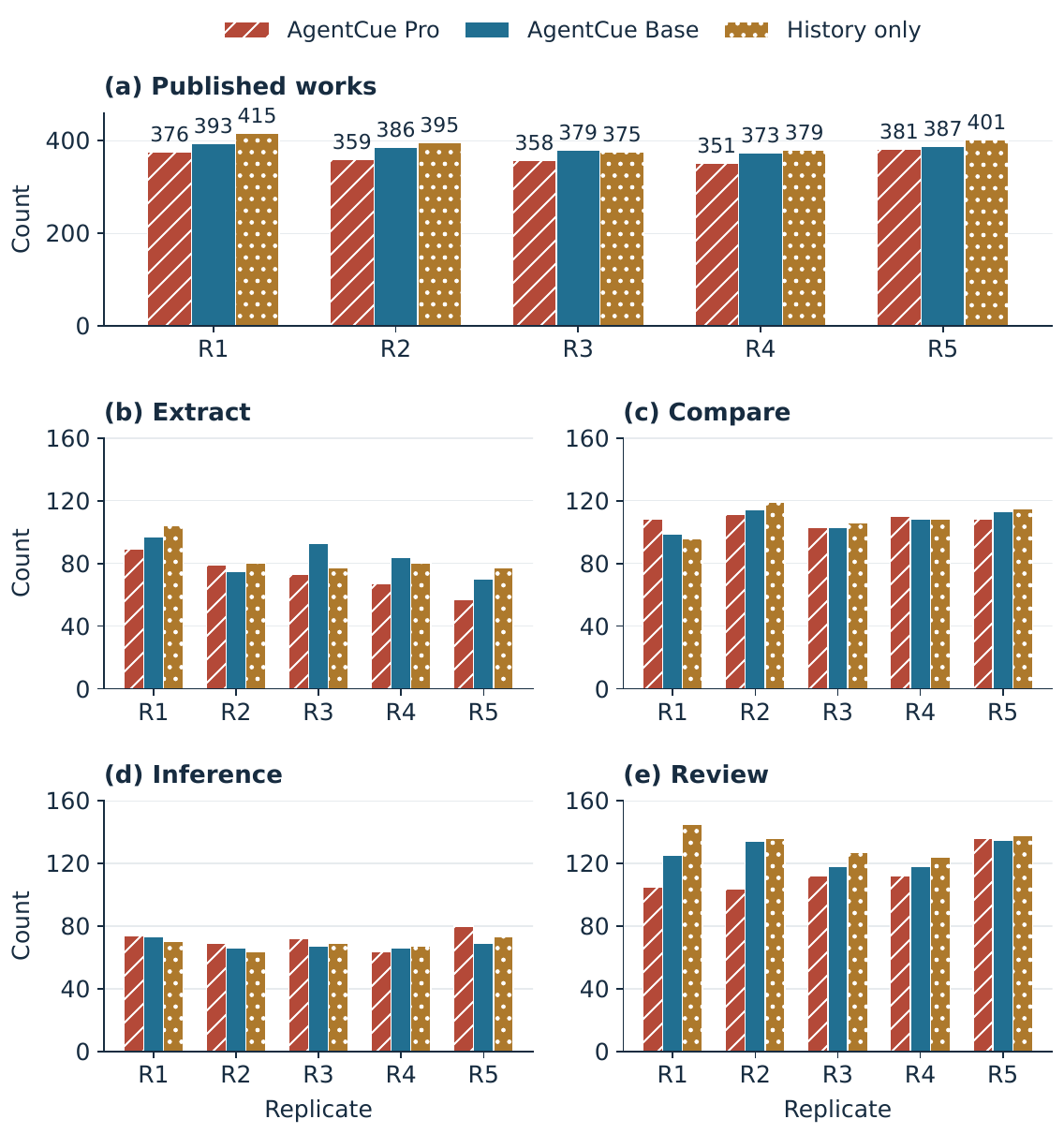}
\caption{Published work during the 192 continuation sessions in each replicate. Panel (a) counts published works. Panels (b)--(e) count results of each task type and share a vertical scale. A work may contribute to more than one task type, so the type counts need not sum to the number of works. All 32 workers publish work in every branch.}
\label{fig:published-work}
\end{figure}

We link a revision to a review when it adopts a specific finding about the answer it replaces. Of the 83 changes in answer correctness after warm-up, 72 meet this criterion. The remaining 11 changes comprise five wrong-to-right and six right-to-wrong changes. The review-change counts in Figure~\ref{fig:review-changes} and Table~\ref{tab:outcomes} include the 69 changes linked to reviews initiated after warm-up. Each R1 branch also has one wrong-to-right change based on a warm-up review. These three changes are excluded from the review-change counts. Two corrections completed during warm-up are already part of the shared starting answers in R3 and R4. Within each branch, a question changes correctness at most once and retains the revised state in its final answer.

\begin{figure}[!htbp]
\centering
\spacingset{1}\captionsetup{font=small}
\includegraphics[width=\textwidth]{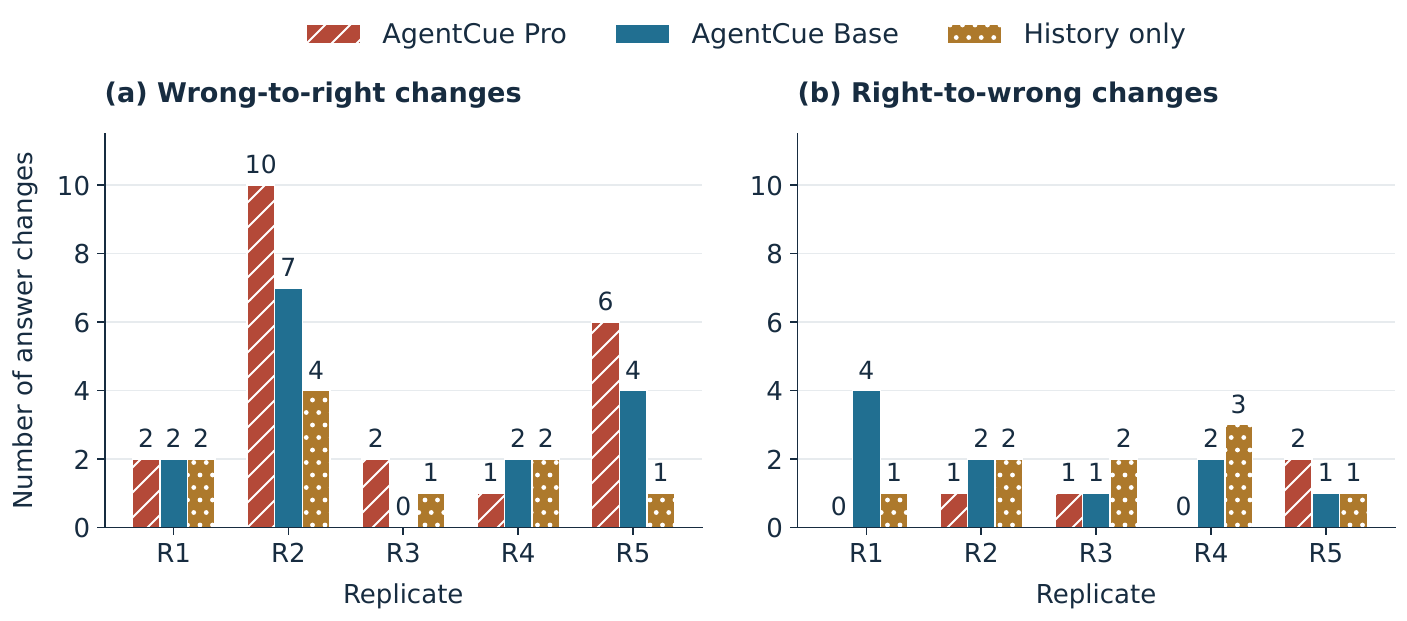}
\caption{Answer changes that adopt findings from reviews initiated after the warm-up sessions. Bars show counts in each replicate for (a) wrong-to-right and (b) right-to-wrong changes. Both panels use the same vertical scale. Changes based on warm-up reviews and revisions without verified adoption of a review finding are excluded.}
\label{fig:review-changes}
\end{figure}

Two answer revisions under Pro in R2 illustrate the review outcomes. For a video pretraining question, the initial answer confuses an interface that represents human actions without loss with a model of their full joint distribution. A reviewer identifies the distinction in the sources, and a later worker uses this finding to change the answer from A to the correct D. In the Alibaba question, a worker adopts its own FAIL review and changes the answer from correct C to incorrect B.

Table~\ref{tab:review-cost} gives review and revision effort, including unsuccessful and unfinished attempts. For each guideline, board requests per wrong-to-right change divide total review and revision board requests by the number of wrong-to-right changes linked to reviews initiated after warm-up. Both totals are pooled across the five replicates.\par

\begin{table}[!htbp]
\centering\small
\spacingset{1}\setlength{\tabcolsep}{4.5pt}
\caption{LongBench review effort across five replicates, including unsuccessful attempts. The last row divides total review and revision board requests by wrong-to-right changes that adopt findings from reviews initiated after warm-up.}
\label{tab:review-cost}
\begin{tabularx}{\textwidth}{@{}>{\raggedright\arraybackslash}Xccc@{}}
\toprule
Measure & AgentCue Pro & AgentCue Base & History only\\
\midrule
Completed reviews & 569 & 630 & 670\\
Review board requests & 7826 & 7062 & 7662\\
Revision board requests & 360 & 216 & 216\\
Review worker hours & 18.90 & 13.32 & 15.16\\
Revision worker hours & 0.83 & 0.49 & 0.50\\
Board requests per wrong-to-right change & 389.8 & 485.2 & 787.8\\
\bottomrule
\end{tabularx}
\end{table}

\subsection{Numerical Estimation and Computational Resources}
\label{app:fit}
\label{app:resources}

We estimate the parameters of the latent space model \eqref{eq:joint} with $d=2$ for all 32 workers and four task types, including zero-count pairs. Let $\bm x$ stack the entries of $\Theta^{(t)}$. The penalized objective is:
\begin{equation}
 f^{(t)}(\bm x)=\frac{\mathcal L^{(t)}(\bm x)}{624}+\frac{0.01}{32}\norm{\bm Z}_{\fro}^2+\frac{0.01}{4}\norm{\bm W}_{\fro}^2,\qquad 624=\binom{32}{2}+32\times4.
\label{eq:empirical-objective}
\end{equation}
Penalty coefficients are fixed throughout. Let $\bm P$ subtract column means from $\bm Z$ and $\bm W$ and the mean from $\bm\beta$, leaving $\bm\alpha$ and $\bm\gamma$ unchanged. Projected updates are:
\begin{equation}
 \bm x_{s+1}=\bm P\!\left\{\bm x_s-\eta_s\bm P\nabla f^{(t)}(\bm x_s)\right\}.
\label{eq:projected-update}
\end{equation}
Armijo backtracking selects $\eta_s$ in \eqref{eq:projected-update} by halving a trial step size until the sufficient-decrease condition with coefficient $10^{-4}$ holds. The first trial step size is $1$. Write $\Delta\bm x_s=\bm x_{s+1}-\bm x_s$ and $\Delta\bm g_s=\bm P\nabla f^{(t)}(\bm x_{s+1})-\bm P\nabla f^{(t)}(\bm x_s)$. Later trial step sizes alternate between the Barzilai--Borwein values $\Delta\bm x_s^\top\Delta\bm x_s/(\Delta\bm x_s^\top\Delta\bm g_s)$ and $\Delta\bm x_s^\top\Delta\bm g_s/(\Delta\bm g_s^\top\Delta\bm g_s)$, restricted to $[10^{-6},10^4]$. If $\Delta\bm x_s^\top\Delta\bm g_s\le10^{-20}$, the trial step size is $1.5\eta_s$, capped at $10^4$.

The first checkpoint uses five random starts with centered $N(\bm0,\bm I_2)$ latent positions, and node effects initialized from the observed network counts. The smallest finite penalized objective selects the estimate. Later checkpoints start from all parameter blocks of the previous estimate. Each start allows 20,000 iterations and stops when $\|\bm P\nabla f^{(t)}(\bm x_s)\|_\infty\leq10^{-5}$. There are 95 estimates: five warm-up estimates and 30 later estimates for each guideline. All 115 optimization starts reach the gradient tolerance (Table~\ref{tab:diagnostics}).

\begin{remark}
Under their assumptions, Theorem~\ref{thm:optimization} and Corollary~\ref{cor:computed} give error bounds for the unpenalized updates in \eqref{eq:algorithm} with prescribed parameter-specific step sizes and a fixed iteration count. They explain how initialization from the preceding checkpoint supports tracking as the networks accumulate work records. Our implementation combines this initialization with the ridge penalty in \eqref{eq:empirical-objective}, adaptive scalar Barzilai--Borwein step sizes, Armijo backtracking, and gradient-based stopping. The reported projected-gradient norms measure first-order stationarity of the penalized objective. The theoretical error bounds apply only to the stated unpenalized algorithm.
\end{remark}

\begin{table}[!htbp]
\centering\footnotesize
\spacingset{1}\setlength{\tabcolsep}{4.5pt}
\caption{Numerical estimation across five replicates. Iterations give ranges over starts. The projected gradient determines the stopping criterion. Seconds sum the elapsed time of all estimates in each row, including count reconstruction and all starts.}
\label{tab:diagnostics}
\begin{tabular}{@{}lrrcrr@{}}
\toprule
Source & Estimates & Starts & Iterations & Max. $\|\bm P\nabla f\|_\infty$ & Seconds\\
\midrule
Warm-up & 5 & 25 & 112--647 & $1.00\times10^{-5}$ & 2.387\\
Base & 30 & 30 & 50--317 & $1.00\times10^{-5}$ & 1.427\\
Pro & 30 & 30 & 43--202 & $9.90\times10^{-6}$ & 1.273\\
History only (after completion) & 30 & 30 & 43--197 & $9.89\times10^{-6}$ & 1.367\\
\bottomrule
\end{tabular}
\end{table}

Table~\ref{tab:resources} reports overall resources across five replicates. Counting the warm-up sessions once gives 640 valid sessions per replicate and 3,200 in total. Worker hours in Table~\ref{tab:review-cost} sum individual review and revision times, including model reasoning and tool latency. Concurrent work contributes separately. Elapsed time in Table~\ref{tab:resources} runs from the first session start to the last session end within each phase. It excludes gaps between phases and estimation after worker activity. Board requests include interactions that the limit of 40 per session does not count. Token totals cover whole sessions, which may include review and other tasks. Cached input is included in the input-token total. Table~\ref{tab:diagnostics} reports estimation time.

\begin{table}[!htbp]
\centering\small
\spacingset{1}\setlength{\tabcolsep}{4pt}
\caption{Resources across five replicates. Minutes are mean $\pm$ sample standard deviation of phase elapsed time. Other columns are totals. M denotes million tokens.}
\label{tab:resources}
\begin{tabular}{@{}lrrrrrr@{}}
\toprule
Phase & Sessions & Board requests & Input (M) & Cached (M) & Output (M) & Minutes\\
\midrule
Warm-up & 320 & 11,375 & 1043.58 & 1006.79 & 2.68 & 7.5 $\pm$ 0.3\\
Base & 960 & 25,376 & 3372.73 & 3247.48 & 6.85 & 26.0 $\pm$ 2.4\\
Pro & 960 & 25,205 & 3518.97 & 3383.90 & 8.60 & 27.1 $\pm$ 2.9\\
History only & 960 & 25,031 & 3120.65 & 2999.10 & 6.41 & 25.4 $\pm$ 1.3\\
\bottomrule
\end{tabular}
\end{table}

\section{Additional Experiment Details: FATE-H}
\label{app:fateh}

This section gives the FATE-H design and network counts, guideline prompts and proof verification, additional results and examples, and numerical estimation and verification scope.

\subsection{Benchmark, Experimental Design, and Network Counts}
\label{app:fateh-design}

The experiment uses all 100 original algebraic theorems in FATE-H \citep{jiang2026fate}. The formal statements and proof environment are fixed throughout. Workers cannot weaken an original theorem or remove an original theorem from the denominator. Proofs are checked with Lean~4.28.0 and a fixed Mathlib revision recorded in the experimental materials. All workers can access the same theorem pool, search the fixed Mathlib index, publish intermediate work, and inspect exact versions of published results.

Each of the five replicates contains 32 workers that use GPT-5.6 Sol with high reasoning effort through Amazon Bedrock. The task-order and worker-activation seeds differ across replicates. Within a replicate, the branches execute in the fixed order Base, Pro, and History only. After the 64 warm-up sessions, the three branches continue separately from the same shared board and worker conversations for 192 continuation sessions each. Each branch therefore has 256 valid sessions. Counting the warm-up sessions once gives 640 valid sessions per replicate and 3,200 across replicates.

Each session permits at most 100 tool calls and 10 Lean checks. The Lean execution limit is 900 seconds, excluding queue time and including any permitted retry. These limits follow ScaleMAS \citep{qi2026scalemas}. Ordinary unsuccessful mathematical work still yields a valid session. An invalid session releases its completion slot, and the records keep its outputs and resource consumption. The experiment contains 3,225 sessions, of which 25 are invalid.

Workers choose among four task types according to the current mathematical need (Table~\ref{tab:fateh-types}). All four can contribute to a complete proof. The cumulative counts follow Section~\ref{sec:Network-Data-Collection}. Each work increases $B_{ik}$ once for worker $i$ and each task type with a distinct substantive output, regardless of the number of outputs of that type. Using or reviewing a published result of another worker increases $A_{i\ell}$ once per work and unordered worker pair. Reads and invitations alone add no counts. These counts describe participation and collaboration. Proof correctness is evaluated separately.

\begin{table}[htbp]
\centering\small\spacingset{1}
\caption{Published results associated with each FATE-H task type.}
\label{tab:fateh-types}
\begin{tabularx}{\textwidth}{@{}lX@{}}
\toprule
Task type & Result \\
\midrule
Retrieve & An applicable Mathlib declaration or published result, with its assumptions and any required adaptation. \\
Derive & A missing proposition, reduction, or proof, including the formal implementation of a mathematical argument. \\
Integrate & Available components assembled into a proof of an original theorem or a concrete proof branch, with an executed check of the resulting statement. \\
Verify & An independent examination of a published result or a bounded repair of a specified defect. Review and repair remain distinct subtypes. \\
\bottomrule
\end{tabularx}
\end{table}

\subsection{Core Prompts and Proof Verification}
\label{app:fateh-prompts-scoring}
\label{app:fateh-prompts}

The boxes below condense the frozen experimental prompts, and Section~\ref{sec:benchmarks-evaluation} describes the three guidelines. Shared instructions apply to all three guidelines. The prompt text describes specialization as the continuation of useful methods across theorems.

\begin{coreprompt}{Shared instructions}
Advance the fixed pool of 100 theorems through retrieval, derivation, integration, and verification. Choose a concrete mathematical gap, inspect the relevant public work, and publish usable results with their assumptions and exact inputs. Continue useful methods across compatible theorems. Change direction when a method fails, an input is unsuitable, work is crowded, or another gap needs attention. All work uses the same session and Lean-check limits.

For independent review, choose a contribution whose recorded lineage excludes you. Read the original theorem and record its requirements and your intended checks before opening the target version. Disclose prior exposure. Review findings concern the stated scope. They do not certify the original theorem. Any worker may submit a complete proof without first obtaining peer input or a PASS review.
\end{coreprompt}

\begin{coreprompt}{History only}
Choose work and collaborators from the shared board, private history, current mathematical needs, available colleagues, and actual contributions. No fitted statistical summary is available.
\end{coreprompt}

\begin{coreprompt}{AgentCue Base}
Before choosing new work or collaborators in each session, read the complete current statistical summary once and retain that version throughout the session. Use the full fitted participation and collaboration means to narrow useful options, then check current needs, availability, and exact public contributions. These means describe cumulative activity, not competence or proof correctness. Record how the summary informed your choice, including when concrete needs or weak support led you to another choice.
\end{coreprompt}

\begin{coreprompt}{AgentCue Pro: choice and feedback}
Apply Base's guidance using the same summary and resources. Interpret the fitted means through their activity and affinity components. First identify the contribution the proof needs. Consider a candidate's participation in that task type, its collaboration mean with you, its availability, and its actual work. No combined score or mandatory maximization is prescribed. Reuse an available proof component without activating its author solely to read it.

Use actual experience and observed results to continue a useful method across problems. Participation proportions describe past choices. They are not allocation quotas. Failed proofs, unusable interfaces, neglected integration, or crowded work can justify changing direction. Repetition alone does not establish expertise.
\end{coreprompt}

\begin{coreprompt}{AgentCue Pro: additional review checks}
Before opening the contribution under review, also record an uncalibrated confidence annotation, one to four scope constraints, and one to three concrete verification questions. After opening it, examine each question against the original theorem, a Mathlib declaration, a checked lemma, or a compiler result, and record the evidence. An existing exact check can settle a question without new compilation.

Use PASS only when all planned questions are resolved, FAIL when concrete evidence contradicts a claim, and UNRESOLVED when uncertainty remains. These verdicts concern the review scope. A complete proof certificate determines the benchmark score, regardless of whether peer PASS is present.
\end{coreprompt}

A theorem receives credit only for a complete proof of its unchanged original statement that was checked and explicitly submitted within the experimental budget and then passed fresh offline verification. This final check uses the saved proof without repair or additional proof search. Intermediate lemmas, proof plans, and review verdicts receive no theorem-level credit. A certified proof does not require a PASS review, and a PASS cannot replace certification. Each theorem contributes at most one success. Unfinished and unattempted theorems remain in the denominator of 100.

\subsection{Additional Results}
\label{app:fateh-outcomes}
\label{app:fateh-results}

Table~\ref{tab:fateh-per-team} gives each replicate's certified-theorem counts after the 64 warm-up sessions and at 256 sessions. Figure~\ref{fig:fateh-progress}(a) partitions the 100 theorems into seven subjects using the first two levels of the benchmark's mathematical tags. Other topics in abstract algebra and commutative algebra remain separate. Certification is calculated within each subject and replicate before averaging across replicates. Figure~\ref{fig:fateh-progress}(b) uses seven checkpoints, reconstructed from the first valid session associated with each certified theorem. All 100 original theorems remain in the denominator throughout.

\begin{table}[htbp]
\centering\small\spacingset{1}\setlength{\tabcolsep}{4.5pt}
\caption{Certified original theorems out of 100 in FATE-H. Warm-up gives the counts after the 64 warm-up sessions, and each branch has 256 valid sessions, including the warm-up sessions. The final row gives the mean and sample standard deviation across five replicates.}
\label{tab:fateh-per-team}
\begin{tabular}{lrrrr}
\toprule
Replicate & Warm-up & AgentCue Pro & AgentCue Base & History only \\
\midrule
R1 & 30 & 81 & 78 & 79 \\
R2 & 40 & 80 & 81 & 76 \\
R3 & 36 & 83 & 75 & 76 \\
R4 & 36 & 83 & 80 & 79 \\
R5 & 36 & 80 & 79 & 81 \\
\midrule
Mean $\pm$ SD & $35.60\pm3.58$ & $81.40\pm1.52$ & $78.60\pm2.30$ & $78.20\pm2.17$ \\
\bottomrule
\end{tabular}
\end{table}

The two process measures in Table~\ref{tab:outcomes} summarize the saved records from the 192 continuation sessions of each replicate. A failed proof is counted as repaired and certified when a published repair explicitly targets an earlier published proof that failed a Lean check, passes a new check of the unchanged original theorem, and supplies that theorem's first submission in publication order to pass fresh verification. The repair, new check, and submission must occur during the continuation sessions. Each original theorem is counted at most once per replicate. The earlier failed proof may come from warm-up. This count concerns explicitly recorded repairs of original proofs. Reads of published results of other workers per session equals the number of successful reads of published results of other workers divided by 192. Repeated reads count separately, and statistical-summary reads are excluded. Both measures are calculated separately for each replicate.\par

The example in Figure~\ref{fig:fateh-progress}(c) follows exact published results, saved checks, and the Choice Record of W28 in the Pro branch of R4. Let $k$ be a finite field with $q$ elements and $L$ an extension of $k$ of degree 19. Because 19 is prime, every element of $L$ outside $k$ has a minimal polynomial over $k$ of degree 19, and W19 publishes a checked proof that exactly $q^{19}-q$ elements of $L$ have minimal polynomials of degree 19. W16 publishes a checked result that every monic irreducible polynomial of degree 19 over $k$ that splits in an extension has 19 distinct roots there, and posts the remaining link as an open sub-task. W28 reads the statistical summary at the checkpoint after 128 valid sessions and the open sub-task, computes the fitted participation means in derive and retrieve, and records the activity and affinity terms in the Choice Record before claiming the sub-task. W28 then reads the two checked results and proves that an element is a root of such a polynomial $p$ if and only if its minimal polynomial is $p$, so each such $p$ is the minimal polynomial of exactly 19 elements. W12 proves that $k$ has an extension of degree 19 with $q^{19}$ elements in which every monic irreducible polynomial of degree 19 splits, and counts the polynomials through the map that sends each element to its minimal polynomial, whose fibers have 19 elements. The proof passes fresh verification and is first credited at 239 valid sessions.

\subsection{Numerical Estimation and Verification}
\label{app:fateh-fit}
\label{app:fateh-recovery}

We use the estimation settings in Section~\ref{app:fit}: $d=2$, 32 workers, four task types, and zero-count pairs retained. Estimation begins after the 64 warm-up sessions and continues every 32 valid sessions from the previous estimate, including work already published by active sessions. We estimate the History-only networks after execution. Each replicate has 13 estimates for the statistical summaries of the warm-up, Base, and Pro sessions and six History-only estimates after execution, giving 95 estimates and 115 optimization starts overall. All selected estimates satisfy the projected-gradient tolerance $10^{-5}$, with maximum infinity norm $9.9943\times10^{-6}$. The numerical audit reconstructs counts and recomputes objectives and fitted means from saved publications.

\section{Additional Experiment Details: DeepSWE10}
\label{app:deepswe}

This section gives the DeepSWE10 tasks, design, and network counts, the guideline prompts and patch evaluation, and the results for each replicate.

\subsection{Benchmark, Experimental Design, and Network Counts}
\label{app:deepswe-design}

DeepSWE10 consists of the 10 DeepSWE~v1.1 tasks that DeLM selects for long-horizon software engineering \citep{huang2026deepswev11,mao2026delm}, listed in Table~\ref{tab:deepswe-tasks}. The tasks come from nine open-source repositories in Rust, TypeScript, Go, and Python. Each task supplies the repository at a fixed base commit, a task instruction, and a container image. The hidden tests and reference solutions remain unavailable to the workers, who can run the public tests of each repository and tests of their own.

Each of the five replicates contains 32 workers that use GPT-5.6 Sol with high reasoning effort through Amazon Bedrock, and the five replicates use the same activation and task-order seeds. After the 64 warm-up sessions, the three branches continue separately from the same shared board and worker conversations for at most 96 continuation sessions each, and a branch ends earlier when no worker has actionable work. Across the five replicates, Base, Pro, and History only complete 437, 411, and 405 continuation sessions. Each session lasts at most 2,400 seconds and has no limit on the number of tool calls. Each worker writes and tests code in a private container, each coding command lasts at most 900 seconds, and at most 16 coding commands run at the same time. A session that fails for technical reasons or ends with a policy refusal of the model provider is invalid and releases its slot. The experiment contains 1,590 sessions, of which 17 are invalid.

Workers declare one of four task types for each work (Table~\ref{tab:deepswe-types}). The cumulative counts follow Section~\ref{sec:Network-Data-Collection}. Each work increases $B_{ik}$ once for worker $i$ and the declared task type $k$ of the work. Using or reviewing a published result of another worker increases $A_{i\ell}$ once per work and unordered worker pair. Reads, messages, and invitations alone add no counts. For Base and Pro, the MAS estimates the statistical summary at each checkpoint with the settings in Section~\ref{app:fit}: $d=2$, 32 workers, and four task types.

\begin{table}[htbp]
\centering\spacingset{1}\footnotesize\setlength{\tabcolsep}{4pt}
\caption{DeepSWE10 tasks and the number of replicates, out of five, in which each guideline resolves the task.}
\label{tab:deepswe-tasks}
\begin{tabular}{@{}llccc@{}}
\toprule
Task & Language & Pro & Base & History only \\
\midrule
\texttt{boa-hierarchical-evaluation-cancellation} & Rust & 5 & 5 & 4 \\
\texttt{dynamodb-toolbox-conditional-attribute-requirements} & TypeScript & 5 & 5 & 3 \\
\texttt{dynamodb-toolbox-lazy-recursive-schemas} & TypeScript & 5 & 4 & 3 \\
\texttt{happy-dom-abort-pending-body-reads} & TypeScript & 4 & 4 & 4 \\
\texttt{numba-stencil-boundary-modes} & Python & 5 & 5 & 5 \\
\texttt{opa-template-string-reconstruction} & Go & 2 & 2 & 2 \\
\texttt{oxvg-structural-selector-preservation} & Rust & 0 & 0 & 0 \\
\texttt{pebble-durability-wait-apis} & Go & 4 & 4 & 3 \\
\texttt{scriggo-method-declarations} & Go & 2 & 3 & 2 \\
\texttt{wasmi-trap-coredumps} & Rust & 3 & 3 & 2 \\
\midrule
Total & & 35 & 35 & 28 \\
\bottomrule
\end{tabular}
\end{table}

\begin{table}[htbp]
\centering\small\spacingset{1}
\caption{Published results associated with each DeepSWE10 task type.}
\label{tab:deepswe-types}
\begin{tabularx}{\textwidth}{@{}lX@{}}
\toprule
Task type & Result \\
\midrule
Construction & A requirement, design, behavior contract, or production component that the task still lacks. A patch for the task instruction itself counts as construction. \\
Repair & A correction of a concrete failure of a published result, with the failing result and public evidence of the failure. \\
Verification & An independent test or review of an exact published result, from one component to a whole candidate patch. \\
Integration & A patch that composes at least two independent published components, with an interface contract and evidence of the composition. \\
\bottomrule
\end{tabularx}
\end{table}

\subsection{Core Prompts and Patch Evaluation}
\label{app:deepswe-prompts}

The boxes below condense the frozen experimental prompts, and Section~\ref{sec:benchmarks-evaluation} describes the three guidelines. Shared instructions apply to all three guidelines.

\begin{coreprompt}{Shared instructions}
Work across the fixed 10 tasks on the shared board as one of 32 equal coding workers. Roles, sequences of task types, and publication quotas do not exist. Own a concrete output, never a whole task. Publish usable results promptly, propose independent components with their exact inputs and interfaces, and invite at most three peers per session when a concrete gap exists. Read or apply the code of another worker only through a published result at its exact version. Declare one task type for each work: construction, repair, verification, or integration.

Any worker may designate a published patch as the candidate patch of a task. The first candidate patch of a task must use a published result of another worker on the same task, and no PASS review is required. For an independent review, record a behavior contract, a test plan, and source evidence before opening the target, then apply the exact patch and publish PASS, FAIL, or UNRESOLVED. A PASS concerns only the reviewed patch and the listed checks. Hidden tests and official solutions are unavailable.
\end{coreprompt}

\begin{coreprompt}{History only}
Choose work and peers from public needs, published evidence, and your own experience. No statistical summary is available.
\end{coreprompt}

\begin{coreprompt}{AgentCue Base}
Each session receives the complete current statistical summary, fixed for the whole session. Read the statistical summary together with public work and your own experience, and use the local tool to compare fitted collaboration means across worker pairs or fitted participation means across workers and task types. These values describe fitted cumulative participation and collaboration counts, not competence, correctness, or assignments. Each invitation states how the statistical summary informed the choice.
\end{coreprompt}

\begin{coreprompt}{AgentCue Pro: choice and feedback}
Apply the guidance of Base with the same statistical summary and resources. Connect a comparison of fitted means to a concrete choice and to your own earlier methods and outcomes. Separate the activity term from the affinity term. Nonzero counts and a strong fitted relationship do not establish competence. Compare plausible alternative task types or peers, consider support and uncertainty, and judge whether complementary work helps the current task. Use the actual results to keep or revise your approach.
\end{coreprompt}

\begin{coreprompt}{AgentCue Pro: additional review checks}
Before opening the patch under review, read the frozen source and record an uncalibrated confidence annotation, one to four scope constraints, one to three falsifiable verification questions, and receipts of the source reads. After opening the patch, apply it, reread the relevant source, and run concrete tests. Resolve each question with a literal quote from the source and, for a resolution by test, with execution evidence.

Use PASS only when every question is resolved and at least one test succeeds after the exact patch is applied, FAIL when recorded evidence contradicts a question, and UNRESOLVED otherwise. These verdicts concern only the reviewed patch.
\end{coreprompt}

After all workers of a replicate exit, the official DeepSWE verifier applies the final candidate patch of each task to a fresh container at the base commit and runs the hidden tests once, without retries or model calls. A task counts as resolved only when the patch passes all hidden tests. The F2P pass rate of a task is the proportion of its fail-to-pass tests that the patch passes. A task without a candidate patch counts as unresolved with an F2P pass rate of zero, and all 10 tasks remain in the denominator. PASS reviews do not enter the evaluation.

\subsection{Additional Results}
\label{app:deepswe-results}

Table~\ref{tab:deepswe-per-replicate} gives the resolved tasks, F2P pass rates, and tasks without a candidate patch of each replicate, and Table~\ref{tab:deepswe-tasks} gives the resolution of each task. Each quantity in Table~\ref{tab:outcomes} is calculated separately for each replicate before averaging across replicates. In total, Pro resolves a task in eight cases where History only does not, History only has no candidate patch in five of these cases, and History only resolves a task in one case where Pro does not.

\begin{table}[htbp]
\centering\spacingset{1}\footnotesize
\caption{DeepSWE10 results for each replicate. Resolved tasks and tasks without a candidate patch are out of 10. The final column gives the mean and sample standard deviation across five replicates.}
\label{tab:deepswe-per-replicate}
\begin{tabular}{@{}lcccccc@{}}
\toprule
Guideline & R1 & R2 & R3 & R4 & R5 & Mean $\pm$ SD \\
\midrule
\multicolumn{7}{@{}l}{\textit{Resolved tasks}}\\
AgentCue Pro & 6 & 7 & 8 & 8 & 6 & $7.00\pm1.00$ \\
AgentCue Base & 6 & 7 & 8 & 8 & 6 & $7.00\pm1.00$ \\
History only & 3 & 6 & 7 & 7 & 5 & $5.60\pm1.67$ \\
\midrule
\multicolumn{7}{@{}l}{\textit{F2P pass rate (\%)}}\\
AgentCue Pro & 70.67 & 87.21 & 88.67 & 95.21 & 95.42 & $87.43\pm10.09$ \\
AgentCue Base & 70.67 & 87.21 & 88.67 & 86.67 & 90.64 & $84.77\pm8.03$ \\
History only & 40.67 & 86.30 & 87.21 & 76.67 & 75.42 & $73.25\pm18.99$ \\
\midrule
\multicolumn{7}{@{}l}{\textit{Tasks without a candidate patch}}\\
AgentCue Pro & 2 & 0 & 0 & 0 & 0 & $0.40\pm0.89$ \\
AgentCue Base & 2 & 0 & 0 & 1 & 0 & $0.60\pm0.89$ \\
History only & 5 & 0 & 0 & 1 & 2 & $1.60\pm2.07$ \\
\bottomrule
\end{tabular}
\end{table}

\endgroup
\end{document}